%% file: main.tex
\documentclass[11pt]{article}
\usepackage[utf8]{inputenc}
\usepackage[english]{babel}
\usepackage{amsfonts,amsmath,amssymb,amsthm} %,mathrsfs,mathabx
\usepackage{graphicx}
\usepackage[dvipsnames]{xcolor}
\usepackage[toc,page]{appendix}
\usepackage{authblk}
\usepackage{tikz}
\usepackage{subcaption}
\usepackage{csquotes}
\usetikzlibrary{decorations.pathreplacing, patterns, hobby}
\usepackage{capt-of, caption}  %for captions in minipages
 \usepackage[colorinlistoftodos,textsize=footnotesize,textwidth=2.5cm]{todonotes}
\usepackage{dsfont}
\usepackage[bookmarksnumbered=true]{hyperref}
\usepackage[capitalise]{cleveref}
\crefname{equation}{}{}

\numberwithin{equation}{section}
\usepackage[style=alphabetic,
  sorting = nyt,
  sortcites=true,
  giveninits=true,
  date=year,
  isbn=false,
  maxbibnames=99,
  maxalphanames=6,
  backend=biber]{biblatex}
\DeclareFieldFormat*{titlecase}{\MakeSentenceCase{#1}}

\DeclareSourcemap{
  \maps[datatype=bibtex]{
    \map[overwrite]{
      \step[fieldsource=doi, final]
      \step[fieldset=url, null]
      \step[fieldset=eprint, null]
    }  
  }
}

\DeclareSourcemap{
  \maps[datatype=bibtex, overwrite]{
    \map{
      \step[fieldset=language, null]
      \step[fieldset=month, null]
      \step[fieldset=urldate, null]
    }
  }
}

\DeclareSourcemap{
  \maps[datatype=bibtex, overwrite]{
    \map{
      \pertype{article}
      \step[fieldset=note, null]
    }
  }
}

\AtEveryBibitem{%
  \clearlist{language}%
  \clearlist{month}%
  \clearlist{urldate}%
  \clearfield{urldate}%
}

\renewbibmacro*{in:}{%
  \setunit{\addcomma\space}%
  \ifentrytype{article}
    {}
    {\printtext{%
       \bibstring{in}\intitlepunct}}}

\renewcommand*{\intitlepunct}{\addspace}

\DeclareFieldFormat[misc]{title}{\mkbibquote{#1}}
\DeclareFieldFormat[article,inproceedings]{volume}{\mkbibbold{#1}}
\DeclareFieldFormat[inproceedings]{series}{\mkbibitalic{#1}}

\title{Mesoscopic Observables of the Dilute Fermi Gas at Low Energy}

\author[1,$*$,$\dagger$]{Niels Benedikter}
\author[2,$*$]{Emanuela L. Giacomelli}
\author[3,$\ddagger$]{Asbjørn~Bækgaard~Lauritsen}
\author[4,$\S$]{Sascha Lill}

\affil[1]{ORCID: \href{https://orcid.org/0000-0002-1071-6091}{0000-0002-1071-6091}, e--mail: \href{mailto:niels.benedikter@unimi.it}{\textnormal{\nolinkurl{niels.benedikter@unimi.it}}}}
\affil[2]{ORCID: \href{https://orcid.org/0000-0002-8147-3250}{0000-0002-8147-3250}, e--mail: \href{mailto:emanuela.giacomelli@unimi.it}{\textnormal{\nolinkurl{emanuela.giacomelli@unimi.it}}}}
\affil[3]{ORCID: \href{https://orcid.org/0000-0003-4476-2288}{0000-0003-4476-2288}, e--mail: \href{mailto:lauritsen@ceremade.dauphine.fr}{\textnormal{\nolinkurl{lauritsen@ceremade.dauphine.fr}}}}
\affil[4]{ORCID: \href{https://orcid.org/0000-0002-9474-9914}{0000-0002-9474-9914}, e--mail: \href{mailto:sali@math.ku.dk}{\textnormal{\nolinkurl{sali@math.ku.dk}}}}

\affil[$*$]{Università degli Studi di Milano, Via Cesare Saldini 50, 20133 Milano, Italy}
\affil[$\dagger$]{External Scientific Member of Basque Center for Applied Mathematics, Alameda de Mazarredo 14, 48009 Bilbao, Bizkaia, Spain}
\affil[$\ddagger$]{CNRS \& CEREMADE, Universit\'e Paris--Dauphine, PSL University, 75016 Paris, France}
\affil[$\S$]{Department of Mathematical Sciences, Universitetsparken 5, 2100 Copenhagen, Denmark}

\newcommand{\cE}{\mathcal{E}}
\newcommand{\cF}{\mathcal{F}}

\newcommand{\cN}{\mathcal{N}}

\newcommand{\CCC}{\mathbb{C}}

\newcommand{\RRR}{\mathbb{R}}

\newcommand{\ZZZ}{\mathbb{Z}}

\newcommand{\Bel}{\textnormal{Bel}}

\renewcommand{\d}{\textnormal{d}}
\newcommand{\di}{\textnormal{d}}

\newcommand{\exc}{{\rm exc}}

\newcommand{\F}{\mathrm{F}}

\newcommand{\I}{\mathrm{I}}

\renewcommand{\r}{\textnormal{r}}

\newcommand{\supp}{\mathrm{supp}}

\newcommand{\norm}[1]{\left\lVert #1 \right\rVert}
\newcommand{\abs}[1]{\left\vert #1 \right\vert}
\newcommand{\kF}{k_\F}

\newcommand{\ud}{\,\textnormal{d}}
\newcommand{\eps}{\varepsilon}
\newcommand{\expect}[1]{\left\langle #1 \right\rangle}

\newtheorem{lemma}{Lemma}[section]
\newtheorem{proposition}[lemma]{Proposition}
\newtheorem{theorem}[lemma]{Theorem}
\newtheorem{corollary}[lemma]{Corollary}

\theoremstyle{definition}
\newtheorem{definition}[lemma]{Definition}
\newtheorem{remark}[lemma]{Remark}

\AddToHook{env/proposition/begin}{\crefalias{lemma}{proposition}}
\AddToHook{env/theorem/begin}{\crefalias{lemma}{theorem}}
\AddToHook{env/corollary/begin}{\crefalias{lemma}{corollary}}
\AddToHook{env/lemma/begin}{\crefalias{lemma}{lemma}}
\AddToHook{env/remark/begin}{\crefalias{lemma}{remark}}
\AddToHook{env/definition/begin}{\crefalias{lemma}{definition}}

\begin{document}
\maketitle
\begin{abstract}
We consider a dilute quantum gas of interacting spin-1/2 fermions in the thermodynamic limit. For a trial state that resolves the ground state energy to the precision of the Huang--Yang formula, we  compute the expectation values of one-body observables on a mesoscopic scale. 
%Our result agrees with the formal perturbative argument of Belyakov (Sov.~Phys.~JETP 13: 850--851 (1961)).
\end{abstract}

\tableofcontents
\section{Introduction}

Dilute quantum gases have attracted a lot of attention in recent years, especially due to advances in the experimental accessibility of ultracold atoms and molecules, and their potential to realize programmable quantum simulators. In this paper, we focus on a dilute gas of fermionic quantum particles. More precisely, we consider $N$ fermions with spin $1/2$ in a box $\Lambda:= [-L/2,L/2]^3$ with periodic boundary conditions. The particles interact via a repulsive pair interaction. The Hamiltonian describing the system is
\begin{equation}\label{eq: hamiltonian}
  H_N := -\sum_{j=1}^N\Delta_{x_j} + \sum_{1\leq i< j\leq N} V(x_i - x_j)
\end{equation}
and acts on the antisymmetric Hilbert space $\bigwedge^N L^2(\Lambda,\mathbb{C}^2)$. Here $V$ is the periodization of a repulsive potential $V_\infty: \mathbb{R}^3 \to \mathbb{R}$. (Detailed assumptions on the potential will be specified later.) We denote by $N_\sigma$ the number of particles with spin $\sigma \in \{\uparrow, \downarrow\}$. Since $H_N$ leaves invariant the subspace $\mathfrak{h}(N_\uparrow, N_\downarrow) \subset \bigwedge^N L^2(\Lambda,\mathbb{C}^2)$ of states with exactly $N_\uparrow$ spin-up and $N_\downarrow$ spin-down particles (where $N = N_\uparrow + N_\downarrow$), we may restrict our analysis to this subspace.
Our analysis is carried out in the thermodynamic limit: we first take $L \to \infty$ while keeping the particle densities $\rho_\sigma := N_\sigma/L^3$ fixed. Then we expand in the dilute regime, where the densities tend to zero in the sense that $a^3 \rho_\uparrow \to 0 $ and $ a^3 \rho_\downarrow \to 0$, with $ a $ denoting the s-wave scattering length of $V$. 

A longstanding open problem concerning the ground state energy density was recently resolved in~\cite{GHNS24,GHNS25} (see also \cite{CWZ25} for the positive temperature setting) by the proof of a conjecture of Huang and Yang from 1957~\cite{HY57}. Denoting by $e(\rho_\uparrow, \rho_\downarrow)$ the ground state energy density, the Huang--Yang formula for the case $\rho_\uparrow=\rho_\downarrow = \rho/2$ reads
\begin{equation}\label{eq: HY formula}
  e\Big(\frac{\rho}{2}, \frac{\rho}{2}\Big) = \frac{3}{5}(3\pi^2)^{\frac{2}{3}}\rho^{\frac{5}{3}} + 2\pi a \rho^2 + \frac{4}{35}(11 -2\log2) (9\pi)^{\frac{2}{3}} a^2 \rho^{7/3} + o(\rho^{\frac{7}{3}}) \quad \text{as }\rho \to 0\;.
\end{equation}
The validity of this formula to order $a\rho^2$ was first proved in \cite{LSS05}
and later reproved with improved accuracy of the error term \cite{FGHP21,Gia23,Gia24} using  bosonization and \cite{Lau25} using cluster expansion. The analysis eventually arrived at the order $a^2\rho^{7/3}$ \cite{GHNS24,GHNS25} and was also extended to positive temperature \cite{CWZ25}.
For the spin-polarized case (i.\,e., $\rho_\downarrow = 0$), the analogue of the order $a\rho^2$ term was proved as an upper \cite{LS24} and a lower \cite{LS24b} bound. For the upper bound, the next term in the expansion is also known, though it should not be thought of as the analogue of the order $a^2\rho^{7/3}$ term in \eqref{eq: HY formula}.

In the present paper, we move beyond the expectation value of the energy and proceed to study generic observables. Naturally, any such analysis starts with the study of one-body observables, and will be restricted to a mesoscopic energy scale, which is in correspondence to a given precision of the energy of the state. More precisely, let $\hat{a}^\ast_{q,\sigma}$ and $\hat{a}_{q,\sigma}$ be the fermionic creation and annihilation operators, respectively, for momentum $q$ and spin $\sigma$ (see \Cref{sec:mathdef} for rigorous definitions). We are going to consider a translation invariant state $\Psi$ in the system's Hilbert space, so any one-body observable can be expressed through the expectation values $\langle \Psi, n_{q,\sigma} \Psi \rangle$ of the operator counting the number of particles with momentum $q$ and spin $\sigma$,
\[
	n_{q, \sigma} := \hat{a}_{q, \sigma}^* \hat{a}_{q, \sigma}\,.
\]
In absence of interactions, these expectation values are simple to compute. In fact, in this case, at fixed $L$, the ground state is the free Fermi gas state
\[
	\psi_{\mathrm{FFG}} := \prod_{\sigma\in \{\uparrow, \downarrow\}}\prod_{k\in\mathcal{B}_{\F}^\sigma} \hat{a}_{k,\sigma}^\ast \Omega\,,
\]
where $\Omega$ is the vacuum vector in fermionic Fock space and
\begin{equation}
\label{eq:fermiball}
	\mathcal{B}_{\F}^\sigma := \Big\{k \in \frac{2\pi}{L}\mathbb{Z}^3\; : \; |k| < k_{\F}^\sigma \Big\}
\end{equation}
is the Fermi ball associated to spin $\sigma \in \{\uparrow, \downarrow\}$. The Fermi momentum for the spin orientation $\sigma \in \{\uparrow, \downarrow\}$ is given by
\[
    k_{\F}^\sigma := (6\pi^2)^{1/3}\rho_\sigma^{1/3} + o_{L \to \infty}(1) \;.
\]
Following \cite{FGHP21, Gia23, Gia24, GHNS24, GHNS25} we restrict our attention to the case of a completely filled Fermi ball, i.\,e., we assume that the particle numbers are such that $N_\sigma = |\mathcal{B}_{\F}^\sigma|$. We then find
\begin{equation}\label{eq: ffg momentum distrubution}
	\langle \psi_{\mathrm{FFG}}, n_{q, \sigma} \psi_{\mathrm{FFG}} \rangle = \begin{cases}
		1 \quad \text{if } q \in  \mathcal{B}_{\F}^\sigma \,,  \\
		0 \quad \text{if } q \notin  \mathcal{B}_{\F}^\sigma \,.
	\end{cases} \;
\end{equation}
However, in the presence of interactions, the ground state (and any state energetically close to the ground state) has a complex entanglement structure where particles are correlated with each other. This makes it tremendously more challenging to compute expectation values---in fact, not only for the ground state, but already for trial states which are constructed to be energetically very close to the ground state. In particular, since the kinetic energy does not have a spectral gap uniform in the system size, already very weak interactions $ V \not= 0$ may change the expectation value \eqref{eq: ffg momentum distrubution} drastically, and even destroy the jump discontinuity typical for Fermi liquids (see \cite{Mig57} for a detailed discussion of the problem).
%In fact it is an open problem to prove that in the system's ground state the number of excitations with respect to \eqref{eq: ffg momentum distrubution} is pointwise much smaller than one, even though this was conjectured already by Landau in 1956 \cite{Lan56}.
For a fermionic gas in a high-density limit, where the interaction is weak and $ \rho \to \infty $ for fixed $ L $ (mean-field scaling limit), the distribution of momenta has been recently computed for two constructions of a trial state approximating the ground state energetically \cite{BL23} and \cite{BLN25}, relying on the bosonization methods of \cite{BNPSS20,BNPSS21,BNPSS21dyn,BPSS22,Ben26} and \cite{CHN21,CHN22,CHN23, CHN24, Chr23PhD}, respectively. The fact that two different constructions of approximate ground states yield essentially the same predictions for the behavior of the momentum distribution supports Landau's proposal \cite{Lan56} that Fermi liquid behavior is universal, at least up to exponentially small corrections resulting from Cooper pairing. The analysis in a completely different parameter regime (dilute instead of high density) would further substantiate this conjecture.
%While any statement about the momentum distribution of the true ground state still remains elusive, we believe that a further explicit example of a low-energy trial state in a different physical regime supports Landau's conjecture that the presence of a Fermi surface (up to much smaller corrections due to Kohn--Luttinger superconductivity) is universal in two and more spatial dimensions.
 This is a further reason motivating us to analyze the expectation values of mesoscopic one-body observables in a the trial state $\Psi \in \mathfrak{h}(N_\uparrow, N_\downarrow)$ employed in~\cite{GHNS24} to establish the optimal upper bound on the expansion~\eqref{eq: HY formula}.
A complete understanding of observables in the true ground state would require a significantly deeper analysis of the system due to the vanishing spectral gap, but we are convinced that our method of extracting the leading order through a conjugation with unitary transformations simplifies this problem, reducing it to showing that the number of further excitations created by the conjugated Hamiltonian is close to zero. We aim to return to this problem in the future.

Our main result provides the thermodynamic limit of the expectation value of one-body observables on a mesoscopic momentum scale in  the trial state constructed by \cite{GHNS24}. In particular, we obtain, via three explicit unitary transformations, Belyakov’s corrected prediction for momenta of the order of the Fermi radius.

\paragraph{Notation.} We denote by $C$ a constant independent of $L$ and $\rho$, whose value may change from line to line. For the Fourier transform of functions $f:\Lambda\rightarrow\mathbb{C}$ we use the convention
\[
	\hat{f}(p) = \int_\Lambda \di x\, f(x)e^{-ip\cdot x} \;, \qquad f(x) = \frac{1}{L^3}\sum_{p\in\Lambda^\ast}\hat{f}(p)e^{ip\cdot x} \;,
\]
with $\Lambda^* := \frac{2\pi}{L}\ZZZ^3$.
For $g:\mathbb{R}^3\rightarrow \mathbb{C}$ instead we denote its Fourier transform by $\mathcal{F}(g)$, using the differing normalization convention
\[
 	\mathcal{F}(g)(p) = \frac{1}{(2\pi)^3}\int_{\mathbb{R}^3} \di x\, g(x) e^{-ip\cdot x} \;.
\]
Moreover, for $f\in L^2(\Lambda, \mathbb{C}^2)$, we write $\hat{f}(\cdot, \sigma) = \hat{f}_\sigma(\cdot)$ and $f(\cdot, \sigma) = f_\sigma(\cdot)$. Often we write $\int$ and $\sum_p$ in place of $\int_\Lambda$ and $\sum_{p\in\Lambda^\ast}$, respectively. Furthermore, we define $f_x(\cdot) := f(\cdot\, - x)$.
 Finally, we write $o_{L\to\infty}(1)$ for any term which vanishes in the limit $L\to \infty$.

\section{Main Result}
\label{sec:mainresult}

Since we consider a translation invariant state, the expectation value of any one-particle observable (in second quantization written as $\di\Gamma(A)$, where $A$ is an operator on the one-particle Hilbert space) is diagonal with respect to momenta; that is, it maybe be computed by integrating a function of momentum against the excitation density operator
\begin{equation} \label{eq:nqexc}
	n^{(\exc)}_{q,\sigma} := \begin{cases} 1- n_{q,\sigma} &\text{if } q \in \mathcal{B}_{\F}^\sigma\,, \\ n_{q,\sigma} &\text{if } q\notin\mathcal{B}_{\F}^\sigma\,,  \end{cases}
\end{equation}
which is unitarily equivalent to $n_{q,\sigma}$, see \Cref{eq:nexc_R_trafo}. In the ground state of the non-interacting system $\psi_\textnormal{FFG}$, as a consequence of \eqref{eq: ffg momentum distrubution}, we do not have any excitations, so
\[
   \langle \psi_\text{FFG}, n^{(\exc)}_{q,\sigma} \psi_\text{FFG} \rangle = 0\,.
\]
Due to the Pauli principle we have $0 \leq n_{q,\sigma} \leq 1$, and therefore $0 \leq n^{(\exc)}_{q,\sigma} \leq 1$. In the thermodynamic limit, we have to restrict attention to a version of the excitation density which is averaged over a mesoscopic momentum scale. For this purpose, we define the averaging function $\hat{g}$ as the restriction to $\Lambda^*$ of some function $\mathcal{F}(g) \in L^\infty (\mathbb{R}^3)\cap L^2(\mathbb{R}^3)$. The mesoscopic observables are then
\begin{equation}\label{eq:ngpro}
	n_{g,\sigma}^{(\exc)}
    := \frac{(2\pi)^3}{L^3} \sum_{p\in\Lambda^\ast }\hat{g}(p)n_{p,\sigma}^{(\exc)}\;, \qquad
    n_g^{(\exc)}
    := \sum_{\sigma\in\{\uparrow, \downarrow\}} n_{g,\sigma}^{(\exc)} \;,
\end{equation}
and the predictions for its ground state expectations~\cite{Bel61} are
\begin{equation} \label{eq:ngBelyakov_2}
\begin{aligned}
    n^{(\Bel)}_{g,\sigma}
	&:= \frac{a^2}{\pi^4} \int_{\RRR^3} \d p
		\int_{|r| < k_{\F}^\sigma < |r+p|} \hspace{-0.3cm}\d r
		\int_{|r'| < k_{\F}^{\sigma^\perp} < |r'-p|}  \hspace{-0.3cm}\d r'\,
		\frac{{\mathcal{F}(g)}(r+p) + {\mathcal{F}(g)}(-r)}{(|r+p|^2 - |r|^2 + |r'-p|^2 - |r'|)^2} \;, \\
    n^{(\Bel)}_{g}
    &:= \sum_{\sigma\in\{\uparrow, \downarrow\}} n^{(\Bel)}_{g,\sigma} \;,
\end{aligned}
\end{equation}
with $\sigma^\perp =~\downarrow$ if $\sigma =~\uparrow$ and $\sigma^\perp =~\uparrow$ if $\sigma =~\downarrow$.
We can now state our main result.
\begin{theorem}[Main Result] \label{thm:main}
Let $ V_\infty \in L^2(\RRR^3) $ be non-negative, radial and compactly supported. Then, there exists a trial state $\Psi\in \mathfrak{h}(N_\uparrow, N_\downarrow)$ such that
\begin{enumerate}
\item it has energy density close to the ground state energy density in the sense that 
\begin{equation} \label{eq:main_energy}
	\limsup_{L \to \infty} \frac{1}{L^3}\left|\langle \Psi, H_N \Psi\rangle - E_{L}(N_\uparrow, N_\downarrow) \right|
    \leq C (\rho^{\frac 73 + \frac{1}{120}}) \qquad \text{as }\rho \to 0\;;
\end{equation}
\item and given $ \kappa > 0 $, there exist $ C_\kappa, C > 0 $ such that as $ \rho_\uparrow + \rho_\downarrow = \rho \to 0 $, we have
\begin{equation} \label{eq:main}
	\limsup_{L \to \infty} \abs{\langle \Psi, n_g^{(\exc)} \Psi\rangle - n^{(\Bel)}_{g}}
	\leq C \|\hat{g}\|_\infty \rho^{\frac{5}{3} + \frac{1}{9}}
        + C_\kappa\rho^{2-\kappa} {\|\hat g\|_2} \;.
\end{equation}
\end{enumerate}
Moreover, if $ \hat{g} \ge 0 $, then there exists $C,c > 0$ such that 
\begin{equation} \label{eq:ngBelyakovbound}
	c {\rho^{\frac 53}_\sigma}  \int_{{\mathbb{R}^3}} \frac{{\mathcal{F}(g)}(\kF^\sigma p)}{1+|p|^4} \di p
    \leq
    n^{(\Bel)}_{g,\sigma}
	\leq C {\rho^{\frac 53}_\sigma} \int_{{\mathbb{R}^3}}  \frac{{\mathcal{F}(g)}(\kF^\sigma p)}{1+|p|^4} \di p
    \;.
\end{equation}
\end{theorem}
\begin{remark}[Energy Precision]
We use a trial state whose energy density approximates the ground state energy density with an error that is subleading relative to the Huang-Yang correction, as is evident from \Cref{eq:main_energy} and~\Cref{eq: HY formula}.
\end{remark}
The proof of \cref{thm:main} is given in \cref{sec.proof.prop.main}, where we finalize it in \cref{subsec:proofmain}. As trial state $ \Psi $, we use the one of~\cite{GHNS24}, defined in \Cref{subsec:trialstate}. The energy statement \Cref{eq:main_energy} follows from~\cite{GHNS24,GHNS25}, while our novel contributions are \Cref{eq:main} and \Cref{eq:ngBelyakovbound}.

An interesting special case for $n_g^{(\exc)}$ is the excitation density averaged over a small ball in momentum space: Let $q\in (2\pi / L)\mathbb{Z}^3$ and $0< \alpha < 1/27$ (see Remark \ref{rem: alpha} for a discussion of the choice of the parameter $\alpha$) and let $\chi_{q,\alpha,\sigma}$ be the periodic function on $\Lambda$ given by the Fourier coefficients $\widehat{\chi}_{q,\alpha,\sigma}(k) = \mathcal{F}(\chi_{q,\alpha,\sigma})(k)$, $k\in \Lambda^*$, where 
\begin{equation} \label{eq:chi_q_alpha}
    \mathcal{F}({\chi}_{q,\alpha,\sigma})(k) := 
    \begin{cases} 1 &\text{if } k \in \RRR^3 \text{ satisfies } |k-q| \leq \rho^{\frac{1}{3} + \alpha}_\sigma\,, \\ 0&\text{for all other } k \in \RRR^3\,.\end{cases}
    \qquad 
\end{equation}
This corresponds to a ball in momentum space centered at $q$; as $\rho \to 0$, its radius $\rho^{\frac 13 + \alpha}$ is much smaller than the radius of the Fermi ball $ k_{\F}^\sigma \sim \rho^{1/3} $. 
The observables corresponding to the choices $\hat g = \widehat{\chi}_{q,\alpha,\sigma}$ and $\hat g = \widehat{\chi}_{q,\alpha,\uparrow} + \widehat{\chi}_{q,\alpha,\downarrow}$ are
%\footnote{The factor $(2\pi)^3/L^3$ on the right-hand side of \eqref{eq:nexcg} ensures that, in the limit $L\rightarrow \infty$, the quantity $n_{q,\alpha}^{\mathrm{exc}}$ converges to an integral over $\mathbb{R}^3$ with respect to $\d p$ without any additional prefactor.}
\begin{equation}\label{eq:nexcg}
    n^{(\exc)}_{q,\alpha,\sigma}
    := \frac{(2\pi)^3}{L^3}\sum_{p\in\frac{2\pi}{L}\mathbb{Z}^3}\widehat{\chi}_{q,\alpha,\sigma}(p)\, n^{(\exc)}_{p,\sigma} \;, \qquad
    n^{(\exc)}_{q,\alpha}
    := \sum_{\sigma\in\{\uparrow, \downarrow\}} n^{(\exc)}_{q,\alpha,\sigma} \;.
\end{equation}
Belyakov's \cite{Bel61} formula for the expectation value, averaged with the same $\hat g$, takes the form
\begin{equation} \label{eq:ngBelyakov}
\begin{aligned}
	n^{(\Bel)}_{q,\alpha, \sigma}
	&:= \frac{a^2}{\pi^4} \int_{\RRR^3} \! \d p
		\int_{|r| < k_{\F}^\sigma < |r+p|} \! \d r
		\int_{|r'| < k_{\F}^{\sigma^\perp} < |r'-p|} \! \d r'
		\frac{\mathcal{F}({\chi}_{q,\alpha,\sigma})(r+p) + {\mathcal{F}({\chi}_{q,\alpha,\sigma})}(-r)}{(|r+p|^2 - |r|^2 + |r'-p|^2 - |r'|^2)^2} \;, \\
    n^{(\Bel)}_{q,\alpha}
    &:= \sum_{\sigma\in\{\uparrow, \downarrow\}} n^{(\Bel)}_{q,\alpha, \sigma} \;,
\end{aligned}
\end{equation}
where we recall $\sigma^\perp =~\downarrow$ if $\sigma =~\uparrow$ and $\sigma^\perp =~\uparrow$ if $\sigma =~\downarrow$. Our main result for this averaging is as follows.

\begin{corollary} \label{cor:main}
Let $ V_\infty \in L^2(\RRR^3) $ be non-negative, radial and compactly supported. Then, given any $q\in (2\pi /L)\mathbb{Z}^3$ and $0< \alpha< 1/27$, for the trial state $\Psi\in \mathfrak{h}(N_\uparrow, N_\downarrow)$ from Theorem~\ref{thm:main} we have
\begin{equation} \label{eq:main_cor_1}
	\limsup_{L\to\infty}\abs{
        \langle \Psi, n^{(\exc)}_{q,\alpha} \Psi\rangle
        - n^{(\Bel)}_{q,\alpha}}
	\leq C \rho^{\frac 53 + \frac 19}  \qquad \text{as }\rho \to 0\;,
\end{equation}
%where $n_{q,\alpha}^{\mathrm{exc}}$ and $n^{(\Bel)}_{q,\alpha}$ were defined in \eqref{eq:nexcg} and \eqref{eq:ngBelyakov}, respectively.
Furthermore, there exist constants $c, C > 0$ such that for all $|q| \leq C \rho^{\frac 13}_\sigma$, we have
\begin{equation} \label{eq:main_cor_leadingorder}
	c {\rho^{\frac 53 + 3 \alpha}_\sigma}
	\leq n^{(\Bel)}_{q,\alpha, {\sigma}}
	\leq C {\rho^{\frac 53 + 3 \alpha}_\sigma} \;.
\end{equation} 
\end{corollary}
\begin{remark}[Mesoscopic Scale]\label{rem: alpha} According to \eqref{eq:main_cor_1} and \eqref{eq:main_cor_leadingorder}, we are interested in $0< \alpha < 1/27$ for the error term to be subleading with respect to $n^{\mathrm{(Bel)}}_{q,\alpha}$. 
\end{remark}

In the physics literature, Czy\.z and Gottfried \cite{CG61} computed expectation values of $n_{q,\sigma}$, along with a numerical evaluation at selected values of the momentum $ |q| $. Belyakov~\cite{Bel61} subsequently proposed explicit formulas; however, both expressions in the regime  $k_{\F}^\sigma < |q| < \sqrt{3} k_{\F}^\sigma$  were evaluated incorrectly in~\cite{Bel61}. Sartor and Mahaux re-analyzed Belyakov’s integral~\cite{SM80} and eventually \cite{SM82} corrected the error. Figure~\ref{fig:Sartor_Mahaux_plot} shows a plot of their corrected prediction.
\begin{figure}[h]
    \centering
    \includegraphics[width=0.42\linewidth]{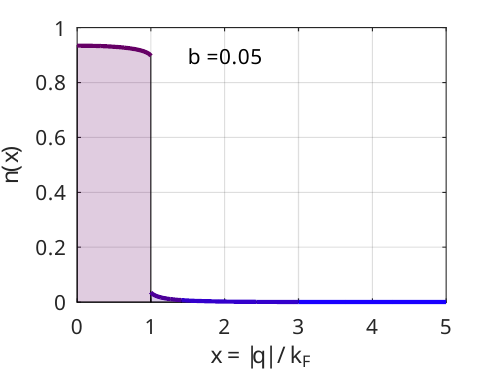}
    \hspace{2em}
    \includegraphics[width=0.42\linewidth]{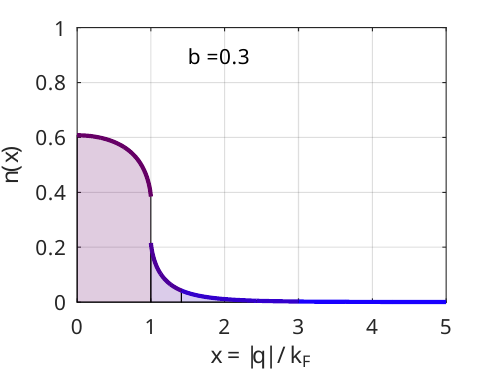}
    \caption{Plot of the expecation value $n(x) := \langle a_{q,\sigma}^* a_{q, \sigma} \rangle$ rescaled according to $x := |q|/k_{\F}$, as predicted by~\cite{SM80,SM82} for equal spin populations $\rho_\uparrow = \rho_\downarrow = \rho/2$ and for two values of the dimensionless density parameter $b := (3 \rho a^3 /\pi)^{2/3}$.}
    % Notice that there are four different formulas for $x$ in the intervals $(0,1)$, $(1,\sqrt{2})$, $(\sqrt{2},3)$, and $(3,\infty)$.
    \label{fig:Sartor_Mahaux_plot}
\end{figure}
%

%Let us add a few remarks before proving \Cref{cor:main} from \Cref{thm:main}.

%\begin{remark}[Microscopic scale]\label{rem: averaging} 
%One might attempt to reach the microscopic scale by introducing $\widehat{\chi}_{q,L}(k) = \mathcal{F}(\chi_{q,L})(k)= \chi (|k - q| < \pi/ L) L^3$ and replacing $\widehat{\chi}_{q,\alpha}(\cdot)$ in \eqref{eq:ngBelyakov} with $\widehat{\chi}_{q,L}$. 
%This would yield Belyakov's original prediction for $ \langle \Psi, n^{\mathrm{exc}}_{q,\sigma} \Psi \rangle $. However, the leading-order contribution $ n^{(\Bel)}_{q,\alpha,\sigma} $ is then independent of $ L $, while the error would scale as $L^3 $ (see also Remark \ref{rem: key pro}). 
%However, then the error would become much larger than the leading term as $ L \to \infty $.
%\end{remark}
\begin{remark}[Fermi surface]\label{rem: jump} 
%\textcolor{blue}{Although Belyakov's formula predicts a discontinuity of $n_{q,\sigma}$ at the Fermi surface, 
%Our averaged result does not prove the existence of a discontinuity of $n_{q,\sigma}$ at the Fermi surface:
%It is always possible to smoothen out $n_{q,\sigma}$ at an arbitrarily small scale, causing an arbitrarily small change of $n_{q,\alpha, \textcolor{magenta}{\sigma}}^{\mathrm{exc}}$. Proving a discontinuity of $n_{q,\sigma}$ would require a removal of the averaging, which is very difficult, see the previous remark.}

Our result implies a lower bound on the slope of the momentum distribution $\langle \Psi, n_{q,\sigma} \Psi \rangle$ at the Fermi momentum, in the sense that there exist two points $|q_1| > k_{\F}^\sigma$ and $|q_2| < k_{\F}^\sigma$ such that $\frac{|\langle \Psi, (n_{q_1,\sigma}-n_{q_2,\sigma}) \Psi \rangle|}{|q_1-q_2|} > c \rho^{-\frac 13 - \alpha}$. This is a remnant of the jump discontinuity in the momentum distribution of the non-interacting Fermi gas.
\end{remark}
%\begin{remark}[Validity of Belyakov's formula for $|q|\leq C\rho^{1/3}_\sigma$] Note that due to the choice of the localization function (see \eqref{eq:chi_q_alpha}), the bounds in \eqref{eq:main_cor_leadingorder} are meaningful only for momenta $q$ such that $|q|\leq C\rho^{1/3}_\sigma$.
%\end{remark}
\begin{remark}
\Cref{cor:main} shows that for a particular trial state (the one of \cite{GHNS24}), the expectation value agrees with that conjectured in the physics literature for the ground state, for momenta satisfying  $|q|\leq C\rho^{1/3}_\sigma$ and resolved on the mesoscopic scale.
The simpler trial states used in \cite{LSS05,FGHP21,Gia23,Lau25} cannot be expected to give the same result. In fact, Belyakov's formula for the excitation density arises from second order perturbation theory  in the state \cite{Bel61}. 
Second order perturbation theory in the state corresponds to third order perturbation theory in the energy. Thus, we expect Belyakov's formula to hold for trial states resolving the ground state energy to third order perturbation theory. 
The trial states of \cite{LSS05,FGHP21,Gia23,Lau25} resolve the ground state energy to order $a\rho^2$. This term arises from second order perturbation theory (see also \cite{LS24b} for the analogous term in the spin-polarized setting). Hence, one should not expect the trial states of \cite{LSS05,FGHP21,Gia23,Lau25} to reproduce the Belyakov term. More specifically yet, from a calculation similar to the present paper one sees that the trial states of \cite{FGHP21,Gia23} yield the expectation value at the same order as the Belyakov term, but with the wrong constant. For the Jastrow-type trial states of \cite{LSS05,Lau25}, the formal expansion of \cite[Eq.~(4.9)]{GGR71} also suggests an excitation density of the same order as the Belyakov term, but with the wrong constant. 
\end{remark}

\subsection{Proof of Corollary~\ref{cor:main}}
\label{sec:key_proposition_and_proof_of_theorem_ref_thm_main}

\begin{proof}[Proof of \Cref{cor:main}]
The validity of Beliakov's formula~\eqref{eq:main_cor_1} follows by plugging $ \hat{g}(p) = {\widehat{\chi}_{q,\alpha,\sigma}(p)} $ into~\eqref{eq:main}, where (compare~\eqref{eq:chi_q_alpha}) $\norm{\hat{g}}_\infty \le C$ and $\norm{\hat{g}}_2 \le C \rho^{{ \frac 12 + \frac 32 \alpha}}$.
The bound on $n^{(\Bel)}_{q,\alpha}$ follows from~\eqref{eq:ngBelyakovbound}
by choosing $ \mathcal{F}(g)(\cdot) = \mathcal{F}({\chi}_{q,\alpha,\sigma})(\cdot)$.
Indeed 
\[
    \int_{\mathbb{R}^3} \!\d p \frac{\mathcal{F}(\chi_{q,\alpha,\sigma})(\kF^\sigma  p)}{1+|p|^4} = \int_{|\kF^\sigma p - q| \leq \rho^{\frac{1}{3} + \alpha}_\sigma}  \frac{\d p}{1+|p|^4} = \frac{1}{(\kF^{\sigma})^3} \int_{|s-q|< \rho^{\frac{1}{3} + \alpha}_\sigma}  \frac{\d s}{1+|(\kF^\sigma)^{-1}s|^4} \;.
\]
If $q = \mathcal{O}(\rho^{1/3}_\sigma)$, a further change of variables yields
\[
    \frac{1}{(\kF^{\sigma})^3} \int_{|s-q|< \rho^{\frac{1}{3} + \alpha}_\sigma} \d s \frac{1}{1+|(\kF^\sigma)^{-1}s|^4}\leq C \frac{\rho^{1+3\alpha}_\sigma}{(\kF^\sigma)^3} \leq C \rho^{3\alpha}_\sigma\;.
\]
Moreover, since $q = \mathcal{O}(\rho^{1/3}_\sigma)$ and $|s-q| \leq \rho^{1/3 + \alpha}_\sigma$, we have that $|s|\leq C\rho^{1/3}_\sigma$ and thus
\[
    \frac{1}{(\kF^{\sigma})^3} \int_{|s-q|< \rho^{\frac{1}{3} + \alpha}_\sigma} \d s \frac{1}{1+|(\kF^\sigma)^{-1}s|^4}\geq C \frac{\rho^{1+3\alpha}_\sigma}{(\kF^\sigma)^3} = C \rho^{3\alpha}_\sigma \;.
\]
Therefore
\[
     C \rho^{3\alpha}_\sigma\leq \int_{\mathbb{R}^3} \d p \frac{\mathcal{F}(\chi_{q,\alpha,\sigma})(\kF^\sigma  p)}{1+|p|^4} \leq C \rho^{3\alpha}_\sigma\qquad \mbox{for}\,\,\, |q| \leq C\rho^{\frac{1}{3}}_\sigma \;.
\]
Using this bound in \eqref{eq:main} yields \eqref{eq:main_cor_1}.
\end{proof}

\section{Notation and Trial State}
\label{sec:mathdef}
Our remaining task is to prove \cref{thm:main}. We use second quantization to construct a trial state through unitary operators defined as exponentials of creation and annihilation operators. To compute expectation values, these exponentials will be expanded to second order and the error terms estimated in position space representation.

\subsection{Second Quantization}
We briefly review second quantization. The fermionic Fock space is
\begin{equation}
	\cF_{\mathrm{f}}
	:= \bigoplus_{N = 0}^\infty L^2(\Lambda ; \CCC^2)^{\wedge N} \;.
\end{equation}
The vacuum vector in Fock space is denoted by $\Omega = (1, 0, 0, 0,  \ldots)$.
We denote by $ a^*(f)$ and $a(f) $ the standard fermionic creation and annihilation operators acting on Fock space and satisfying the canonical anticommutation relations (CAR), i.\,e., for all $f, g \in L^2(\Lambda ; \CCC^2)$ we have
\begin{equation}
\label{eq:CAR}
	\{ a(f), a^*(g)\}
	= \langle f, g \rangle \;, \qquad
	\{ a(f), a(g) \}
	= 0
	= \{ a^*(f), a^*(g) \}
	\;.
\end{equation}
Since we consider fermions, we have the operator norm bounds
\begin{equation}\label{eq: Pauli}
\Vert a^*(f) \Vert = \Vert f \Vert_2 \quad \text{and} \quad \Vert a(f) \Vert = \Vert f \Vert_2 \;.
\end{equation}
 %In particular these operators are defined on all of $\cF_{\mathrm{f}}$.
 In the following, we use the notation $a_\sigma(f) := a(\delta_\sigma f)$.

 Let $\sigma \in \{\uparrow, \downarrow\}$ and $k\in \Lambda^\ast$. Writing $(\delta_\sigma f_k)(x,\sigma^\prime) = \delta_{\sigma, \sigma^\prime}L^{-3/2}e^{ik\cdot x}$, we define the corresponding annihilation operators
\begin{equation}
\label{eq:aadagger}
	\hat{a}_{k, \sigma}
    := a(\delta_\sigma f_{k})
    := \frac{1}{L^{3/2}}\int_\Lambda \di x\, a_{x,\sigma}e^{-ik\cdot x}\,,
\end{equation}
where we introduced the operator-valued distributions $a_{x,\sigma} := a(\delta_{x,\sigma})$ with $\delta_{x,\sigma}(y,\sigma^\prime) = \delta_{\sigma, \sigma^\prime}\delta(x-y)$.
The number operators are
\begin{equation}
\label{eq:cN}
	\mathcal{N}
    := \sum_{\sigma\in \{\uparrow, \downarrow\}}\mathcal{N}_\sigma\,, \qquad
    \cN_\sigma
    := \sum_{k \in \Lambda^\ast} \hat{a}_{k, \sigma}^* \hat{a}_{k, \sigma}\,.
\end{equation}

\subsection{Construction of the Trial State}
\label{subsec:trialstate}

In this section we construct the trial state used in \cref{thm:main}, which is the one used in~\cite{GHNS24} to derive the ground state energy density up to Huang--Yang precision. The construction applies three unitary transformations $R$, $T_1$, and $T_2$ to the vacuum vector $\Omega \in \cF_{\mathrm{f}}$, so the final trial state is
\begin{equation*}
    \Psi = R T_1 T_2 \Omega \;.
\end{equation*}
\paragraph{The particle--hole transformation $R$.} The unitary operator $R$ can be used to represent the non-interacting ground state. It is defined via two orthogonal projections $u,v:L^2(\Lambda, \mathbb{C}^2)\rightarrow L^2(\Lambda, \mathbb{C}^2)$ with integral kernels
\begin{equation}\label{eq: def uv}
	u_{\sigma,\sigma^\prime}(x;y)
    := \frac{\delta_{\sigma,\sigma^\prime}}{L^3}\sum_{k\notin \mathcal{B}_{\F}^\sigma}e^{ik\cdot (x-y)}\,,\qquad
    v_{\sigma,\sigma^\prime}(x;y)
    := \frac{\delta_{\sigma,\sigma^\prime}}{L^3}\sum_{k\in \mathcal{B}_{\F}^\sigma}e^{ik\cdot (x-y)}\,.
	\end{equation}
These projections satisfy $uv = vu = 0$, and $v$ is the one-particle reduced density matrix of $\psi_\textnormal{FFG}$, whereas $u = 1 - v$. For simplicity we write $u_\sigma := u_{\sigma,\sigma}$ and  $v_\sigma := v_{\sigma,\sigma}$, and introduce the functions $u_{x,\sigma}(y) := u_\sigma(y;x)$ and $v_{x,\sigma}(y) := v_\sigma(y;x)$.
In Fourier space, $u$ and $v$ can be represented as multiplications by
\begin{equation}\label{eq: def u hat, v hat}
    \hat{u}_\sigma(k)
    := \begin{cases}
      0 &\text{if } |k| \leq k_{\F}^\sigma\,, \\ 1 &\text{if } |k| > k_{\F}^\sigma\,,
  \end{cases}\qquad
    \hat{v}_\sigma(k)
    := \begin{cases} 1 &\text{if } |k| \leq k_{\F}^\sigma\,, \\ 0 &\text{if } |k| > k_{\F}^\sigma\,. \end{cases}
\end{equation}
We will view $u_\sigma$ and $v_\sigma$ also as functions with these Fourier coefficients. Since the integrals kernels $u_\sigma(x;y)$ and $v_\sigma(x;y)$ depend only on $x-y$, the notation $u_{x,\sigma}$ and $v_{x,\sigma}$ is consistent with the notation $f_x(y) = f(y-x)$ for a function $f$ introduced earlier.
\begin{definition}[Particle--Hole Transformation]\label{def: particle hole} For $u$ and $v$ given by \eqref{eq: def uv}, the particle--hole transformation $R: \mathcal{F}_{\mathrm{f}}\rightarrow \mathcal{F}_{\mathrm{f}}$ is the unitary operator defined by
\begin{equation}\label{eq: R*aR}
    R\Omega := \psi_{\textnormal{FFG}} \;, \qquad
	R^\ast a^\ast_{x,\sigma} R
    := a^\ast_\sigma(u_{x,\sigma}) + a_\sigma(v_{x,\sigma})\,.
\end{equation}
\end{definition}
Note that \eqref{eq: R*aR} is equivalent to
\[
	R^\ast \hat{a}_{k,\sigma} R = \hat{a}_{k,\sigma} \text{ if } k\notin\mathcal{B}_{\F}^\sigma
	\quad\text{and}\quad
	R^\ast \hat{a}_{k,\sigma} R = \hat{a}^*_{-k,\sigma} \text{ if } k\in\mathcal{B}_{\F}^\sigma\;.
\]
Moreover the reader may verify that $R = R^*$.
\paragraph{Quasi-bosonic Bogoliubov transformations $T_1$ and $T_2$.} The unitary operators $T_1$ and $T_2$ are defined in terms of quasi-bosonic operators, which are defined as pairs of fermionic annihilation or creation operators. Specifically, for $ \sigma \in \{\uparrow, \downarrow\} $ and $p,r \in \Lambda^*$,
\begin{equation}\label{eq: def bp bpr}
b_{p,r,\sigma}
:= \hat{u}_\sigma(p+r)\hat{v}_\sigma(r)\hat{a}_{p+r,\sigma}\hat{a}_{-r,\sigma}\;, \qquad
b_{p,\sigma}
:= \sum_{r\in\Lambda^*}b_{p,r,\sigma}\,,
\end{equation}
with adjoints $ b_{p,r,\sigma}^\ast := (b_{p,r,\sigma})^\ast $ and $ b_{p,\sigma}^\ast := (b_{p,\sigma})^\ast $.
Both $T_1$ and $T_2$ are quadratic exponentials in these quasi-bosonic operators. However, $T_1$ governs the high-momentum behavior and uses a less refined kernel, while $T_2$ is tailored to low momenta and employs a more refined kernel. The kernels will be defined using the scattering solutions.
\begin{definition}[Scattering Solutions]\label{def:scattering-phi-eta} We  define the periodic function $\varphi:\Lambda\to \mathbb{C}$ by
$\hat \varphi(p)=\mathcal{F}(\varphi_\infty)(p)$ for $p\in \Lambda^* \setminus \{0\}$, namely
\begin{equation}\label{eq: def phi}
  \varphi(x) := \frac{1}{L^3}\sum_{p\in\Lambda^*\setminus\{0\}}\mathcal{F}(\varphi_\infty)(p) e^{ip\cdot x}\,,
\end{equation}
where $\varphi_\infty:\mathbb{R}^3\to \mathbb{R}$ is the solution of the zero-energy scattering equation
\begin{equation}
\label{eq:scatteringequations}
		2 \Delta \varphi + V_\infty (1 - \varphi)  = 0 \;, \qquad \varphi(x) \to 0 \quad \text{as } |x| \to \infty\,.
\end{equation}
Moreover, for $\varepsilon := \rho^{2/3 +\delta}$, with an arbitrarily  large parameter\footnote{The constraint on $\delta$ is relaxed compared to \cite[Prop.~4.1, Prop.~4.9, and Sec.~5]{GHNS24}. We require only $\delta > 2/9$ instead of $16/63 < \delta < 8/9$, as a consequence of the improved estimate of the number operator; compare Lemma~\ref{lem: est number op} to \cite[Proposition 3.5]{GHNS24}.
%, while still guaranteeing the same approximation for the ground state energy.
}
${\delta > 2/9}$, we define
\begin{equation}\label{eq: def eta epsilon}
  \hat{\eta}_{r,r^\prime}^\varepsilon (p) := \frac{8\pi a}{ \lambda_{r,p} +  \lambda_{r',-p} + 2\varepsilon}\,,\qquad \lambda_{r,p} := \left\lvert \lvert r+p\rvert^2 - \lvert r\rvert^2 \right\rvert \,.
\end{equation}
\end{definition}
The functions $\hat{\varphi}$ and $\hat{\eta}^\varepsilon_{r,r'}$ describe scattering of excitations in the high- and low-momentum regimes, respectively.
To distinguish the two regimes, we follow~\cite{GHNS24} and introduce smooth cutoff functions $\mathcal{F}(\chi_{<,\infty})$ and $\mathcal{F}(\chi_{>,\infty})$ satisfying
\begin{equation}\label{eq: def chi< chi>}
\mathcal{F}({\chi}_{<,\infty})(p) := \begin{cases}
                1 &\mbox{if}\,\,\,|p| < 4\rho^{\frac{2}{9}}\,, \\
                0 &\mbox{if}\,\,\, |p| > 5 \rho^{\frac{2}{9}}\,,
              \end{cases}
\qquad
\mathcal{F}({\chi}_{>,\infty})(p) := \begin{cases}
                0 &\mbox{if}\,\,\, |p| < 4\rho^{\frac{2}{9}}\,,\\
                1 &\mbox{if}\,\,\, |p| > 5\rho^{\frac{2}{9}}\,.
              \end{cases}
\end{equation}
We write $\widehat{\chi}_<$ and $\widehat{\chi}_>$ for the finite volume analogues of $\mathcal{F}(\chi_{<,\infty})$ and $\mathcal{F}(\chi_{>,\infty})$.

\begin{definition}[Quasi-Bosonic Transformations $T_1$ and $T_2$]\label{def:bosonic-transformations} Let $\lambda\in \mathbb{R}$. We define
\begin{align}\label{eq: def T1 unitary}
T_{1;\lambda} & :=  \exp \left(\lambda(B_1 - B_1^\ast)\right)\,,\quad B_1 :=  \frac{1}{L^3}\sum_{p\in \Lambda^*}\hat{\varphi}(p)\widehat{\chi}_>(p)\, {b}_{p,\uparrow}{b}_{-p,\downarrow}\,, \\
\label{eq: def T2 unitary}
T_{2;\lambda} & :=  \exp \left(\lambda(B_2 - B_2^\ast)\right)\,,\quad B_2 :=  \frac{1}{L^3}\sum_{p,r,r^\prime \in \Lambda^*}\hat{\eta}_{r,r^\prime}^\varepsilon(p)\widehat{\chi}_<(p) b_{p,r,\uparrow}b_{-p,r^\prime,\downarrow}\,.
\end{align}
The functions $\hat{\varphi}$ and $\hat{\eta}^\varepsilon_{r,r^\prime}$ were introduced in \cref{def:scattering-phi-eta}. We also write $T_1$ and $T_2$ in place of $T_{1;1}$ and $T_{2;1}$, respectively.
\end{definition}

As a first simplification, we rewrite the excitation density using the particle--hole transformation. Recall that $\psi_{\mathrm{FFG}} = R\Omega$; now if we are interested in a state $\Psi$ containing particles which are non-trivially entangled, we can write $\Psi = R\xi$ with $\xi := R^*\Psi \not= \Omega$. The vector $\xi$ represents the excitations with respect to the non-interacting ground state $\psi_{\mathrm{FFG}}$. By definition of the particle--hole transformation and by the CAR, we have
\[
	R \hat{a}^*_{q,\sigma} \hat{a}_{q,\sigma} R^*
    = \begin{cases} 1 - \hat{a}^*_{q,\sigma} \hat{a}_{q,\sigma} &\text{if } q \in \mathcal{B}_{\F}^\sigma\,, \\ \hat{a}^*_{q,\sigma} \hat{a}_{q,\sigma} &\text{if } q\notin\mathcal{B}_{\F}^\sigma\,.  \end{cases}
\]
With $ n_{q, \sigma} = \hat{a}_{q, \sigma}^* \hat{a}_{q, \sigma} $ and \Cref{eq:nqexc}, this is equivalent to $n^{(\exc)}_{q,\sigma} = R n_{q,\sigma} R^*$, so
\begin{equation} \label{eq:nexc_R_trafo}
\langle \Psi, n^{(\exc)}_{q,\sigma} \Psi \rangle = \langle R\xi, R n_{q,\sigma} R^* R\xi \rangle = \langle \xi, n_{q,\sigma} \xi \rangle \;.
\end{equation}
Likewise, for the averaged mesoscopic observable $n^{(\exc)}_{g,\sigma}$ introduced in \cref{eq:ngpro}, if we define
\begin{equation}\label{eq:ng}
	n_{g,\sigma} := \frac{(2\pi)^3}{L^3} \sum_{p\in \Lambda^*} \hat g(p)\, \hat{a}_{p, \sigma}^* \hat{a}_{p, \sigma} \;, \qquad
    n_g := \sum_{\sigma \in \{\uparrow,\downarrow\}} n_{g,\sigma} \;,
\end{equation}
then we have
\begin{align*}
\langle \Psi, n^{(\exc)}_{g,\sigma} \Psi \rangle 
&= \frac{(2\pi)^3}{L^3} \left\langle \xi, \sum_{p\in \Lambda^*}{\hat g(p)}\left( R^* n^{(\exc)}_{p,\sigma} R \right)\xi \right\rangle  
= \frac{(2\pi)^3}{L^3} \sum_{p\in \Lambda^*}{\hat g(p)} \langle \xi,  n_{p,\sigma} \xi \rangle \\
&= \langle \xi, n_{g,\sigma} \xi \rangle \;. 
\end{align*}
In our main result, the trial state is $\Psi = RT_1 T_2 \Omega$, so by this identity we can equivalently compute
\begin{equation} \label{eq:nexc_simplification}
    \langle \Psi, n^{(\exc)}_g \Psi \rangle = \langle T_1 T_2 \Omega, n_g T_1 T_2 \Omega \rangle \;.
\end{equation}
Our main result will follow by an explicit second order expansion of this quantity.

\section{Preliminary Estimates}
\label{sec:usefullemmas}
In this section, we collect estimates that will be used in the analysis of the error terms. 
We introduce first some notation to shorten the formulas for the commutators arising from the Duhamel expansion.

\subsection{Notation}
We will encounter many times operators of the form $a_\sigma(v_x)$ and $a_\sigma(u_x)$ with $x \in \Lambda$ and $\sigma \in \{\uparrow,\downarrow\}$. Thus, with $k \in \Lambda^*$, we adopt the shortened notation 
\begin{equation}
    c_{x,\sigma} := a_\sigma(v_x)\;, 
    \quad 
    d_{x, \sigma} := a_\sigma(u_x)\;,
    \quad 
    \hat c_{k,\sigma} := \hat v_\sigma(k) \hat a_{k,\sigma}\;, 
    \quad 
    \hat d_{k,\sigma} := \hat u_\sigma(k) \hat a_{k,\sigma}\;.
\end{equation}
Further, we introduce the low- and high-momentum cutoffs and operators as
\begin{align}
\label{eq: u<}
  \hat{u}^<_\sigma(k) &:= \begin{cases} 1 &\text{if } k_{\F}^\sigma < |k| \leq 6\rho_\sigma^{2/9} \,,
  \\
  0 &\text{if } |k| < k_{\F}^\sigma \text{ or } |k| > 6\rho_\sigma^{2/9}\,,\end{cases}
  & 
  \hat d_{k, \sigma}^< 
  &:= \hat u^<_\sigma(k) a_{k,\sigma}\;, 
  &
  d_{x,\sigma}^< 
  & := a_\sigma(u^<_x)\;,
\\
\label{eq: def u>}
\hat{u}^>_\sigma(k) 
&:= \begin{cases}
1 & \text{if } |k| \geq 3\rho^{2/9}\,, \\
0 & \text{if } |k| < 3\rho^{2/9}\,,
\end{cases}
& 
  \hat d_{k, \sigma}^> 
  &:= \hat u^>_\sigma(k) a_{k,\sigma}\;, 
  &
  d_{x,\sigma}^> 
  & := a_\sigma(u^>_x) \;. 
\end{align}
Next, using the definition \eqref{eq: def eta epsilon} of $\hat{\eta}^\varepsilon_{r,r^\prime}$, we write as in \cite[Remark 2.9]{GHNS24}
\begin{equation}\label{eq: formula eta t-integral}
  \hat{\eta}^{\varepsilon}_{r,r^\prime}(p) = (8\pi a) \widehat{\chi}_<(p)\int_0^\infty \ud t\, e^{-(|r+p|^2 - |r|^2 + |r^\prime - p|^2 - |r^\prime|^2 + 2\varepsilon)t}\;,
\end{equation}
for $r+p\notin \mathcal{B}_F^{\uparrow}$,  $r^\prime - p\notin \mathcal{B}_F^{\downarrow}$, $r\in \mathcal{B}_F^{\uparrow}$, and $r^{\prime}\in \mathcal{B}_F^{\downarrow}$. 
It will be convenient to introduce 
\begin{equation}\label{eq: def ut vt}
\begin{aligned}
  \hat{v}_{t,\sigma}(\cdot) 
  &:= e^{t|\cdot|^2}\hat{v}_\sigma(\cdot)\,, \qquad 
  & \hat c_{t,k,\sigma} 
  &:= \hat v_{t,\sigma}(k) \hat a_{k,\sigma} \;, \qquad 
  & c_{t,x,\sigma} 
  &:= a_\sigma(v_{t,x}) \;,
  \\ \hat{u}_{t,\sigma}(\cdot) 
  &:= e^{-t|\cdot |^2}\hat{u}^<_{\sigma}(\cdot)\,, \qquad 
  & \hat d_{t,k,\sigma} 
  &:= \hat u_{t,\sigma}(k) \hat a_{k,\sigma} \;, \qquad 
  & d_{t,x,\sigma} 
  &:= a_\sigma(u_{t,x}) \;, 
\end{aligned}
\end{equation}
where $v_{t,\sigma}$ and $u_{t,\sigma}$ are the Fourier transforms of $\hat{v}_{t,\sigma}$ and $\hat{u}_{t,\sigma}$, and we recall the notation $f_x(\cdot) = f(\cdot - x)$.
Finally, for any $g \in L^2(\Lambda)$, we introduce
\begin{equation}
\begin{aligned}
    c_{\sigma}(g_x)
    &:= a_\sigma( (g * v)_x) \;, 
    \qquad 
    &c_{t,\sigma}(g_x)
    &:= a_\sigma( (g * v_t)_x) \;, \\
    d_{\sigma}(g_x)
    &:= a_\sigma( (g * u)_x) \;, 
    \qquad 
    &d_{t,\sigma}(g_x)
    &:= a_\sigma( (g * u_t)_x) \;. \\
\end{aligned}
\end{equation}

\subsection{Estimates for Combinations of Fermionic Operators}

We collect bounds on the scattering function and the operators defined above. Define 
\begin{equation}\label{eq: def tilde phi}
    \widehat {{\varphi}}^>(p)
    :=\mathcal{F}(\varphi_\infty)(p)  \widehat{\chi}_>(p) \;, \quad
    {\varphi}^>(x)
    := \frac{1}{L^3}\sum_{p\in\Lambda^*} \hat \varphi(p) \widehat{\chi}_>(p) e^{ip\cdot x}\;, \quad x\in \Lambda\;.
\end{equation}
\begin{lemma}[Bounds for ${\varphi}^>$ and $\chi_{<}$]\label{lem: bounds phi} {Let $ V_\infty \in L^2(\RRR^3) $ be non-negative, radial and compactly supported.} Then the following bounds hold uniformly in $ \rho $ for $L$ large:
\begin{equation}\label{eq: L norms phi}
    \|{\varphi}^>\|_{L^\infty(\Lambda)} \leq C \;, \qquad
    \|{\varphi}^>\|_{L^2(\Lambda)} \leq C \rho^{-\frac{1}{9}} \;, \qquad
    \|{\varphi}^>\|_{L^1(\Lambda)} \leq C \rho^{-\frac{4}{9}} \;.
\end{equation}
Furthermore, for $ \chi_{<} $ defined by \eqref{eq: def chi< chi>}, we have
\begin{equation}\label{eq:norms chi}
    \|\chi_{<}\|_{L^1(\Lambda)}\leq C \;, \qquad
   	\|\chi_{<}\|_{L^2(\Lambda)}\leq C \rho^{\frac 13} \;.
\end{equation}
\end{lemma}

For the proof, we refer to \cite[Lemmas~A.2 and A.5]{GHNS24}. The second estimate in \eqref{eq:norms chi} follows directly from the support properties of $\widehat{\chi}_<$.
We may also bound creation and annihilation operators by the following lemma.

\begin{lemma}
\label{lem:v bounds}
There exists $C>0$ such that for all $ \sigma \in \{ \uparrow, \downarrow \} $, $x \in \Lambda$, and $t\in [0,\infty)$,
\begin{equation}  \label{eq: L2 norm v u conv}
    \norm{\hat{v}_\sigma}_2
    \leq C \rho^{\frac 12} \;, \qquad
    \|\hat{v}_{t,\sigma}\|_2
    \leq Ce^{t(k_{\F}^\sigma)^2}\rho^{\frac 12}\;, \qquad
    \|\hat{u}_{t,\sigma}^<\|_2
    \leq Ce^{-t(k_{\F}^\sigma)^2}\rho^{\frac 13} \;,
\end{equation}
and for $\# \in \{ \cdot, *\} $, we have the operator norm bounds
\begin{equation} \label{eq: a conv L2 norm v u conv}
\begin{aligned}
    \|c_{x,\sigma}^{\#}\|
    & = \|\hat{v}_\sigma\|_2 \;, \qquad 
    &\|c_\sigma^{\#}(g_x)\| 
    & =\|{v}_\sigma \ast {g}\|_2 
    \leq \norm{\hat g}_\infty 
    \|\hat{v}_\sigma\|_2 \;, \\
    \|c_{t,x,\sigma}^{\#}\| 
    &=  \|\hat{v}_{t,\sigma}\|_2 \;, \qquad 
    &\|c_{t,\sigma}^{\#}(g_x)\| 
    &=\|{v}_{t,\sigma} \ast {g}\|_2 
    \leq \norm{\hat g}_\infty \|\hat{v}_{t,\sigma}\|_2 \;, \\
    \|(d_{t,x,\sigma}^<)^{\#}\|
    &=  \|\hat{u}_{t,\sigma}^<\|_2 \;, \qquad
    & \|d_{t,\sigma}^{\#}(g_x)\| 
    &=\|{u}_{t,\sigma} \ast {g}\|_2 
    \leq \norm{\hat g}_\infty \|\hat{u}_{t,\sigma}^<\|_2 \;.
\end{aligned}
\end{equation}
\end{lemma}

\begin{proof}
The estimates~\eqref{eq: L2 norm v u conv} follow from the support properties of $\hat{v}_\sigma$ and $\hat{u}_\sigma^<$, compare~\eqref{eq: u<}, and from definition \eqref{eq: def ut vt}.
The estimates~\eqref{eq: a conv L2 norm v u conv} then follow from~\eqref{eq: Pauli}.
\end{proof}

As in \cite{GHNS24, GHNS25}, we frequently encounter operators that are quadratic in the fermionic creation or annihilation operators. More precisely, for $h\in L^1(\Lambda)$, $f,g \in L^2(\Lambda)$, and $\sigma\in\{\uparrow, \downarrow\}$, consider
\begin{equation}\label{eq: def bophgf}
	b_\sigma(h,f,g) := \int \di z\, h(z)a_\sigma(f_z)a_\sigma(g_z) \;.
\end{equation}

\begin{lemma} \label{lem: est b operator}
Let $h\in L^1(\Lambda)$ and  $f,g \in L^2(\Lambda)$; then
\begin{equation} \label{eq: est b op 0}
\begin{aligned}
    \|b_\sigma(h,f,g)\|
    &\leq \|h\|_1 \|\hat{f}_\sigma\|_2 \|\hat{g}_\sigma\|_2\;, \\
\end{aligned}
\end{equation}
and furthermore
\begin{equation} \label{eq: est b op}
    \|b_\sigma(\varphi^>_z, v, u^>)\| \leq C \rho^{\frac{7}{18}} \;, \qquad 
    \|b_\sigma(\varphi^>_z, g*v, u^>)\| \leq C \rho^{\frac{7}{18}} \norm{\hat{g}}_\infty \;.
\end{equation}
\end{lemma}
\begin{proof}
The bound~\eqref{eq: est b op 0} follows as in the proof of~\cite[Lemma~5.3]{FGHP21}. 
The first bound in~\eqref{eq: est b op} is~\cite[Lemma~3.1]{GHNS24} with $\gamma = \frac{1}{9}$, and the second bound follows by the same strategy, using $\norm{\partial^\alpha(g*v)}_2 = \norm{(-ik)^\alpha \hat{g} \hat{v}}_2 \le \norm{\hat{g}}_\infty \norm{\partial^\alpha v}_2$ for any multi-index $ \alpha = (\alpha_1, \alpha_2, \alpha_3) \in \mathbb{N}_0^3 $ and derivative $ \partial^\alpha f(x) = \partial_{x_1}^{\alpha_1} \partial_{x_2}^{\alpha_2} \partial_{x_3}^{\alpha_3} f(x) $.
\end{proof}

Moreover we use the following lemma to estimate operators which are quadratic in the fermionic creation and annihilation operators by the number operator.
\begin{lemma}\label{lem: a*a with N}
For all $ \sigma \in \{ \uparrow, \downarrow \} $, $\psi\in\mathcal{F}_{\mathrm{f}}$ and $\hat{f}\in \ell^\infty(\Lambda^\ast)$, we have 
\begin{equation} \label{eq: est a*a N}
	\int \ud x\, \langle \psi, a^\ast_\sigma(f_x)a_\sigma(f_x) \psi\rangle \leq \|\hat{f}\|_\infty^2 \langle \psi, \mathcal{N} \psi \rangle \;.
\end{equation}
\end{lemma}
\begin{proof}
Direct computation, see \cite[Lemma 3.5]{GHNS25} for details.
\end{proof}
We now state two lemmas for terms with four or six fermionic operators.
\begin{lemma}[Estimate for Quartic Terms]\label{lem:a*a*aa} 
Let  $W_1,W_2\in L^1(\Lambda)$, $W_3,W_4\in  L^2(\Lambda)$ and $\hat f, \hat k\in \ell^\infty(\Lambda^*)$.
Then for all $\sigma,\sigma'\in \{\uparrow,\downarrow\}$, $\# \in \{\cdot, *\}$ and $\psi\in\mathcal{F}_{\mathrm{f}}$, we have
 \begin{align}\label{eq:4aa}
 & \int_{\Lambda^4} \di x \di y \di z \di z' |W_1(x-y) W_2 (z-z') W_3 (x-z) W_4 (y-z')| \times \\
&\quad \times \left| \left\langle \psi,  a^*_\sigma(f_x) a_{\sigma'}^{\#} (g_y) a^{\#}_{\sigma} (h_z) a_{\sigma'} (k_{z'}) \psi  \right\rangle \right| \nonumber \\
 &\le \|W_1\|_1 \|W_2\|_1 \|W_3\|_2 \|W_4\|_2 \|\hat{g}_{\sigma'}\|_2 \|\hat{h}_\sigma\|_2 \left(\int \di x \| a_{\sigma}  (f_x) \psi\|^2 \right)^{\frac 1 2}\left(\int \di z' \| a_{\sigma'}  (k_{z'}) \psi\|^2 \right)^{\frac 1 2}\nonumber \;.
 \end{align}
\end{lemma}
For the proof of Lemma~\ref{lem:a*a*aa} we refer to \cite[Lemma 3.9]{GHNS25}.
\begin{lemma}[Estimate for $a^\ast a^\ast a^\ast aaa$]\label{lem:a*a*a*aaa}
Let $\hat{f}, \hat{g},\hat{h} \in \ell^\infty(\Lambda^\ast)$. Let $\{X(x)\}_{x\in \Lambda}$, $\{Y(y)\}_{y\in \Lambda}$  be two families of bounded operators on $\mathcal{F}_{\mathrm{f}}$. Then for all $\sigma,\sigma'\in \{\uparrow,\downarrow\}$ and $\psi\in\mathcal{F}_{\mathrm{f}}$ we have
\begin{align}\label{eq:6a1}
    &\Bigg| \frac{1}{L^3} \sum_{p\in \Lambda^*} \hat h(p) \Big\langle \int \di x e^{-ipx} X(x) a_\sigma(f_x) \psi, \int \di y e^{-ip\cdot y} Y(y) a_{\sigma'} (g_y) \psi \Big\rangle \Bigg| \\
    &\le \|\hat h\|_\infty  \sup_x \| X(x)\| \sup_y \|Y(y)\| \left( \int \di x  \left\|  a_\sigma(f_x) \psi \right\|^2 \right)^{\frac 1 2} \left(   \int \di y \left\| a_{\sigma'} (g_y) \psi \right\|^2 \right)^{\frac 1 2} \;.  \nonumber
\end{align}
\end{lemma}
For the proof of Lemma~\ref{lem:a*a*a*aaa} we refer to \cite[Lemma 3.7]{GHNS25}.

\subsection{Integral Estimates for Norms of \texorpdfstring{$u$ and $v$}{u and v}}

The following bounds can be found in \cite[Lemma~3.3]{GHNS25} with $\gamma = \frac 19$ plugged in.

\begin{lemma}[Integration over $t$]\label{lem:t} For $v_{t,\sigma}$ and $u_{t,\sigma}$ defined as in \eqref{eq: def ut vt}, given any $\kappa >0$, there exist $C_\kappa, C > 0$ such that
\[
	\int_0^{\infty} \di t \, e^{-2t\varepsilon}e^{2t(k_{\F}^\sigma)^2}\|\hat{u}_{t,\sigma}\|_2^2 \leq C\rho^{\frac{2}{9}} \;,\qquad 
	\int_0^{\infty}\di t  \, e^{-2t\varepsilon}e^{-2t(k_{\F}^\sigma)^2}\|\hat{v}_{t,\sigma}\|_2^2 \leq C_\kappa \rho^{\frac{1}{3}-\kappa} \;.
\]
\end{lemma}
For the proof, we refer to \cite[Lemma 3.4]{GHNS24}, where (recalling $\varepsilon = \rho^{\frac 23 + \delta}$) the constant $C_\kappa$ comes from estimating \begin{equation}\label{eq: est log 1} \log(1 + \rho^{\frac 23} \varepsilon^{-1}) = \log(1 + \rho^{-\delta}) \le C \delta |\log(\rho)| \le C_\kappa \rho^{-\kappa} \;.\end{equation}
In particular, $C_\kappa$ blows up as $\delta \to \infty$, but is finite at fixed, arbitrarily large $\delta$.
From Lemma \ref{lem:t}, it follows that 
\begin{equation}\label{eq: est int t uv}
	\int_0^\infty \di t \, e^{-2t\varepsilon}\|\hat{u}_{t,\sigma}\|_2 \|\hat{v}_{t,\sigma}\|_2
    \leq C_\kappa \rho^{\frac{5}{18}-\kappa} \;. 
\end{equation}
\begin{lemma}[Integration over $t$ and $s$]\label{lem:ts} For $\hat{v}_{t,\sigma},\hat{u}_{t,\sigma}$ defined as in \eqref{eq: def ut vt} with $\sigma\in\{\uparrow, \downarrow\}$, $t\in\mathbb{R}$, given any $\kappa > 0$, there exist $ C_\kappa, C > 0 $ such that
\begin{equation}\label{eq: int ts}
\begin{aligned}
	\int_0^{\infty}\di t  \int_0^\infty \di s  \, e^{2(t+s) (-(k_{\F}^\sigma)^2 -(k_{\F}^{\sigma^\prime})^2 -\varepsilon)}
        \|\hat{v}_{t+s, \sigma}\|_2^2
        \|\hat{v}_{t+s, \sigma^\prime}\|_2^2
    &\leq  C_\kappa \rho^{\frac{2}{3} - \kappa} \;, \\
	\int_0^{\infty}\di t  \int_0^\infty \di s  \, e^{2(t+s) ((k_{\F}^\sigma)^2 +(k_{\F}^{\sigma^\prime})^2 -\varepsilon)} 
        \|\hat{u}_{t+s,\sigma}\|_2^2\|
        \hat{u}_{t+s,\sigma^\prime}\|_2^2
    &\leq C\rho^{\frac{4}{9}} \;, \\
	\int_0^{\infty}\di t  \int_0^\infty \di s  \, e^{2(t+s) (-(k_{\F}^\sigma)^2 +(k_{\F}^{\sigma^\prime})^2 -\varepsilon)} 
        \|\hat{v}_{t+s, \sigma}\|_2^2
        \|\hat{u}_{t+s, \sigma^\prime}\|_2^2 
    &\leq C_\kappa \rho^{\frac{5}{9} - \kappa} \;.
\end{aligned}
\end{equation}
\end{lemma}
\begin{proof}
We start by proving the first bound in \eqref{eq: int ts}:
\begin{align*}
	&\int_0^{\infty}\di t  \int_0^\infty \di s  \, e^{2(t+s) (-(k_{\F}^\sigma)^2 -(k_{\F}^{\sigma^\prime})^2 -\varepsilon)}
        \|\hat{v}_{t+s, \sigma}\|_2^2
        \|\hat{v}_{t+s, \sigma^\prime}\|_2^2 
\\ 
& = \frac{1}{L^6}\sum_{p,k}\hat{v}_\sigma(p)\hat{v}_{\sigma^\prime}(k)\left(\int_0^\infty \di t \, e^{- 2t((k_{\F}^\sigma)^2 + (k_{\F}^{\sigma^\prime})^2 - |p|^2 - |k|^2 + \varepsilon)}\right)^2 
\\ 
& = \frac{1}{4 L^6}\sum_{p,k}\frac{\hat{v}_\sigma(p)\hat{v}_{\sigma^\prime}(k)}{((k_{\F}^\sigma)^2 - |p|^2 + (k_{\F}^{\sigma^\prime})^2 - |k|^2 + \varepsilon)^2} \;.
\end{align*}
We now use that for $ |p| \le |k^\sigma_{\F}| $ and $ |k| \le |k^\sigma_{\F}| $,
\[
	\frac{1}{((k_{\F}^\sigma)^2 - |p|^2 + (k_{\F}^{\sigma^\prime})^2 - |k|^2 + \varepsilon)^2}
    \leq \frac{C}{((k_{\F}^\sigma)^2 - |p|^2 + {\varepsilon/2})((k_{\F}^{\sigma^\prime})^2 - |k|^2 + {\varepsilon/2})} \;.
\]
Bounding the sum by the respective integral, we conclude
\begin{align*}
	&\int_0^{\infty}\di t  \int_0^\infty \di s  \, e^{2(t+s) (-(k_{\F}^\sigma)^2 -(k_{\F}^{\sigma^\prime})^2 -\varepsilon)}
        \|\hat{v}_{t+s, \sigma}\|_2^2
        \|\hat{v}_{t+s, \sigma^\prime}\|_2^2  
	\\ 
	&\leq \frac{C}{L^6}\sum_{p,k}\frac{\hat{v}_\sigma(p)\hat{v}_{\sigma^\prime}(k)}{((k_{\F}^\sigma)^2 - |p|^2 + {\varepsilon/2})((k_{\F}^{\sigma^\prime})^2 - |k|^2 + {\varepsilon/2})}
    \leq C_\kappa \rho^{\frac{2}{3} - \kappa} \;,
\end{align*}
for any $\kappa >0$. To see the last inequality, view the sum as a Riemann sum and bound as in \cite{GHNS24}
\begin{equation}\label{eq: est with log}
    \sum_p \frac{\hat{v}_\sigma(p)}{(k_{\F}^\sigma)^2 - |p|^2 + \varepsilon/2} \leq C \kF^\sigma \log\left( 1 + \frac{2(\kF^\sigma)^2}{\varepsilon} \right) 
    \leq C_\kappa \rho^{1/3 -\kappa} \;,\qquad \forall \kappa >0 \;.
\end{equation}
The second inequality in~\eqref{eq: int ts} can be proved in the same way:
\begin{align*}
& \int_0^{\infty}\di t  \int_0^\infty \di s  \, e^{2(t+s) ((k_{\F}^\sigma)^2 +(k_{\F}^{\sigma^\prime})^2 -\varepsilon)} 
    \|\hat{u}_{t+s,\sigma}\|_2^2
    \|\hat{u}_{t+s,\sigma^\prime}\|_2^2 
\\ 
&= \frac{1}{L^6}\sum_{p,k}\hat{u}^<_\sigma(p)\hat{u}^<_{\sigma^\prime}(k) \left(\int_0^{\infty} \di t \, e^{-2t( |p|^2 + |k|^2 - (k_{\F}^\sigma)^2 - (k_{\F}^{\sigma^\prime})^2 + \varepsilon)}\right)^2
\\ 
& \leq \frac{C}{L^6}\sum_{p,k}\frac{\hat{u}^<_\sigma(p)\hat{u}^<_{\sigma^\prime}(k)}{(|p|^2 - (k_{\F}^\sigma)^2 + \varepsilon)(|k|^2 - (k_{\F}^{\sigma^\prime})^2 + \varepsilon)}
\leq C \rho^{\frac{4}{9}} \;.
\end{align*}
Finally, we prove the third inequality in~\eqref{eq: int ts}. Writing the term explicitly we find
\begin{align*}
&\int_0^{\infty}\di t  \int_0^\infty \di s  \, e^{2(t+s) (-(k_{\F}^\sigma)^2 +(k_{\F}^{\sigma^\prime})^2 -\varepsilon)} 
    \|\hat{v}_{t+s, \sigma}\|_2^2
    \|\hat{u}_{t+s, \sigma^\prime}\|_2^2 
\\ 
& = \frac{1}{L^6}\sum_{k,p} \hat{v}_\sigma(p) \hat{u}^<_{\sigma^\prime}(k)
    \left( \int_0^{\infty}\di t \, 
    e^{-2t(\varepsilon + (k_{\F}^\sigma)^2 - |p|^2 + |k|^2 - (k_{\F}^{\sigma^\prime})^2)} \right)^2
\\ 
&= \frac{1}{4 L^6}\sum_{k,p} \frac{\hat{v}_{\sigma}(p) \hat{u}^<_{\sigma^\prime}(k)}{(|k|^2 - (k_{\F}^{\sigma^\prime})^2 + (k_{\F}^\sigma)^2 - |p|^2 + \varepsilon)^2} 
\\ 
&\leq \left(\frac{1}{2 L^3}\sum_p \frac{\hat{v}_{\sigma}(p)}{(k_{\F}^{\sigma^\prime})^2 - |p|^2 + {\varepsilon/2}}\right)
    \left(\frac{1}{2 L^3}\sum_k \frac{\hat{u}^<_{\sigma^\prime}(k)}{|k|^2 - (k_{\F}^{\sigma^\prime})^2 + {\varepsilon/2}}\right)
\leq C_\kappa \rho^{\frac{5}{9} -\kappa} \;.
\end{align*}
This concludes the proof.
\end{proof}

\subsection{Estimate of the Number Operator}

We frequently need bounds for the number operator $ \cN $.
\begin{lemma}[Estimate for $\mathcal{N}$]\label{lem: est number op} Recall $T_{1;\lambda}$ and $T_{2;\lambda}$ from~\eqref{eq: def T1 unitary} and~\eqref{eq: def T2 unitary} with $\lambda\in [0,1]$. Then for any $\kappa>0$, there exists $ C_\kappa > 0 $ such that
\begin{equation}\label{eq: est number op}
\begin{aligned}
  \langle T_{2;\lambda}\Omega, \mathcal{N} T_{2;\lambda}\Omega\rangle
  &\leq C_\kappa L^3\rho^{\frac{14}{9} -\kappa} \;, \\
  \langle T_{1;\lambda}T_2\Omega, \mathcal{N} T_{1;\lambda}T_2\Omega\rangle
  &\leq C_\kappa L^3\rho^{\frac{14}{9} -\kappa} \;.
\end{aligned}
\end{equation}
\end{lemma}
This differs from~\cite[Prop.~3.5]{GHNS24} due to a refined treatment of the logarithmic factors appearing in Lemma \ref{lem:t}, see  \eqref{eq: est log 1}.
\begin{proof}
We use Grönwall's lemma as in~\cite[Prop.~3.5]{GHNS24} (a method which goes back to the study of the dynamics of many-body systems, e.\,g., \cite{RS09,BPS14}): we compute
\begin{align*}
    &|\partial_\lambda \langle T_{2;\lambda}\Omega, \mathcal{N} T_{2;\lambda}\Omega\rangle| \\
    &\leq C \int_0^\infty \di t \; e^{-2 t \varepsilon}
        \int \di x \di y \; |\chi_<(x-y)|
            |\langle T_{2;\lambda}\Omega,
            a_\uparrow(u_{t,x}) a_\uparrow(v_{t,x}) a_\downarrow(u_{t,y}) a_\downarrow(v_{t,y})
            T_{2;\lambda}\Omega \rangle| \\
    &\leq C L^{\frac 32} \rho^{\frac 12}
        \| \cN^{\frac 12} T_{2;\lambda}\Omega \|
        \int_0^\infty \di t \; e^{-2 t \varepsilon}
        \norm{\hat{u}_{t,\uparrow}}_2 \norm{\hat{v}_{t,\uparrow}}_2 \;.
\end{align*}
We estimate the integral as in \eqref{eq: est int t uv}, then apply the Cauchy--Schwarz inequality, followed by Gr\"onwall's lemma recalling that $\lambda \leq 1$, and arrive at the claimed bounds.
\end{proof}

\section{Many-Body Analysis}
To prove \cref{thm:main} and extract the leading-order of $ \langle T_1T_2\Omega, n_g T_1T_2\Omega \rangle $, we perform two Duhamel expansions and estimate all resulting error terms.
\subsection{Second Order Duhamel Expansion}
\label{sec:Duhamel}

\begin{lemma}[Duhamel Expansion]
\label{lem:Duhamelexpansion}
We have
\begin{equation} \label{eq:Duhamelexpansion}
\begin{aligned}
& \expect{T_1T_2\Omega, n_g T_1T_2\Omega} \\
  & = \int_0^1 \ud \lambda_2 \, (1-\lambda_2) \expect{T_{2;\lambda_2}\Omega, [ [n_g, B_2-B_2^*], B_2-B_2^* ] T_{2;\lambda_2}\Omega}
  \\
  & \quad + \int_0^1 \ud \lambda_2 \expect{T_{2;\lambda_2}\Omega, [ [n_g, B_1-B_1^*],B_2-B_2^*] T_{2;\lambda_2}\Omega}
  \\ & \quad
      + \int_0^1 \ud \lambda_1 \, (1-\lambda_1) \expect{T_{1;\lambda_1}T_2\Omega, [[n_g, B_1-B_1^*], B_1-B_1^*] T_{1;\lambda_1}T_2\Omega} \;.
\end{aligned}
\end{equation}
\end{lemma}
\begin{proof}
Consider any bounded operator $ A $.
We first expand to second order in $T_1$ and then expand the resulting terms in $T_2$, so all terms are expanded to second order joint in $T_1$ and $T_2$:
\begin{align*}
&\expect{T_1T_2\Omega, A T_1T_2\Omega} \\
  & = \expect{T_2\Omega, A T_2\Omega} + \expect{T_2\Omega, [A,B_1-B_1^*] T_2\Omega}
  \\ & \quad
    + \int_0^1 \ud \lambda_1 \, (1-\lambda_1) \expect{T_{1;\lambda_1}T_2\Omega, [[A, B_1 - B_1^*], B_1 - B_1^*] T_{1;\lambda_1}T_2\Omega}
  \\
  & = \expect{\Omega, A \Omega}
    + \expect{\Omega, [A,B_2-B_2^*] \Omega}
  \\ & \quad 
    + \int_0^1 \ud \lambda_2 \, (1-\lambda_2) \expect{T_{2;\lambda_2}\Omega, [ [A, B_2-B_2^*], B_2 - B_2^*] T_{2;\lambda_2}\Omega}
  \\ & \quad 
    + \expect{\Omega, [A,B_1-B_1^*] \Omega}
    + \int_0^1 \ud \lambda_2 \expect{T_{2;\lambda_2}\Omega, [ [A, B_1 - B_1^*], B_2 - B_2^*] T_{2;\lambda_2}\Omega}
  \\ & \quad 
    + \int_0^1 \ud \lambda_1 \, (1-\lambda_1) \expect{T_{1;\lambda_1}T_2\Omega, [[A, B_1 - B_1^*], B_1 - B_1^*] T_{1;\lambda_1}T_2\Omega} \;.
\end{align*}
For $A=n_g = \frac{(2\pi)^3}{L^3} \sum_{k,\sigma} \hat g(k) \hat a_{k,\sigma}^* \hat a_{k,\sigma}$ as in \eqref{eq:ng}, the vacuum expectation of $A$ vanishes since $ \hat a_{k,\sigma} \Omega = 0 $. Likewise, $[A, B_j - B_j^*]$ amounts to a sum of a term with 4 creation operators and another one with 4 annihilation operators, so its vacuum expectation vanishes. This renders the desired result.
\end{proof}

\subsection{Analysis of \texorpdfstring{$[[n_g , B_2 - B_2^\ast], B_2 - B_2^\ast]$}{[[ng,B2-B2*],B2-B2*]}}\label{sec: B2}

The main result of this section is the following proposition. 
\begin{proposition}\label{pro:B2}
Let $n_g$ and $B_2$ as in \eqref{eq:ng} and \eqref{eq: def T2 unitary}, respectively. Then
\begin{multline*}
	 [[n_g, B_2 - B_2^\ast], B_2 - B_2^\ast] 
	\\ 
	= {\frac{2(2\pi)^3}{L^9}} \sum_{\substack{|r| \leq \kF^\uparrow < |r+p| \\ |r^\prime| \leq \kF^\downarrow < |r^\prime - p|}}\left(\hat{g}(r+p) + \hat{g}(-r) + \hat{g}(r^\prime - p) + \hat{g}(-r^\prime) \right)\left(\hat{\eta}^{\varepsilon}_{r,r^\prime}(p) \widehat{\chi}_<(p)\right)^2 + \mathcal{E}_{B_2} \;,
\end{multline*}
where, given $ \kappa > 0 $, there exists $ C_\kappa > 0 $ such that for any $\psi\in \mathcal{F}_{\mathrm{f}}$,
\[
	|\langle \psi, \mathcal{E}_{B_2}\psi \rangle|
    \leq C_\kappa {L^{-3}}\|\hat{g}\|_\infty \rho^{\frac 12 -\kappa}\langle\psi, \mathcal{N}\psi\rangle \;.
\]
\end{proposition}
\begin{proof}
First, note that (recall the formula in \eqref{eq: def T2 unitary} for $B_2$)
\begin{align*}
 		B_2 
 		&  
        = \frac{1}{L^3}\sum_{p,r,r^\prime \in \Lambda^*}\hat{\eta}_{r,r^\prime}^\varepsilon(p)\widehat{\chi}_<(p) 
        \hat{d}_{r+p,\uparrow}\hat{c}_{-r,\uparrow}\hat{d}_{r^\prime - p, \downarrow}\hat{c}_{-r^\prime, \downarrow} \;.
\end{align*}
Recall that $n_g = n_{g,\uparrow} + n_{g,\downarrow}$, with $n_{g,\sigma}$ as in \eqref{eq:ng}. We only show the computations for the double commutator with $n_{g,\uparrow}$, the commutator with $n_{g,\downarrow}$ is analyzed similarly.
Using the definition of $B_2$ in \eqref{eq: def T2 unitary}, we compute $[[n_{g,\uparrow}, B_2 - B_2^\ast], B_2 - B_2^\ast]$. 
By a straightforward but lengthy computation, putting all terms in normal order, we find
\begin{align} \label{eq:I1I9_expansion}
	[[n_{g,\uparrow}, B_2 - B_2^\ast], B_2 - B_2^\ast] = {\frac{(2\pi)^3}{L^3}}\sum_{j=1}^9 \mathrm{I}_j + \mathrm{h.c.} \;,
\end{align}
with
\allowdisplaybreaks[4]
\begin{align*}
	\mathrm{I}_1 &:= \frac{1}{L^6}\sum_{p,r,r^\prime} \left(\hat{g}(r+p) + \hat{g}(-r) \right)\left(\hat{\eta}^{\varepsilon}_{r,r^\prime}(p) \widehat{\chi}_<(p)\right)^2\hat{u}_\uparrow(r+p)\hat{u}_\downarrow(r^\prime - p)\hat{v}_\uparrow(r)\hat{v}_\downarrow(r^\prime) \;,
\\
\mathrm{I}_2 & = -\frac{1}{L^6}\sum_{\sigma\neq\sigma^\prime} \sum_{p,r,r^\prime}
    \left(\delta_{\sigma,\uparrow}(\hat{g}(r+p) + \hat{g}(-r))
    +\delta_{\sigma,\downarrow} (\hat{g}(r^\prime - p) + \hat{g}(-r^\prime)) \right)
    \left(\hat{\eta}^{\varepsilon}_{r,r^\prime}(p) \widehat{\chi}_<(p)\right)^2 \times
\\ 
&\qquad \times \hat{v}_\sigma(r) \hat{u}_{\sigma^\prime}(r^\prime - p) \hat{v}_{\sigma^\prime}(r^\prime)
        \hat{d}^\ast_{r+p, \sigma}
        \hat{d}_{r+p,\sigma} \;,
\\
	\mathrm{I}_3 &:= \frac{1}{L^6}\sum_{p,k,r,r^\prime}\left(\hat{g}(r+p) + \hat{g}(-r) \right) \hat{\eta}^\varepsilon_{r,r^\prime}(p)\widehat{\chi}_<(p)\hat{\eta}^{\varepsilon}_{r,r^\prime}(k)\widehat{\chi}_<(k)\times
	\\ 
	&\qquad \times \hat{v}_\uparrow(r)\hat{v}_\downarrow(r^\prime)
        \hat{d}_{r+p,\uparrow}^\ast
        \hat{d}_{r^\prime - p, \downarrow}^\ast
        \hat{d}_{r^\prime - k, \downarrow}
        \hat{d}_{r+k,\uparrow} \;,
\\
	\mathrm{I}_4 &:= - \frac{1}{L^6}\sum_{\sigma\neq \sigma^\prime}\sum_{p,r, r^\prime} 
    \left(\delta_{\sigma,\uparrow}(\hat{g}(r+p) + \hat{g}(-r))
    +\delta_{\sigma,\downarrow} (\hat{g}(r^\prime - p) + \hat{g}(-r^\prime)) \right)
    \times
	\\ 
	&\qquad \times \left(\hat{\eta}^{\varepsilon}_{r,r^\prime}(p) \widehat{\chi}_<(p)\right)^2 \hat{u}_\sigma(r+p)\hat{u}_{\sigma^\prime}(r^\prime - p)\hat{v}_{\sigma^\prime}(r^\prime)
        \hat{c}_{-r,\sigma}^\ast 
        \hat{c}_{-r,\sigma} \;,
\\
	\mathrm{I}_5 &:= \frac{1}{L^6}\sum_{\sigma\neq \sigma^\prime}\sum_{p,k,r,r^\prime}
    \left(\delta_{\sigma,\uparrow}(\hat{g}(r+p) + \hat{g}(-r))
    +\delta_{\sigma,\downarrow} (\hat{g}(r^\prime - p) + \hat{g}(-r^\prime)) \right)
    \hat{\eta}_{r,r^\prime}^{\varepsilon}(p)\widehat{\chi}_<(p)\times
	\\ 
	&\quad \times \hat{\eta}_{r,r^\prime - p +k}(k)\widehat{\chi}_<(k)\hat{u}_{\sigma^\prime}(r^\prime - p)\hat{v}_\sigma(r)
        \hat{d}_{r+p,\sigma}^\ast 
        \hat{c}^\ast_{-r^\prime, \sigma^\prime} 
        \hat{c}_{-r^\prime + p - k, \sigma^\prime} 
        \hat{d}_{r+k,\sigma} \;,
\\
	\mathrm{I}_6 &:= - \frac{1}{L^6}\sum_{\sigma\neq \sigma^\prime}\sum_{p,k,r,r^\prime,s}
    \left(\delta_{\sigma,\uparrow}(\hat{g}(r+p) + \hat{g}(-r))
    +\delta_{\sigma,\downarrow} (\hat{g}(r^\prime - p) + \hat{g}(-r^\prime)) \right)
    \hat{\eta}^{\varepsilon}_{r,r^\prime}(p) \times
	\\ 
	&\qquad \times  \widehat{\chi}_<(p)\hat{\eta}^{\varepsilon}_{r,s}(k)\widehat{\chi}_<(k)\hat{v}_\sigma(r)
        \hat{d}_{r+p,\sigma}^\ast
        \hat{c}_{-r^\prime, \sigma^\prime}^\ast
        \hat{d}^\ast_{r^\prime - p, \sigma^\prime}
        \hat{d}_{s-k,\sigma^\prime}
        \hat{c}_{-s, \sigma^\prime}
        \hat{d}_{r+k,\sigma} \;,
\\
	\mathrm{I}_7 &:= \frac{1}{L^6}\sum_{\sigma\neq \sigma^\prime}\sum_{p,r,r^\prime, s}\left(\delta_{\sigma,\uparrow}(\hat{g}(r+p) + \hat{g}(-r))
    +\delta_{\sigma,\downarrow} (\hat{g}(r^\prime - p) + \hat{g}(-r^\prime)) \right)
    \hat{\eta}_{r,r^\prime}^{\varepsilon}(p) \times
	\\ 
	&\qquad \times \widehat{\chi}_<(p)^2\hat{\eta}_{s,r^\prime}^{\varepsilon}(p) \hat{u}_{\sigma^\prime}(r^\prime - p) \hat{v}_{\sigma^\prime}(r^\prime) 
        \hat{c}^{\ast}_{-r,\sigma}
        \hat{d}_{r+p,\sigma}^\ast 
        \hat{d}_{s+p, \sigma}
        \hat{c}_{-s, \sigma} \;,
\\
	\mathrm{I}_8 &:= \frac{1}{L^6}\sum_{p,r,r^\prime, k}\left(\hat{g}(r+p) + \hat{g}(-r) \right)\hat{\eta}_{r,r^\prime}^{\varepsilon}(p)\widehat{\chi}_<(p)\hat{\eta}_{r+p-k,r^\prime - p + k}^{\varepsilon}(k)\widehat{\chi}_<(k) \times
	\\ 
	&\quad \times \hat{u}_\uparrow(r+p)\hat{u}_\downarrow(r'-p)
        \hat{c}_{-r, \uparrow}^\ast 
        \hat{c}_{-r^\prime, \downarrow}^\ast 
        \hat{c}_{-r^\prime +p-k,\downarrow}
        \hat{c}_{-r-p+k, \uparrow} \;,
\\
	\mathrm{I}_9 &:= -\frac{1}{L^6}\sum_{\sigma\neq \sigma^\prime}\sum_{p,k,r,r^\prime,s}
    \left(\delta_{\sigma,\uparrow}(\hat{g}(r+p) + \hat{g}(-r))
    +\delta_{\sigma,\downarrow} (\hat{g}(r^\prime - p) + \hat{g}(-r^\prime)) \right)
    \hat{\eta}_{r,r^\prime}^\varepsilon(p) \times
	\\ 
	&\qquad \times   \widehat{\chi}_<(p)\hat{\eta}_{r+p-k,s}(k)\widehat{\chi}_<(k)\hat{u}_\sigma(r+p)
        \hat{c}_{-r,\sigma}^\ast 
        \hat{c}_{-r^\prime, \sigma^\prime}^\ast 
        \hat{d}_{r^\prime - p, \sigma^\prime}^\ast
        \hat{d}_{s-k, \sigma^\prime}
        \hat{c}_{-s, \sigma^\prime}
        \hat{c}_{-r-p+k, \sigma} \;.
\end{align*}	
\allowdisplaybreaks[1]

The first term, $\mathrm{I}_1$, is the constant (not containing any operators) we want to extract. We estimate $\mathrm{I}_j$ for $j=2, \ldots 9$, which are all error terms.

By anti-normal ordering the $c$ and $c^*$ operators in $\mathrm I_5$ we absorb the term $\mathrm I_2$ as
\begin{align*}
	\mathrm{I}_2 + \mathrm{I}_5  &= -\frac{1}{L^6}\sum_{\sigma\neq \sigma^\prime}\sum_{p,k,r,r^\prime}
    \left(\delta_{\sigma,\uparrow}(\hat{g}(r+p) + \hat{g}(-r))
    +\delta_{\sigma,\downarrow} (\hat{g}(r^\prime - p) + \hat{g}(-r^\prime)) \right)
    \hat{\eta}_{r,r^\prime}^{\varepsilon}(p) \times
	\\ 
	&\quad \times \widehat{\chi}_<(p)\hat{\eta}_{r,r^\prime - p +k}(k)\widehat{\chi}_<(k)\hat{u}_{\sigma^\prime}(r^\prime - p)\hat{v}_\sigma(r)
    \hat{d}_{r+p,\sigma}^\ast \hat{c}_{-r^\prime + p - k, \sigma^\prime}\hat{c}^\ast_{-r^\prime, \sigma^\prime} \hat{d}_{r+k,\sigma} \;. 
\end{align*}
Now, as observed in \cite{GHNS24}, the constraints on $p$, $r$, and $r^\prime$ allow us to replace both $\hat d_{r+p,\sigma}$ and $\hat d_{r+k,\sigma}$ by $\hat d_{r+p,\sigma}^<$ and $\hat d_{r+k,\sigma}^<$, defined in \eqref{eq: u<} above. With~\eqref{eq: formula eta t-integral} and~\eqref{eq: def ut vt}, we can then conveniently re-write the right-hand side in position space and split it according to the $4$ summands $\hat g(\cdot)$ and
$\mathrm{I}_5  + \mathrm{I}_2= \mathrm{A} + \mathrm{B} + \mathrm{C} + \mathrm{D}$ with
\begin{align*}
	\mathrm{A} &:=  -(8\pi a)^2 \int_0^\infty \di t \int_0^\infty \di s \, e^{-2(t+s)\varepsilon} \int \di x \di y \di z \di z^\prime \chi_<(x-y)\chi_<(z-z^\prime)\times
	\\ 
	&\quad \times u_{t+s,\downarrow}(y;z^\prime)v_{t+s,\uparrow}(x;z) d^\ast_{t,\uparrow}(g_x)c_{s,z'\downarrow}c_{t,y,\downarrow}^\ast d_{s,z,\uparrow} \;,
\\
	\mathrm{B} &:= -(8\pi a)^2 \int_0^\infty \di t \int_0^\infty \di s  \, e^{-2(t+s)\varepsilon}\int \di x \di y \di z \di z^\prime \chi_<(x-y)\chi_<(z-z^\prime)\times
	\\ 
	& \quad \times u_{t+s,\downarrow}(y;z^\prime)(v_{t+s,\uparrow}\ast g)(x;z)  d^\ast_{t,x\uparrow}c_{s,z',\downarrow} c_{t,y,\downarrow}^\ast d_{s,z,\uparrow} \;,
\\
	\mathrm{C} &:= -(8\pi a)^2 \int_0^\infty \di t \int_0^\infty \di s  \, e^{-2(t+s)\varepsilon}\int \di x \di y \di z \di z^\prime \chi_<(x-y)\chi_<(z-z^\prime)\times
	\\ 
	& \quad \times (u_{t+s,\uparrow}\ast g)(y;z^\prime)v_{t+s,\downarrow}(x;z)  d^\ast_{t,x,\downarrow}c_{s,z',\uparrow} c_{t,y,\uparrow}^\ast d_{s,z,\downarrow} \;,
\\
	\mathrm{D} &:= -(8\pi a)^2 \int_0^\infty \di t \int_0^\infty \di s  \, e^{-2(t+s)\varepsilon}\int \di x \di y \di z \di z^\prime \chi_<(x-y)\chi_<(z-z^\prime)\times
	\\ 
	& \quad \times u_{t+s,\uparrow}(y;z^\prime)v_{t+s,\downarrow}(x;z)  
    d^\ast_{t,x,\downarrow}c_{s,z',\uparrow} c_{t,\uparrow}^\ast(g_y) d_{s,z,\downarrow} \;.
\end{align*}

We start with estimating $\mathrm{A}$. An application of \Cref{lem:v bounds,lem:a*a*aa}, followed by \Cref{lem: a*a with N} with $ \Vert \hat{u}_{t,\sigma} \Vert_\infty \le e^{-t(k^\sigma_{\F})^2} $ gives 
\begin{align*}
	|\langle\psi, \mathrm{A}\psi\rangle|
    &\leq  C\|\hat{g}\|_\infty \int_0^\infty \di t  \int_0^\infty \di s \, e^{-2\varepsilon(t+s)} \int \di x \di y \di z \di z^\prime\, |\chi_<(x-y)||\chi_<(z-z^\prime)| \times \\ 
	&\qquad \times |u_{t+s, \downarrow}(y;z^\prime)|
        |v_{t+s,\uparrow}(x;z)|
        \|\hat{v}_{s, \downarrow}\|_2
        \|\hat{v}_{t,\downarrow}\|_2
        \|d_{t,\uparrow}(g_x)\psi\| 
        \|d_{s,z,\uparrow}\psi\| \\
    &\leq C\|\hat{g}\|_\infty \|\chi_<\|_1^2 \langle \psi, \mathcal{N}\psi\rangle \times
	\\ 
	&\quad \times \left(\int_0^\infty \di t  \int_0^\infty \di s \, e^{-2(t+s)\varepsilon} e^{-(t+s)(k_{\F}^\uparrow)^2}
        \|\hat{v}_{s,\downarrow}\|_2
        \|\hat{v}_{t,\downarrow}\|_2
        \|\hat{u}_{t+s,\downarrow}\|_2
        \|\hat{v}_{t+s, \uparrow}\|_2 \right) \;.
\end{align*}
With the aid of \Cref{lem:t,lem:ts} we estimate the integrals with respect to $t$ and $s$:
\begin{align}
    &\int_0^\infty \di t \int_0^\infty \di s \, e^{-2\varepsilon(t+s)}e^{-(t+s)(\kF^\uparrow)^2}
        \|\hat{v}_{s,\downarrow}\|_2
        \|\hat{v}_{t,\downarrow}\|_2
        \|\hat{u}_{t+s,\downarrow}\|_2
        \|\hat{v}_{t+s, \uparrow}\|_2
\nonumber \\ 
&\leq \left(\int_0^\infty \di t \int_0^\infty \di s \, e^{-2\varepsilon(t+s)}e^{2(t+s)(\kF^\downarrow)^2} e^{-2(t+s)(\kF^\uparrow)^2}
        \|\hat{u}_{t+s,\downarrow}\|_2^2
        \|\hat{v}_{t+s,\uparrow}\|_2^2 \right)^{\frac{1}{2}} \times
\nonumber \\ 
&\quad \times \left(\int_0^\infty \di t \int_0^\infty \di s \, e^{-2\varepsilon(t+s)}e^{-2(t+s)(\kF^\downarrow)^2}
        \|\hat{v}_{s,\downarrow}\|_2^2
        \|\hat{v}_{t,\downarrow}\|_2^2 \right)^{\frac{1}{2}}
\leq C_\kappa \rho^{\frac{11}{18} - \kappa} \;. \label{eq:st_estimation_strategy}
\end{align}
The terms $B$, $C$ and $D$ are estimated likewise. By Lemma~\ref{lem: bounds phi}, {$\|\chi_<\|_1^2 \le C$}, so
\[
	|\langle\psi, (\mathrm{I}_2 +\mathrm{I}_5) \psi\rangle|
    \leq C_\kappa \|\hat{g}\|_\infty\rho^{\frac{11}{18} - \kappa}\langle \psi, \mathcal{N}\psi\rangle \;.
\]
Next, we consider $\mathrm{I}_3$. We split this into two terms: $\mathrm{I}_{3;a}$ with $\hat g(p+r)$ and $\mathrm{I}_{3;b}$ with $\hat g(-r)$: 
\begin{equation} \label{eq:I3_split}
\begin{aligned}
	\mathrm{I}_{3;a} &:= \frac{1}{L^6}\sum_{p,k,r,r^\prime}\hat{g}(r+p) \hat{\eta}^\varepsilon_{r,r^\prime}(p)\widehat{\chi}_<(p)\hat{\eta}^{\varepsilon}_{r,r^\prime}(k)\widehat{\chi}_<(k)
        \hat{v}_\uparrow(r)
        \hat{v}_\downarrow(r^\prime)
        \hat{d}_{r+p,\uparrow}^\ast 
        \hat{d}_{r^\prime - p, \downarrow}^\ast 
        \hat{d}_{r^\prime - k, \downarrow}
        \hat{d}_{r+k,\uparrow}\;, \\
	\mathrm{I}_{3;b} &:= \frac{1}{L^6}\sum_{p,k,r,r^\prime}\hat{g}(r) \hat{\eta}^\varepsilon_{r,r^\prime}(p)\widehat{\chi}_<(p)\hat{\eta}^{\varepsilon}_{r,r^\prime}(k)\widehat{\chi}_<(k)
        \hat{v}_\uparrow(r)
        \hat{v}_\downarrow(r^\prime)
        \hat{d}_{r+p,\uparrow}^\ast
        \hat{d}_{r^\prime - p, \downarrow}^\ast
        \hat{d}_{r^\prime - k, \downarrow}
        \hat{d}_{r+k,\uparrow} \;.
\end{aligned}
\end{equation}
We re-write $\mathrm{I}_{3;a}$ in position space, using~\eqref{eq: formula eta t-integral} and~\eqref{eq: def ut vt}, where the momentum constraints allow replacing $ \hat{d}_{r+p,\uparrow} $ and $ \hat{d}_{r+k,\uparrow} $ by $ \hat{d}_{r+p,\uparrow}^< $ and $ \hat{d}_{r+k,\uparrow}^< $:
\begin{align*}
	\mathrm{I}_{3; a} &:= (8\pi a)^2\int_0^{\infty}\di t  \int_0^\infty \di s \, e^{-2\varepsilon(t+s)}\int \di x \di y \di z \di z^\prime \, 
    \chi_<(x-y)\chi_<(z-z^\prime)\times
	\\ 
	&\qquad \times  v_{t+s,\uparrow}(x;z)v_{t+s,\downarrow}(y;z^\prime)d^\ast_{t,\uparrow}(g_x) d_{t,y,\downarrow}^\ast d_{s,z',\downarrow} d_{s,z,\uparrow}
    \;.
\end{align*}
Using~\Cref{lem:a*a*aa,lem:v bounds,lem: a*a with N}, we get   
\begin{align*}
	|\langle \psi, \mathrm{I}_{3;a}\psi\rangle| 
	&\leq C\|\hat{g}\|_\infty\|\chi_<\|_1^2 \langle \psi, \mathcal{N}\psi\rangle \times
	\\ 
	&\quad \times \left(\int_0^\infty \di t \int_0^\infty \di s \, e^{-2\varepsilon(t+s)}e^{-(t+s)(\kF^\uparrow)^2}  \|\hat{v}_{t+s,\uparrow}\|_2 \|\hat{v}_{t+s,\downarrow}\|_2 \|\hat{u}_{t,\downarrow}\|_2 \|\hat{u}_{s,\downarrow}\|_2\right) \,.
\end{align*}
For the term $\mathrm{I}_{3;b}$, we proceed in the same way, {using  also \eqref{eq: a conv L2 norm v u conv}, i.\,e., $\|v_{t+s, \uparrow}\ast g\|_2 \leq \|\hat{g}\|_\infty\|\hat{v}_{t+s,\uparrow}\|_2$}. The integrals with respect to $t$ and $s$ are bounded via \Cref{lem:t,lem:ts} as in \eqref{eq:st_estimation_strategy}. 
Thus, recalling \Cref{lem: bounds phi} to bound $\norm{\chi_<}_1 \le C$, we conclude that 
\[
	|\langle \psi, \mathrm{I}_3 \psi\rangle|
    \leq C_\kappa \|\hat{g}\|_\infty\rho^{\frac{5}{9} -\kappa}\langle \psi, \mathcal{N}\psi\rangle \;.
\]
To estimate $\mathrm{I}_4$, it is convenient to combine it with $\mathrm{I}_7$, similarly as what we did for $\mathrm{I}_2$ and $\mathrm{I}_5$, by partially anti-normal ordering the $d$ and $d^*$ operators in $\mathrm{I}_7$. 
We thus write 
\begin{align*}
	\mathrm{I}_4 + \mathrm{I}_7 &= -\frac{1}{L^6}\sum_{\sigma\neq \sigma^\prime}\sum_{p,r,r^\prime, s}\left(\delta_{\sigma,\uparrow} (\hat{g}(r+p) + \hat{g}(-r)) + \delta_{\sigma,\downarrow} (\hat{g}(r^\prime - p) + \hat{g}(-r^\prime))\right) \times
	\\ 
	&\qquad \times \hat{\eta}_{r,r^\prime}^{\varepsilon}(p)\widehat{\chi}_<(p)^2 \hat{\eta}_{s,r^\prime}^{\varepsilon}(p)
    \hat{u}_{\sigma^\prime}(r^\prime - p) \hat{v}_{\sigma^\prime}(r^\prime) 
    \hat{c}^{\ast}_{-r,\sigma}\hat{d}_{s+p, \sigma}\hat{d}_{r+p,\sigma}^\ast \hat{c}_{-s, \sigma} \;.
\end{align*}
To estimate the term above, we proceed as for $\mathrm{I}_2 + \mathrm{I}_5$.
Rewriting in position space, {we split $\mathrm{I}_4 + \mathrm{I}_7 $ according to the $4$ summands $\hat g(\cdot)$, i.\,e.,} $\mathrm{I}_4 + \mathrm{I}_7 = \mathrm{E} + \mathrm{F} + \mathrm{G} + \mathrm{H}$ with
\begin{align*}
	\mathrm{E} &:=-(8\pi a)^2 \int_0^\infty \di t \int_0^\infty \di s  \, e^{-2(t+s)\varepsilon}\int \di x \di y\, (\chi_< \ast (u_{t+s,\downarrow}v_{t+s,\downarrow}) * \chi_<)(x-y)\times
	\\ 
	&\quad \times c^\ast_{t,x,\uparrow} d_{s,y,\uparrow} d_{t,\uparrow}(g_x)^* c_{s,y,\uparrow} \;,
\\
	\mathrm{F} &:=-(8\pi a)^2 \int_0^\infty \di t \int_0^\infty \di s  \, e^{-2(t+s)\varepsilon}\int \di x \di y\, (\chi_< \ast (u_{t+s,\downarrow}v_{t+s,\downarrow}) * \chi_<)(x-y) \times
	\\ 
	&\quad \times c_{t,\uparrow}(g_x)^* d_{s,y,\uparrow} d_{t,x,\uparrow}^* c_{s,y,\uparrow} \;,
\\
	\mathrm{G} &:=-(8\pi a)^2 \int_0^\infty \di t \int_0^\infty \di s  \, e^{-2(t+s)\varepsilon}\int \di x \di y\, (\chi_< \ast ((u_{t+s,\uparrow} * g) v_{t+s,\uparrow}) * \chi_<)(x-y) \times
	\\ 
	&\quad \times c_{t,x,\downarrow}^* d_{s,y,\downarrow} d_{t,x,\downarrow}^* c_{s,y,\downarrow} \;,
\\
	\mathrm{H} &:=-(8\pi a)^2 \int_0^\infty \di t \int_0^\infty \di s  \, e^{-2(t+s)\varepsilon}\int \di x \di y\, (\chi_< \ast (u_{t+s,\uparrow} (v_{t+s,\uparrow} * g)) * \chi_<)(x-y)\times
	\\ 
	&\quad \times c_{t,x,\downarrow}^* d_{s,y,\downarrow} d_{t,x,\downarrow}^* c_{s,y,\downarrow} \;.
\end{align*}

All the terms  can be estimated in the same way. For instance, $\mathrm{E}$ can be bounded {with the aid of Cauchy--Schwarz and using $\|\chi_< \ast (u_{t+s,\downarrow}v_{t+s, \downarrow}) \ast \chi_<\|_1 \leq \| \chi_<\|_1^2 \|u_{t+s,\downarrow}v_{t+s, \downarrow}\|_1 \leq \|\chi_<\|_1^2 \|\hat{u}_{t+s, \downarrow}\|_2 \|\hat{v}_{t+s,\downarrow}\|_2$} together with \Cref{lem: a*a with N,lem:v bounds} as
\begin{align*}
	|\langle \psi, \mathrm{E}\psi\rangle| &\leq C\|\hat{g}\|_\infty \, \|(\chi_< \ast \chi_<)\|_1 \, \langle \psi, \mathcal{N}\psi\rangle \times 
	\\ 
	&\quad \times \left(\int_0^\infty \di t \int_0^\infty \di s \, e^{-2(t+s)\varepsilon} e^{(t+s)(k_{\F}^\uparrow)^2} \|\hat{u}_{t,\uparrow}\|_2\|\hat{u}_{s,\uparrow}\|_2\|\hat{u}_{t+s, \downarrow}\|_2 \|\hat{v}_{t+s, \downarrow}\|_2\right)
	\\ 
	&\leq C\|\hat{g}\|_\infty \langle \psi, \mathcal{N}\psi\rangle \left(\int_0^\infty \di t  \int_0^\infty \di s  \,e^{-2(t+s)\varepsilon} e^{2(t+s)(k_{\F}^\uparrow)^2}\|\hat{u}_{t,\uparrow}\|_2^2 \|\hat{u}_{s,\uparrow}\|_2^2\right)^{\frac{1}{2}}\times
	\\ 
	&\quad \times \left(\int_0^\infty \di t  \int_0^\infty \di s  \,e^{-2(t+s)\varepsilon} \|\hat{u}_{t+s,\downarrow}\|_2^2 \|\hat{v}_{t+s,\downarrow}\|_2^2\right)^{\frac 12} \;,
\end{align*}
where we used $\|(\chi_< \ast \chi_<)\|_1 \leq C\|\chi_<\|_1^2 \leq C$ (see~\eqref{eq:norms chi}). With Lemma~\ref{lem:t} and Lemma~\ref{lem:ts}, we conclude that
\[
	|\langle \psi, (\mathrm{I}_4 + \mathrm{I}_7) \psi\rangle|
    \leq C_\kappa \|\hat{g}\|_\infty\rho^{\frac{1}{2} -\kappa}\langle \psi, \mathcal{N}\psi\rangle \;.
\]
Next we consider the term $\mathrm{I}_6$. As for the other terms, we can rewrite it in position space as $\mathrm{I}_6 = \mathrm{I}_{6;a} + \mathrm{I}_{6;b} + \mathrm{I}_{6;c} + \mathrm{I}_{6;d}$ with
\allowdisplaybreaks[4]
\begin{equation} \label{eq:I6_split}
\begin{aligned}
	\mathrm{I}_{6;a}
    &:= - \frac{(8\pi a)^2}{L^3}
        \int_0^\infty \di t  \int_0^\infty  \di s
        \sum_r \hat{v}_{t+s,\uparrow}(r)   \, e^{-2(t+s)\varepsilon} \int \di x \di y \di z \di z^\prime\, e^{ir\cdot (x-z)} \times
	\\ 
	&\qquad \times \chi_<(x-y)\chi_<(z-z^\prime)
    d_{t,\uparrow}(g_x)^* c_{t,y,\downarrow}^\ast d_{t,y,\downarrow}^\ast 
    d_{s,z',\downarrow} c_{s,z',\downarrow} d_{s,z,\uparrow} \;,
\\
	\mathrm{I}_{6;b}
    &:= - \frac{(8\pi a)^2}{L^3}
        \int_0^\infty \di t  \int_0^\infty  \di s
        \sum_r \hat{v}_{t+s,\uparrow}(r)\hat g(r) e^{-2(t+s)\varepsilon} \int \di x \di y \di z \di z^\prime\, e^{ir\cdot (x-z)} \times
	\\ 
	&\qquad \times \chi_<(x-y)\chi_<(z-z^\prime)
    d^\ast_{t,x,\uparrow} c_{t,y,\downarrow}^\ast d_{t,y,\downarrow}^\ast d_{s,z',\downarrow} c_{s,z',\downarrow} d_{s,z,\uparrow} \;,
\\
	\mathrm{I}_{6;c}
    &:= - \frac{(8\pi a)^2}{L^3}
        \int_0^\infty \di t  \int_0^\infty  \di s
        \sum_r \hat{v}_{t+s,\downarrow}(r)
        e^{-2(t+s)\varepsilon} \int \di x \di y \di z \di z^\prime\, e^{ir\cdot (x-z)}  \times
	\\ 
	&\qquad \times \chi_<(x-y)\chi_<(z-z^\prime)
    d^\ast_{t,x,\downarrow} c_{t,y,\uparrow}^\ast d_{t,\uparrow}(g_y)^\ast d_{s,z',\uparrow} c_{s,z',\uparrow} d_{s,z,\downarrow} \;,
\\
	\mathrm{I}_{6;d}
    &:= - \frac{(8\pi a)^2}{L^3}
        \int_0^\infty \di t  \int_0^\infty  \di s
        \sum_r \hat{v}_{t+s,\downarrow}(r)  e^{-2(t+s)\varepsilon} \int \di x \di y \di z \di z^\prime\, e^{ir\cdot (x-z)} \times
	\\ 
	&\qquad \times \chi_<(x-y)\chi_<(z-z^\prime)
    d^\ast_{t,x,\downarrow} c_{t,\uparrow}(g_y)^\ast d_{t,y,\uparrow}^\ast d_{s,z',\uparrow} c_{s,z',\uparrow} d_{s,z,\downarrow} \;. 
\end{aligned}
\end{equation}
\allowdisplaybreaks[0]

We start with $\mathrm{I}_{6;a}$. We use Lemma \ref{lem:a*a*a*aaa} together with $|\hat{v}_{t+s,\sigma}(r)|\leq e^{(t+s)(k_{\F}^\sigma)^2}$ to write  
\begin{align}\label{eq: est CS r I6}
	&|\langle \psi, \mathrm{I}_{6;a}\psi\rangle |
    \leq C \int_0^{\infty}\di t \int_0^\infty \di s  \, e^{-2(t+s)\varepsilon}
        \Bigg\vert \frac{1}{L^3} \sum_r \hat{v}_{t+s,\uparrow}(r) \times\nonumber
	\\ 
	& \qquad\times \Big\langle \int \di x e^{-ir\cdot x} b_{\downarrow}((\chi_<)_x, v_t, u_t) d_{t,\uparrow}(g_x) \psi,
        \int \di z e^{-ir\cdot z} b_{\downarrow}((\chi_<)_z, v_s, u_s) d_{s,z,\uparrow} \psi
        \Big\rangle \Bigg\vert \nonumber
	\\
  &\leq C \int_0^{\infty}\di t  \int_0^\infty \di s  \, e^{-2(t+s)\varepsilon}e^{(t+s)(k_{\F}^\uparrow)^2}
  \sup_ x \left\| b_{\downarrow}((\chi_<)_x, v_t, u_t)\right\|
  \sup_z \left\|b_{\downarrow}((\chi_<)_z, v_s, u_s)\right\| \times \nonumber
  \\ 
  &\quad \times
  \sqrt{\int \di x \, \|d_{t,\uparrow}(g_x)\psi\|^2}
  \sqrt{\int \di z\, \|d_{s,z,\uparrow}\psi\|^2} \;.
\end{align}
Using then \Cref{lem: est b operator,lem: a*a with N} we find
\begin{align*}
	&|\langle \psi, \mathrm{I}_{6;a}\psi\rangle| 
	\\ 
	&\leq C \norm{\chi_<}_1^2\|\hat{g}\|_\infty \left(\int_0^\infty \di t \int_0^\infty \di s \, e^{-2(t+s)\varepsilon} \|\hat{u}_{t,\downarrow}\|_2 \|\hat{v}_{t,\downarrow}\|_2 \|\hat{u}_{s,\downarrow}\|_2 \|\hat{v}_{s,\downarrow}\|_2 \right)\langle \psi, \mathcal{N}\psi\rangle  
	\\ 
	&\leq C \norm{\chi_<}_1^2 \|\hat{g}\|_\infty 
    \left(\int_0^\infty \di t \, e^{-2t\varepsilon} \|\hat{v}_{t,\downarrow}\|_2 \|\hat{u}_{t,\downarrow}\|_2\right)^2 \langle \psi, \mathcal{N}\psi\rangle  \\
&    \leq C_\kappa \rho^{\frac{5}{9} - \kappa}
        \langle \psi, \mathcal{N}\psi\rangle \;,
\end{align*}
where in the last estimate we used \Cref{lem: bounds phi,lem:t} (see \eqref{eq: est int t uv}).
The terms $\mathrm{I}_{6;b}$, $\mathrm{I}_{6;c}$ and $\mathrm{I}_{6;d}$ can be estimated in the same way, using also \Cref{lem:v bounds}. Therefore
\[
	|\langle \psi, \mathrm{I}_6\psi\rangle|
    \leq C_\kappa \|\hat{g}\|_\infty\rho^{\frac{5}{9} - \kappa}\langle \psi, \mathcal{N}\psi\rangle \;.
\]

We now consider $\mathrm{I}_8$, which in position space we can write as $\mathrm{I}_8 = \mathrm{I}_{8;a} + \mathrm{I}_{8;b}$ with
\allowdisplaybreaks[4]
\begin{align*}
    \mathrm{I}_{8;a} &:= (8\pi a)^2 \int_0^\infty \di t  \int_0^\infty \di s \, e^{-2(t+s)\varepsilon}\int \di x \di y \di z \di z^\prime \, \chi_<(x-y)\chi_<(z-z^\prime)\times
    \\ 
    &\quad \times (u_{t+s, \uparrow}\ast g) (x;z)u_{t+s, \downarrow}(y;z^\prime) c^\ast_{t,x,\uparrow} c^\ast_{t,y,\downarrow} c_{s,z',\downarrow} c_{s,z,\uparrow} \;,
\\
    \mathrm{I}_{8;b} &:= (8\pi a)^2 \int_0^\infty \di t  \int_0^\infty \di s \, e^{-2(t+s)\varepsilon}\int \di x \di y \di z \di z^\prime \, \chi_<(x-y)\chi_<(z-z^\prime)\times
    \\ 
    &\quad \times u_{t+s, \uparrow}(x;z)u_{t+s, \downarrow}(y;z^\prime)
    c^\ast_{t,\uparrow}(g_x) c^\ast_{t,y,\downarrow} c_{s,z',\downarrow} c_{s,z,\uparrow} \;.
\end{align*}
\allowdisplaybreaks[0]
Both terms can be estimated via Lemma~\ref{lem:a*a*aa}, yielding
\begin{align*}
    &|\langle \psi, \mathrm{I}_8\psi\rangle| \leq C\|\hat{g}\|_\infty\|\chi_<\|_1^2 { \bigg(\int_0^\infty \di t \int_0^\infty \di s  \, e^{-2(t+s) \varepsilon}\sqrt{\int \d y\, \|c_{t,y,\downarrow}\psi\|^2}\sqrt{\int \d z\, \|c_{s,z,\uparrow}\psi\|^2}\times}
    \\ 
    &\hspace{12em} \times \|\hat{u}_{t+s, \uparrow}\|_2 \|\hat{u}_{t+s, \downarrow}\|_2 \|\hat{v}_{t,\uparrow}\|_2 \|\hat{v}_{s,\downarrow}\|_2\bigg)
    \\
   & \leq C\|\hat{g}\|_\infty\|\chi_<\|_1^2 \langle \psi, \mathcal{N}\psi\rangle\times
    \\ 
    &\qquad \times \left(\int_0^\infty \di t \int_0^\infty \di s  \, e^{-2(t+s) \varepsilon} e^{t(\kF^\downarrow)^2}e^{s(\kF^\uparrow)^2} \|\hat{u}_{t+s, \uparrow}\|_2 \|\hat{u}_{t+s, \downarrow}\|_2 \|\hat{v}_{t,\uparrow}\|_2 \|\hat{v}_{s,\downarrow}\|_2\right)
    \\ 
    &\leq  C \|\hat{g}\|_\infty\langle \psi, \mathcal{N}\psi\rangle\left(\int_0^\infty \! \di t \int_0^\infty \! \di s \, e^{-2(t+s)\varepsilon} e^{2(t+s)(\kF^\uparrow)^2} e^{2(t+s)(\kF^\downarrow)^2}\|\hat{u}_{t+s,\uparrow}\|_2^2 \|\hat{u}_{t+s, \downarrow}\|_2^2\right)^{\frac{1}{2}}\times
    \\ 
    &\qquad \times \left(\int_0^\infty \di t \int_0^\infty \di s \, e^{-2(t+s)\varepsilon} e^{-2t(\kF^\uparrow)^2}e^{-2s(\kF^\downarrow)^2} \|\hat{v}_{t,\uparrow}\|_2^2 \|\hat{v}_{s,\downarrow}\|_2^2\right)^{\frac{1}{2}} 
    \\ 
    &\leq C_\kappa \|\hat{g}\|_\infty \rho^{\frac{5}{9} - \kappa}\langle \psi, \mathcal{N}\psi\rangle \;,
\end{align*} 
where we used Lemma \ref{lem:t} and Lemma \ref{lem:ts}.

To conclude the analysis, we consider $\mathrm{I}_9$. We write it as a sum of four terms, $\mathrm{I}_9 = \mathrm{I}_{9;a} + \mathrm{I}_{9;b} + \mathrm{I}_{9;c} + \mathrm{I}_{9;d}$, with 
\begin{align*}
    \mathrm{I}_{9;a} &:= - \frac{(8\pi a)^2}{L^3}
        \int_0^\infty \di t  \int_0^\infty \di s
        \sum_r \hat{u}_{t+s, \uparrow}(r) \hat{g}(r) 
        e^{-2(t+s)\varepsilon} \int \di x \di y \di z \di z^\prime\, e^{ir\cdot (x-z)}\times
    \\ 
    &\qquad \times \chi_<(x-y)\chi_<(z-z^\prime)
    c_{t,x,\uparrow}^\ast c^\ast_{t,y,\downarrow} d_{t,y,\downarrow}^\ast d_{s,z',\downarrow} c_{s,z',\downarrow} c_{s,z,\uparrow} \;,
\\
    \mathrm{I}_{9;b} &:= - \frac{(8\pi a)^2}{L^3}
        \int_0^\infty \di t  \int_0^\infty \di s
        \sum_r \hat{u}_{t+s, \uparrow}(r) 
        e^{-2(t+s)\varepsilon} \int \di x \di y \di z \di z^\prime\, e^{ir\cdot (x-z)}\times
    \\ 
    &\qquad \times \chi_<(x-y)\chi_<(z-z^\prime)
    c_{t,\uparrow}^\ast(g_x) c^\ast_{t,y,\downarrow} d_{t,y,\downarrow}^\ast d_{s,z',\downarrow} c_{s,z',\downarrow} c_{s,z,\uparrow} \;,
\\
    \mathrm{I}_{9;c} &:= - \frac{(8\pi a)^2}{L^3}
        \int_0^\infty \di t  \int_0^\infty \di s
        \sum_r \hat{u}_{t+s, \downarrow}(r)
        e^{-2(t+s)\varepsilon} \int \di x \di y \di z \di z^\prime\, e^{ir\cdot (x-z)}\times
    \\ 
    &\qquad \times \chi_<(x-y)\chi_<(z-z^\prime)
    c_{t,x,\downarrow}^\ast c^\ast_{t,y,\uparrow} d_{t,\uparrow}^\ast(g_y) d_{s,z',\uparrow} c_{s,z',\uparrow} c_{s,z,\downarrow} \;,
\\
    \mathrm{I}_{9;d} &:= -\frac{(8\pi a)^2}{L^3}
        \int_0^\infty \di t  \int_0^\infty \di s
        \sum_r \hat{u}_{t+s, \downarrow}(r) 
        e^{-2(t+s)\varepsilon}  \int \di x \di y \di z \di z^\prime\, e^{ir\cdot (x-z)}\times
    \\ 
    &\qquad \times \chi_<(x-y)\chi_<(z-z^\prime)
    c_{t,x,\downarrow}^\ast c^\ast_{t,\uparrow}(g_y) d_{t,y,\uparrow}^\ast d_{s,z',\uparrow} c_{s,z',\uparrow} c_{s,z,\downarrow} \;.
\end{align*}
We first consider instance $\mathrm{I}_{9;a}$. The estimate is as the one for {the terms in} $\mathrm{I}_6$:
We first apply Lemma \ref{lem:a*a*a*aaa} with $|\hat{u}_{t+s, \uparrow}(r) \hat{g}(r)| \leq e^{-(t+s)(k_{\F}^\uparrow)^2}\|\hat{g}\|_\infty$, and~\Cref{lem: a*a with N} with $ \norm{\hat{v}_{s,\uparrow}}_\infty \leq e^{{s}(k_{\F}^\uparrow)^2} $. Then
\begin{align}\label{eq: est CS r I9}
|\langle \psi, \mathrm{I}_{9;a}\psi\rangle |&\leq C\|\hat{g}\|_\infty \int_0^{\infty}\di s  \int_0^\infty \di t  \, e^{-2(t+s)\varepsilon}e^{-(t+s)(k_{\F}^\uparrow)^2} {\sqrt{\int \d x \| c_{t,x,\uparrow}\psi\|^2}\times }\nonumber
\\ 
& \quad\times {\sqrt{\int \d z \|c_{s,z,\uparrow}\psi\|^2}}\sup_x \norm{b_\downarrow((\chi_<)_x, v_t, u_t)}
    \sup_z \norm{b_\downarrow((\chi_<)_z, v_s, u_s)}
     \\
&\leq C\|\hat{g}\|_\infty \int_0^{\infty}\di s  \int_0^\infty \di t  \, e^{-2(t+s)\varepsilon}e^{-(t+s)(k_{\F}^\uparrow)^2}\times\nonumber
\\ 
& \quad\times \sup_x \norm{b_\downarrow((\chi_<)_x, v_t, u_t)}
    \sup_z \norm{b_\downarrow((\chi_<)_z, v_s, u_s)}
    e^{(t+s) (k_{\F}^\uparrow)^2}
    \expect{\psi, \cN\psi} \nonumber \\
&\leq C \norm{\chi_<}_1^2 \norm{\hat g}_{\infty} 
        \left(\int_0^\infty \di t \, e^{-2t\eps} \norm{\hat v_{t,\downarrow}}_2\norm{\hat u_{t,\downarrow}}_2\right)^2
        \expect{\psi, \cN\psi}\;. \nonumber
\end{align}
The estimates on $\mathrm{I}_{9;b}$ through $\mathrm{I}_{9;d}$ are analogous. Bounding the $t$-integral via Cauchy--Schwarz and~\Cref{lem:t} results in
\begin{equation*}
    \abs{\expect{\psi, \mathrm{I}_9 \psi}} 
    \leq C_\kappa {\|\hat{g}\|_\infty}\rho^{\frac 59 - \kappa} \expect{\psi,\cN\psi} \;.
\end{equation*}
This concludes the proof. 
\end{proof}

\subsection{Analysis of \texorpdfstring{$[[n_g , B_1 - B_1^\ast], B_1 - B_1^\ast]$}{[[ng,B1-B1*],B1-B1*]}}
In this section, we prove that {the double commutator $[[n_g , B_1 - B_1^\ast], B_1 - B_1^\ast] $ gives rise only to} error terms.

\begin{proposition}\label{pro:B1}
Let $n_g $ and $B_1$ as in \eqref{eq:ng} and \eqref{eq: def T1 unitary}.
For any $\psi\in \mathcal{F}_{\mathrm{f}}$, we have 
\[
    |\langle \psi, [[n_g , B_1 - B_1^\ast], B_1 - B_1^\ast] \psi \rangle| \leq C \|\hat{g}\|_\infty \rho^{\frac{5}{3} + \frac{1}{9}} + C{L^{-3}\left(\rho^{\frac{7}{9}}\|\hat{g}\|_\infty  + \rho^{\frac{4}{9}} \|\hat{g}\|_2\right)} \langle \psi, \mathcal{N}\psi\rangle \;.
\]
\end{proposition}

\begin{proof}
We estimate only $ L^3 [[n_{g,\uparrow}, B_1 - B_1^\ast], B_1 - B_1^\ast]$; the error terms in $  [[n_{g,\downarrow}, B_1 - B_1^\ast], B_1 - B_1^\ast]$ can be treated the same way.
We have
\[
    [[n_{g,\uparrow}, B_1 - B_1^\ast], B_1 - B_1^\ast] = \frac{(2\pi)^3}{L^3}\sum_{j=1}^9 \mathrm{I}_j \;.
\]
with the terms  $\I_1, \ldots, \I_9$ being of the same form as those in~\eqref{eq:I1I9_expansion}, but with $\hat{\eta}_{\cdot, \cdot}(\cdot) \chi_<(\cdot)$ replaced by $\hat{\varphi}(\cdot) \chi_>(\cdot) = \hat{\varphi}^>(\cdot)$. All the terms are errors that have to be estimated.
The first term is the constant contribution:
\begin{align} \label{eq:I1_estimate_B1B1}
	\mathrm{I}_1 &= \frac{1}{L^6}\sum_{p,r,r^\prime} \left(\hat{g}(r+p) + \hat{g}(-r) \right) |\hat{\varphi}(p)|^2 \widehat{\chi}_>(p)^2 \hat{u}_\uparrow(r+p)\hat{u}_\downarrow(r^\prime - p)\hat{v}_\uparrow(r)\hat{v}_\downarrow(r^\prime) \;.
\end{align}
This can be easily estimated by using that $|\hat{\varphi}(p)|\leq C/|p|^2$ and $ |p| > C \rho^{\frac 29} $, which gives 
\[
	|\mathrm{I}_1 | \leq \frac{C\|\hat{g}\|_\infty}{L^6}\sum_{p,r,r^\prime}\frac{\widehat{\chi}_>(p)^2}{|p|^4}\hat{v}_\uparrow(r)\hat{v}_\downarrow(r^\prime)\leq C\|\hat{g}\|_\infty\rho^2\sum_p\frac{\widehat{\chi}_>(p)^2}{|p|^4}\leq C\|\hat{g}\|_\infty L^3\rho^{\frac{5}{3} + \frac{1}{9}} \;.
\]
Next, we consider 
\begin{align*}
\mathrm{I}_2 & = -\frac{1}{L^6}\sum_{\sigma\neq\sigma^\prime} \sum_{p,r,r^\prime}\left(\delta_{\sigma,\uparrow}(\hat{g}(r+p) + \hat{g}(-r)) + \delta_{\sigma,\downarrow}(\hat{g}(r^\prime - p) + \hat{g}(-r^\prime))\right) \hat{\varphi}(p)^2 \widehat{\chi}_>(p)^2 \times
\\ 
&\qquad \times \hat{u}_{\sigma^\prime}(r^\prime - p)\hat{v}_\sigma(r)\hat{v}_{\sigma^\prime}(r^\prime)\hat{d}^\ast_{r+p, \sigma}\hat{d}_{r+p,\sigma} \;.
\end{align*}
Contrary to the bound for the $B_2$-$B_2$ commutator term above it is not useful to combine this term with the analogue of the term $\mathrm{I}_5$ by means of anti-normal ordering since $\hat u_\sigma$ is not bounded in $\ell^2$-norm.
Instead, we use that $\Vert \hat\varphi^> \Vert_\infty^2 \le \Vert \varphi^> \Vert_1^2 \le C \rho^{-\frac 89}$ by Lemma~\ref{lem: bounds phi}, and we estimate
\begin{equation}
\begin{aligned} \label{eq:I2_estimate_B1B1}
	|\langle \psi, \mathrm{I}_2\psi\rangle| &\leq \frac{C\|\hat{g}\|_\infty }{L^6}\sum_{\sigma\neq \sigma^\prime}
    \sum_{p,r,r^\prime}
    \rho^{-\frac 89}
    \hat{v}_\sigma(r)\hat{v}_{\sigma^\prime}(r^\prime)\|\hat{d}_{r+p,\sigma}\psi\|^2 
    \leq C\|\hat{g}\|_\infty \rho^{2 - \frac 89}\sum_{p}\|\hat{a}_{p,\sigma}\psi\|^2 
	\\ 
	&\leq C\|\hat{g}\|_\infty \rho^{\frac{10}{9}}\langle \psi, \mathcal{N}\psi\rangle \;.
\end{aligned}
\end{equation}
Next, we split $\mathrm{I}_3$ as in the proof of Proposition~\ref{pro:B2}:
\begin{align*}
	\mathrm{I}_3 &= \frac{1}{L^6} \sum_{p,k,r,r^\prime}\left(\hat{g}(r+p) + \hat{g}(-r) \right) \hat{\varphi}^>(p) \hat{\varphi}^>(k)
    \hat{v}_\uparrow(r)\hat{v}_\downarrow(r^\prime)
        \hat{d}_{r+p,\uparrow}^\ast
        \hat{d}_{r^\prime - p, \downarrow}^\ast
        \hat{d}_{r^\prime - k, \downarrow}
        \hat{d}_{r+k,\uparrow}
	\\ 
	&= \mathrm{I}_{3;a} + \mathrm{I}_{3;b} \;.
\end{align*}
Consider
\begin{align*}
	&\mathrm{I}_{3;a} = \int \di x \di y \di z \di z^\prime\, \varphi^>(x-y)\varphi^>(z-z^\prime) v_{\uparrow}(x;z) v_{\downarrow}(y;z^\prime)
        d_\uparrow^\ast(g_x) 
        d^\ast_{y,\downarrow} 
        d_{z',\downarrow}
        d_{z,\uparrow} 
	\\ 
	&= -\int \di x \di z^\prime  d_\uparrow^\ast(g_x) 
    \left(\int \di y\, \varphi^>(x-y) v_\downarrow(y;z^\prime) d_{y,\downarrow}^\ast\right)
    \left(\int \di z\, \varphi^>(z-z^\prime) v_\uparrow(x;z) d_{z,\uparrow}\right)d_{z',\downarrow} \,. 
\end{align*}
Since $ d^*_{y,\sigma} = a_\sigma(u_y) $ with $0\leq \hat{u}_{\sigma}\leq 1$, we deduce that
\begin{equation} \label{eq:intoperator_estimate}
\begin{aligned}
    \left\|\int \di y \,\varphi^>(x-y) v_\downarrow(y;z^\prime) d_{y,\downarrow} \right\| 
    &= \|a( (\varphi^>_x v_{z',\downarrow}) * u_\downarrow)\|
    \leq \| \varphi^>_x v_{z',\downarrow}  \|_2
        \| \hat{u}_\downarrow \|_\infty
    \leq \|\varphi^>_x v_{z^\prime,\downarrow}\|_2 \;, \\
    \left\|\int \di z \,\varphi^>(z-z^\prime) v_\uparrow(x;z) d_{z,\uparrow} \right\| 
    &\leq \|\varphi^>_{z^\prime} v_{x,\uparrow}\|_2 \;. 
\end{aligned}
\end{equation}
With the aid of the two bounds above, we get 
\begin{align*}
	|\langle \psi, \mathrm{I}_{3;a}\psi\rangle| &\leq C\left(\int \di x \di z^\prime \, \|\varphi_x^>v_{z^\prime,\downarrow}\|^2_2 \|d_\uparrow(g_x)\psi\|^2\right)^{\frac{1}{2}}
    \left(\int \di x \di z^\prime \, \|\varphi_{z^\prime}^>v_{x,\uparrow}\|^2_2 \|d_{z',\downarrow}\psi\|^2\right)^{\frac{1}{2}}
	\\ 
	&\leq C\|\hat{g}\|_\infty \|\varphi^>\|_2^2 \|\hat{v}_\downarrow\|_2 \|\hat{v}_\uparrow\|_2 \langle \psi, \mathcal{N}\psi\rangle \;,
\end{align*}
where we used Lemmas~\ref{lem:v bounds} and \ref{lem: a*a with N}. 
The term $\mathrm{I}_{3;b}$ can be estimated similarly, using the bound for $\|\hat{v}_\uparrow\ast \hat{g}\|_2$ in \Cref{lem:v bounds}.
By Lemma~\ref{lem: bounds phi} we conclude that
\[
	|\langle \psi, \mathrm{I}_{3}\psi\rangle| \leq C\|\hat{g}\|_\infty \rho^{\frac 79}\langle \psi, \mathcal{N}\psi\rangle \;.
\]
The next error term to estimate is
\begin{align*}
	\mathrm{I}_4 &= - \frac{1}{L^6}\sum_{\sigma\neq \sigma^\prime}\sum_{p,r, r^\prime} \left(\delta_{\sigma,\uparrow} (\hat{g}(r+p) + \hat{g}(-r)) + \delta_{\sigma,\downarrow}(\hat{g}(r^\prime - p) + \hat{g}(-r^\prime)) \right) \hat{\varphi}(p)^2 \widehat{\chi}_>(p)^2\times
	\\ 
	&\qquad \times \hat{u}_\sigma(r+p)\hat{u}_{\sigma^\prime}(r^\prime - p)\hat{v}_{\sigma^\prime}(r^\prime)\hat{c}_{-r,\sigma}^\ast \hat{c}_{-r,\sigma} \;.
\end{align*}
In this case, we have
\begin{align*}
	|\langle \psi, \mathrm{I}_4 \psi\rangle| &\leq \frac{C\|\hat{g}\|_\infty }{L^6}\sum_{\sigma\neq \sigma^\prime}\sum_{p,r,r^\prime} |\hat{\varphi}^>(p)|^2 \hat{v}_{\sigma^\prime}(r^\prime) \|\hat{c}_{-r,\sigma}\psi\|^2
    \leq C\rho\|\varphi^>\|_2^2 \langle \psi, \mathcal{N}\psi\rangle
	\\ 
	& \leq C\|\hat{g}\|_\infty \rho^{\frac 79}\langle \psi,\mathcal{N}\psi\rangle \;,
\end{align*}
where we used Lemmas~\ref{lem: bounds phi}.
The next term to estimate is
\begin{align*}
	\mathrm{I}_5 &= \frac{1}{L^6}\sum_{\sigma\neq \sigma^\prime}\sum_{p,k,r,r^\prime}\left(\delta_{\sigma,\uparrow} (\hat{g}(r+p) +\hat{g}(-r)) + \delta_{\sigma,\downarrow}(\hat{g}(r^\prime - p) + \hat{g}(-r^\prime))\right)\hat{\varphi}^>(p)\times
	\\ 
	&\quad \times \hat{\varphi}^>(k) \hat{u}_{\sigma^\prime}(r^\prime - p)\hat{v}_\sigma(r)
        \hat{d}_{r+p,\sigma}^\ast 
        \hat{c}^\ast_{-r^\prime, \sigma^\prime}
        \hat{c}_{-r^\prime + p - k, \sigma^\prime}
        \hat{d}_{r+k,\sigma}
    = \mathrm{I}_{5;a} + \mathrm{I}_{5;b} + \mathrm{I}_{5;c}  + \mathrm{I}_{5;d} \;.
\end{align*}
In order to estimate
\begin{equation*}
	\mathrm{I}_{5;a} 
    = \int \di x \di y \di z \di z^\prime\, \varphi^>(x-y)\varphi^>(z-z^\prime)u_{\downarrow}(y;z^\prime)v_\uparrow(x;z)
        d^\ast_{\uparrow}(g_x) 
        c_{y,\downarrow}^\ast 
        c_{z',\downarrow}
        d_{z,\uparrow} \;,
\end{equation*}
we use that $u_{\downarrow}(y;z^\prime) = \delta(y-z^\prime) - v_{\downarrow}(y;z^\prime)$; accordingly $\mathrm{I}_{5;a} = \mathrm{I}_{5;a;1} + \mathrm{I}_{5;a;2}$ with
\begin{equation*}
\begin{aligned}
	\mathrm{I}_{5;a;1}
    &:= \int \di x \di y \di z \, \varphi^>(x-y)\varphi^>(z-y)v_\uparrow(x;z)
        d^\ast_{\uparrow}(g_x)
        c_{y,\downarrow}^\ast 
        c_{y,\downarrow}
        d_{z,\uparrow} \;, \\
    \mathrm{I}_{5;a;2}
    &:= \int \di x \di y \di z \di z' \, \varphi^>(x-y)\varphi^>(z-z') v_\downarrow(y;z') v_\uparrow(x;z)
        d^\ast_{\uparrow}(g_x)
        c_{y,\downarrow}^\ast 
        c_{z',\downarrow}
        d_{z,\uparrow} \;. 
\end{aligned}
\end{equation*}
Now, for any $\psi\in\mathcal{F}_{\mathrm{f}}$, {by Lemma \ref{lem:v bounds}} we have 
\begin{align*}
	&|\langle \psi, \mathrm{I}_{5;a;1}\psi\rangle|
	\\
	& \leq C\int \di x \di y \di z\, |\varphi^>(x-y)| |\varphi^>(z-y)| |v_\uparrow(x;z)| \|\hat{v}_{\downarrow}\|_2^2  \|d_\uparrow(g_x)\psi\| \|d_{z,\uparrow}\psi\|
	\\ 
	&\leq C\rho^2\left(\int \di x \di y \di z\,  |\varphi^>(x-y)| |\varphi^>(z-y)|\|d_\uparrow(g_x)\psi\|^2\right)^{\frac{1}{2}}\times
	\\ 
	&\quad \times \left(\int \di x \di y \di z\,  |\varphi^>(x-y)| |\varphi^>(z-y)|\|d_{z,\uparrow}\psi\|^2\right)^{\frac{1}{2}}\\
    & \leq C\|\hat{g}\|_\infty \rho^{2}\|\varphi^>\|_1^2 \langle \psi, \mathcal{N}\psi\rangle \;,
\end{align*}
{where we used also Lemma \ref{lem: a*a with N}.}
By Lemma \ref{lem: bounds phi}, we conclude that 
\[
	|\langle \psi, \mathrm{I}_{5;a;1}\psi\rangle|\leq C\|\hat{g}\|_\infty \rho^{\frac{10}{9}}\langle \psi, \mathcal{N}\psi\rangle \;.
\]
To {estimate} $\mathrm{I}_{5;a;2}$, we apply \Cref{lem:a*a*aa}, followed by \Cref{lem:v bounds,lem: a*a with N,lem: bounds phi} to  write 
\[
	|\langle \psi, \mathrm{I}_{5;a;2}\psi\rangle| \leq C\|\hat{g}\|_\infty \rho\|\varphi^>\|_1^2 \|\hat{v}_{\uparrow}\|_2 \|\hat{v}_{\downarrow}\|_2 \langle \psi, \mathcal{N}\psi\rangle 
    \leq C \norm{\hat g}_\infty \rho^{\frac{10}{9}} \expect{\psi, \cN\psi} \;.
\]
For $\mathrm{I}_{5;b}$ and $\mathrm{I}_{5;d}$ we can proceed in a similar way, we omit the details. It remains $\mathrm{I}_{5;c}$. In this case we write $(u_{\downarrow}\ast g) (y;z^\prime) = g(y;z^\prime) - (v_{\downarrow} \ast g)(y;z^\prime)$ and correspondingly define
\begin{equation*}
\begin{aligned}
	\mathrm{I}_{5;c;1}
    &:= \int \di x \di y \di z \di z' \, \varphi^>(x-y)\varphi^>(z-y)g(y;z^\prime)v_\uparrow(x;z)
        d^\ast_{x,\uparrow}
        c_{y,\downarrow}^\ast 
        c_{z^\prime,\downarrow}
        d_{z,\uparrow} \;, \\
    \mathrm{I}_{5;c;2}
    &:= -\int \di x \di y \di z \di z' \, \varphi^>(x-y)\varphi^>(z-z') (v_\downarrow\ast g)(y;z') v_\uparrow(x;z)
        d^\ast_{x,\uparrow}
        c_{y,\downarrow}^\ast 
        c_{z',\downarrow}
        d_{z,\uparrow} \;. 
\end{aligned}
\end{equation*}
The two terms can both be estimated similarly as $\mathrm{I}_{5;a;2}$ above with the aid of Lemma \ref{lem:a*a*aa}. One obtains
\[
	|\langle \psi,\mathrm{I}_{5;c;1}\psi\rangle| \leq C\|\varphi^>\|_1^2 \|\hat{g}\|_2 \|\hat{v}_\uparrow\|_2 \|\hat{v}_{\downarrow}\|_2^2 \langle \psi, \mathcal{N}\psi\rangle\leq C\rho^{\frac{11}{18}}\|\hat{g}\|_2 \langle \psi,\mathcal{N}\psi\rangle \;,
\]
and (bounding $\norm{g*v_\uparrow}_2 \leq \norm{\hat g}_\infty \norm{\hat v_\uparrow}_2$)
\[
	|\langle \psi,\mathrm{I}_{5;c;2}\psi\rangle| 
    \leq C\|\varphi^>\|_1^2 \|\hat{g}\|_\infty \|\hat{v}_\uparrow\|_2^2 \|\hat{v}_{\downarrow}\|_2^2 \langle \psi, \mathcal{N}\psi\rangle
    \leq C\|\hat{g}\|_\infty \rho^{\frac{10}{9}}\langle \psi, \mathcal{N}\psi\rangle \;. 
\]
Thus
\[
	|\langle \psi, \mathrm{I}_5\psi\rangle| \leq C\|\hat{g}\|_\infty\rho^{\frac{10}{9}}\langle \psi, \mathcal{N}\psi\rangle + C\rho^{\frac{11}{18}}\|\hat{g}\|_2 \langle \psi, \mathcal{N}\psi\rangle \;. 
\]
Next, we estimate $\mathrm{I}_6$, which we split as in~\eqref{eq:I6_split} into
\begin{align*}
	\mathrm{I}_6 &= - \frac{1}{L^6} \sum_{\sigma \neq \sigma^\prime} \sum_{p,k,r,r^\prime, s}\left(\delta_{\sigma,\uparrow} (\hat{g}(r+p) + \hat{g}(-r)) + \delta_{\sigma,\downarrow}(\hat{g}(r^\prime -p) + \hat{g}(-r^\prime))\right)\hat{\varphi}^>(p) \hat{\varphi}^>(k) \times
	\\ 
	&\qquad \times
    \hat{v}_\sigma(r) 
        \hat{d}_{r+p,\sigma}^\ast 
        \hat{c}_{-r^\prime, \sigma^\prime}^\ast
        \hat{d}^\ast_{r^\prime - p, \sigma^\prime}
        \hat{d}_{s-k,\sigma^\prime}
        \hat{c}_{-s, \sigma^\prime}
        \hat{d}_{r+k,\sigma}
    = \mathrm{I}_{6;a} + \mathrm{I}_{6;b} + \mathrm{I}_{6;c} + \mathrm{I}_{6;d} \;.
\end{align*}
Note that using the constraints on $p$ and $r^\prime$ we can replace 
$\hat d_{r'-p, \sigma'} = \hat u_{\sigma'}(r'-p) \hat a_{r'-p,\sigma'}$ by 
$\hat u_{\sigma'}^>(r'-p) \hat a_{r'-p,\sigma'}$ and we can proceed similarly with $\hat d_{s-k,\sigma'}$ using the constraints on $k$ and $s$.
(with $\hat{u}^>_\sigma$ defined as in \eqref{eq: def u>}).
The first error term written in position space is 
\begin{align*}
	\mathrm{I}_{6;a} &= - \frac{1}{L^3} \sum_r\hat{v}_\uparrow(r)\int \di x \di y \di z \di z^\prime\,e^{ir\cdot (x-z)} \varphi^>(x-y)\varphi^>(z-z^\prime)\times
	\\ 
	&\qquad \times 
        d^\ast_\uparrow(g_x)
        c_{y,\downarrow}^\ast 
        d_{y,\downarrow}^\ast
        d_{z',\downarrow}
        c_{z',\downarrow}
        d_{z,\uparrow} 
	\\ 
	& = - \frac{1}{L^3} \sum_r \hat{v}_\uparrow(r)\left(\int \di x \, e^{ir\cdot x} d^\ast_\uparrow(g_x) b^\ast_{\downarrow}(\varphi^>_x, v, u^>)\right)
    \left(\int \di z \, e^{-ir\cdot z} b_{\downarrow}(\varphi^>_z, v, u^>) d_{z,\uparrow} \right) \;.
\end{align*}
By Lemma \ref{lem:a*a*a*aaa} and using that $0\leq \hat{v}_\uparrow(r)\leq 1$, we get 
\begin{align*}
	|\langle \psi,\mathrm{I}_{6;a}\psi\rangle| 
	&\leq C \left(\sup_x \|b_{\downarrow}(\varphi^>_x, v, u^>) \|\right)^2\left(\int \di x  \|d_\uparrow(g_x)\psi\|^2\right)^{\frac{1}{2}}\left(\int \di z\|d_{z,\uparrow}\psi\|^2\right)^{\frac{1}{2}}
	\\ 
	&\leq C\|\hat{g}\|_\infty \rho^{\frac{7}{9}}\langle \psi, \mathcal{N}\psi\rangle \;,
\end{align*}
where in the last estimate we used \Cref{lem: est b operator,lem: a*a with N} (see \eqref{eq: est b op}). For the term $\mathrm{I}_{6;b}$ we can proceed in the same way using that $|\hat{v}_\sigma(r)\hat{g}(-r)|\leq \|\hat{g}\|_\infty$. For the terms $\mathrm{I}_{6;c}$ and $\mathrm{I}_{6;d}$ we proceed likewise, using {also \Cref{lem:v bounds} (see \eqref{eq: a conv L2 norm v u conv})}. Thus, we conclude that 
\[{}
	|\langle \psi, \mathrm{I}_6\psi\rangle| \leq C\|\hat{g}\|_\infty \rho^{\frac{7}{9}}\langle \psi, \mathcal{N}\psi\rangle \;.
\]
We next consider 
\begin{align*}
	\mathrm{I}_7 &= \frac{1}{L^6}\sum_{\sigma\neq \sigma^\prime}\sum_{p,r,r^\prime, s}\left(\delta_{\sigma,\uparrow} (\hat{g}(r+p) + \hat{g}(-r)) + \delta_{\sigma,\downarrow}(\hat{g}(r^\prime - p) + \hat{g}(-r^\prime)) \right) \hat{\varphi}^>(p)^2\times
	\\ 
	&\qquad \times \hat{u}_{\sigma^\prime}(r^\prime - p) \hat{v}_{\sigma^\prime}(r^\prime) 
    \hat{c}^{\ast}_{-r,\sigma}\hat{d}_{r+p,\sigma}^\ast \hat{d}_{s+p, \sigma}\hat{c}_{-s, \sigma} 
    = \mathrm{I}_{7;a} + \mathrm{I}_{7;b} + \mathrm{I}_{7;c} + \mathrm{I}_{7;d} \;.
\end{align*}
We start with 
\begin{equation*}
    \mathrm{I}_{7;a} = \int \di x \di y \di z \di z^\prime\, \varphi^>(x-y)\varphi^>(z-z') v_{\downarrow}(y;z')u_{\downarrow}(y;z') c^\ast_{x,\uparrow} d_\uparrow^\ast(g_x) d_{z,\uparrow} c_{z,\uparrow} \;.
\end{equation*}
This term can be estimated as $\mathrm{I}_{5;a}$. The same is true for $\mathrm{I}_{7;b}$ and $\mathrm{I}_{7;d}$, which can be estimated as $\mathrm{I}_{5;b}$ and $\mathrm{I}_{5;d}$, respectively. Next, we consider
\begin{equation*}
    \mathrm{I}_{7;c} = \int \di x \di y \di z \di z^\prime\, \varphi^>(x-y)\varphi^>(z-z') v_{\uparrow}(y;z')(u_{\uparrow} * g)(y-z') c^\ast_{x,\downarrow} d_{x,\downarrow}^\ast d_{z,\downarrow} c_{z,\downarrow} \;.
\end{equation*}
This term can be estimated similarly as  $\mathrm{I}_{5;c}$ above, using that $(u_\sigma * g)(y) = g(y) - (v_\sigma * g)(y)$. We omit the details. All together, using \Cref{lem:a*a*aa,lem: a*a with N,lem:v bounds,lem: bounds phi}, we find
\[
	|\langle \psi, \mathrm{I}_{7}\psi\rangle| \leq C\|\hat{g}\|_\infty\rho^{\frac{10}{9}}\langle \psi, \mathcal{N}\psi\rangle + C\rho^{\frac{11}{18}}\|\hat{g}\|_2 \langle \psi, \mathcal{N}\psi\rangle \;.
\]
Next, we estimate
\begin{align*}
	\mathrm{I}_8 &= \frac{1}{L^6}\sum_{p,k,r,r^\prime,}\left(\hat{g}(r+p) + \hat{g}(-r) \right) \hat{\varphi}^>(p) \hat{\varphi}^>(k)\times
	\\
	&\quad \times \hat{u}_\uparrow(r+p)\hat{u}_\downarrow(r'-p)
        \hat{c}_{-r^\prime, \downarrow}^\ast
        \hat{c}_{-r, \uparrow}^\ast
        \hat{c}_{-r-p+k, \uparrow}
        \hat{c}_{-r^\prime +p-k,\downarrow}
    = \mathrm{I}_{8;a} + \mathrm{I}_{8;b} \;.
\end{align*}
We write $\mathrm{I}_{8;a}$ as
\[
	\mathrm{I}_{8;a} = \int \di x \di y \di z \di z^\prime\, \varphi^>(x-y)\varphi^>(z-z^\prime) (u_\uparrow * g)(x-z)u_\downarrow(y;z^\prime)
        c_{y,\downarrow}^\ast
        c^\ast_{x,\uparrow}
        c_{z,\uparrow}
        c_{z',\downarrow} \;.
\]
We decompose $u_\uparrow * g = g - v_\uparrow * g$ and correspondingly write $\mathrm{I}_{8;a} = \mathrm{I}_{8;a;1} + \mathrm{I}_{8;a;2}$. Moreover, in each of the terms we write  $u_\downarrow(y;z') = \delta(y-z') - v_\downarrow(y;z')$. We proceed to estimate
\begin{align*}
\mathrm{I}_{8;a;1}
	&= \int \di x \di y \di z\, \varphi^>(x-y)\varphi^>(z-y) g(x-z)
        c_{y,\downarrow}^\ast
        c^\ast_{x,\uparrow}
        c_{z,\uparrow}
        c_{y,\downarrow}
	\\
	&\quad - \int \di x \di y \di z \di z^\prime\, \varphi^>(x-y)\varphi^>(z-z^\prime) g(x-z)v_\downarrow(y;z^\prime)
        c_{y,\downarrow}^\ast
        c^\ast_{x,\uparrow}
        c_{z,\uparrow}
        c_{z',\downarrow} \;.
\end{align*}
With the aid of \Cref{lem: bounds phi,lem:v bounds} we find
\begin{align*}
	|\langle\psi,\mathrm{I}_{8;a;1}\psi\rangle|
    &\leq C\|\hat{v}_\uparrow\|_2^2 \int \di x \di y \di z \, |\varphi^>(x-y)||\varphi^>(z-y)| |g(x-z)| \|c_{y,\downarrow}\psi\|^2
	\\
	&\quad + \|\hat{v}_\uparrow\|_2^2 \left(\int \di x \di y \di z \di z^\prime\, |\varphi^>(x-y)||\varphi^>(z-z^\prime)| |g(x-z)|^2 \|c_{y,\downarrow}\psi\|^2\right)^{\frac{1}{2}} \times
	\\
	&\qquad \times \left(\int \di x \di y \di z \di z^\prime\, |\varphi^>(x-y)||\varphi^>(z-z^\prime)| |v_\downarrow(y;z')|^2 \|c_{z',\downarrow}\psi\|^2\right)^{\frac{1}{2}}
	\\
	&\leq C \|\varphi^>\|_1 \|\varphi^>\|_2 \|\hat{g}\|_2 \|\hat{v}_\uparrow\|_2^2 \langle \psi,\mathcal{N}\psi\rangle
    + C \|\varphi^>\|_1^2
    \|\hat{g}\|_2 \|\hat{v}_\uparrow\|_2^2 \|\hat{v}_\downarrow\|_2 \langle\psi,\mathcal{N}\psi\rangle
	\\
	&\leq C\rho^{\frac{4}{9}}\|\hat{g}\|_2\langle \psi, \mathcal{N}\psi\rangle \;.
\end{align*}
The term $\mathrm{I}_{8;a;2}$ is the same as $\mathrm{I}_{8;a;1}$ only with $g$ replaced by $-v_\uparrow * g$, so
\begin{equation*}
    \abs{\expect{\psi, \mathrm{I}_{8;a;2} \psi}}
    \leq C \rho^{\frac 49} \norm{v_\uparrow * g}_2 \expect{\psi , \cN \psi} \leq
    C \rho^{\frac{17}{18}} \norm{\hat{g}}_{\infty} \expect{\psi , \cN \psi} \;.
\end{equation*}
Likewise, we split
\[
	\mathrm{I}_{8;b} = \int \di x \di y \di z \di z^\prime\, \varphi^>(x-y)\varphi^>(z-z^\prime) u_\uparrow(x;z) u_\downarrow(y;z^\prime)
        c_{y,\downarrow}^\ast
        c_{\uparrow}^\ast(g_x)
        c_{z,\uparrow}
        c_{z',\downarrow}
\]
according to $u_\downarrow(y;z') = \delta(y-z') - v_\downarrow(y;z')$, as $\mathrm{I}_{8;b} = \mathrm{I}_{8;b;1} + \mathrm{I}_{8;b;2}$. Writing $u_\uparrow(x;z) = \delta(x-z) - v_\uparrow(x;z)$, with the aid of Lemmas~\Cref{lem:v bounds,lem: a*a with N,lem: bounds phi}, we get
\begin{align*}
	|\langle\psi,\mathrm{I}_{8;b;1}\psi\rangle|
    &\leq \|\hat{g}\|_\infty \|\hat{v}_\uparrow\|_2^2
        \int \di x \di y \, |\varphi^>(x-y)|^2 \|c_{y,\downarrow}\psi\|^2
	\\
	&\quad + \|\hat{g}\|_\infty \|\hat{v}_\uparrow\|_2^2
        \int \di x \di y \di z \, |\varphi^>(x-y)||\varphi^>(z-y)||v_\uparrow(x;z)| \|c_{y,\downarrow}\psi\|^2
	\\
	&\leq \|\hat{g}\|_\infty
        \|\varphi^>\|_2^2
        \|\hat{v}_\uparrow\|_2^2
        \langle \psi,\mathcal{N}\psi\rangle
    + \|\hat{g}\|_\infty
        \|\varphi^>\|_1 \|\varphi^>\|_2
        \|\hat{v}_\uparrow\|_2^3
        \langle\psi,\mathcal{N}\psi\rangle
	\\
	&\leq C\rho^{\frac 79}
        \|\hat{g}\|_ \infty
        \langle \psi, \mathcal{N}\psi\rangle \;.
\end{align*}
The term $\mathrm{I}_{8;b;2}$ is bounded similarly,
yielding
\begin{align*}
    |\langle \psi, \mathrm{I}_{8;b;2}\psi\rangle| &\leq \|\hat{g}\|_\infty\|\varphi^>\|_1\|\varphi^>\|_2 \|\hat{v}_\uparrow\|_2^2 \|\hat v_\downarrow \|_2 \langle \psi, \mathcal{N}\psi\rangle + \|\hat{g}\|_\infty \|\varphi^>\|_1^2 \|\hat{v}_\uparrow\|^3_2\|\hat{v}_\downarrow\|_2 \langle \psi,\mathcal{N}\psi\rangle
    \\
    &\leq C\|\hat{g}\|_\infty \rho^{\frac{17}{18}}\langle \psi,\mathcal{N}\psi\rangle \;,
\end{align*}
where we used also \Cref{lem:a*a*aa}.
Combining the estimates, we get
\begin{align*}
    \abs{\expect{\psi, \mathrm{I}_8 \psi}}
    \leq C \rho^{\frac 79} \norm{\hat{g}}_{\infty} \expect{\psi , \cN \psi}
    + C \rho^{\frac 49} \norm{\hat{g}}_2 \expect{\psi , \cN \psi} \;.
\end{align*}
To conclude the proof, we estimate $\mathrm{I}_9$.
\begin{align*}
	\mathrm{I}_9 &= -\frac{1}{L^6}\sum_{\sigma\neq \sigma^\prime}\sum_{p,k,r,r^\prime,s}\left( \delta_{\sigma,\uparrow} (\hat{g}(r+p) + \hat{g}(-r)) + \delta_{\sigma,\downarrow}(\hat{g}(r^\prime - p) + \hat{g}(-r^\prime)) \right)\hat{\varphi}^>(p) \hat{\varphi}^>(k)\times
	\\
	&\qquad \times \hat{u}_\sigma(r+p)
        \hat{c}_{-r,\sigma}^\ast
        \hat{c}_{-r^\prime, \sigma^\prime}^\ast
        \hat{d}_{r^\prime - p, \sigma^\prime}^\ast
        \hat{d}_{s-k, \sigma^\prime}
        \hat{c}_{-s, \sigma^\prime}
        \hat{c}_{-r-p+k, \sigma}
    = \mathrm{I}_{9;a} + \mathrm{I}_{9;b} + \mathrm{I}_{9;c} + \mathrm{I}_{9;d} \;.
\end{align*}
As for $\mathrm{I}_6$ above, we can replace
$\hat d_{r'-p, \sigma'} = \hat u_{\sigma'}(r'-p) \hat a_{r'-p,\sigma'}$ by
$\hat u_{\sigma'}^>(r'-p) \hat a_{r'-p,\sigma'}$ using the constraints on $p$ and $r^\prime$ and similarly with $\hat d_{s-k,\sigma'}$
(with $\hat{u}^>_\sigma$ defined as in \eqref{eq: def u>}) due to the constraints on $s$ and $k$.
We start with $\mathrm{I}_{9;a}$ which we write in position space as
\begin{align*}
	\mathrm{I}_{9;a} &= -\frac{1}{L^3} \sum_r \hat{u}_\uparrow(r)\hat{g}(r) \int \di x \di y \di z \di z^\prime\, e^{-ir\cdot (x-z)}\varphi^>(x-y)\varphi^>(z-z^\prime)\times
	\\
	&\quad \times c^\ast_{x,\uparrow} c^\ast_{y,\downarrow} d^\ast_{y, \downarrow}d_{z',\downarrow} c_{z',\downarrow} c_{z,\uparrow}
	\\
	&= -\frac{1}{L^3} \sum_r \hat{u}_{\uparrow}(r)\hat{g}(r) \int \di x \di z\, e^{ir\cdot (x-z)} c^\ast_{x,\uparrow} b_{\downarrow}^\ast(\varphi^>_x, v, u^>) b_{\downarrow}(\varphi^>_z, v, u^>) c_{z,\uparrow} \;.
\end{align*}
This term can be estimated as $\mathrm{I}_{6;a}$ using \Cref{lem:v bounds} and the fact that $|\hat{u}_\uparrow(k) \hat{g}(k)| \leq \|\hat{g}\|_\infty \|\hat{u}_\uparrow\|_\infty \leq \|\hat{g}\|_\infty$.
The terms $\mathrm{I}_{9;b}, \mathrm{I}_{9;c}$ and $\mathrm{I}_{9;d}$ are estimated similarly; thus
\[
	|\langle\psi, \mathrm{I}_{9}\psi\rangle| \leq C\|\hat{g}\|_\infty \rho^{\frac{7}{9}}\langle \psi, \mathcal{N}\psi\rangle \;.
\]
This concludes the proof. 
\end{proof}
\subsection{Analysis of \texorpdfstring{$[[n_g , B_1 - B_1^\ast], B_2 - B_2^\ast]$}{[[ng,B1-B1*],B2-B2*]}}
{The main goal of this section is to prove that $[[n_g , B_1 - B_1^\ast], B_2 - B_2^\ast]$ gives rise only to error terms.} In this section we use the notations introduced above and we prove the following proposition.
\begin{proposition}\label{pro:B1B2}
Let $B_1$ and $B_2$ as in \eqref{eq: def T1 unitary} and \eqref{eq: def T2 unitary}, respectively. 
Then, given $ \kappa > 0 $, there exists $ C_\kappa > 0 $ such that for any $\psi\in \mathcal{F}_{\mathrm{f}}$,
\[
	 |\langle \psi, [[n_g , B_1 - B_1^\ast], B_2 - B_2^\ast] \psi \rangle| \leq  C\|\hat{g}\|_\infty \rho^{\frac{5}{3} + \frac{1}{9}} + C_\kappa{L^{-3}}\|\hat{g}\|_\infty \rho^{\frac{5}{9} -\kappa}\langle \psi, \mathcal{N}\psi\rangle \;.
\]
\end{proposition}
\begin{proof}
Computing $[[n_g, B_1 - B_1^\ast], B_2 - B_2^\ast]$ as in~\eqref{eq:I1I9_expansion}, results in similar terms $\I_1, \ldots, \I_9$, but with one $\hat{\eta}_{\cdot,\cdot}(\cdot) \chi_<(\cdot)$ replaced by $\hat{\varphi}(\cdot) \chi_>(\cdot) = \hat{\varphi}^>(\cdot)$. All these terms are errors that have to be estimated. Throughout the proof, $\kappa>0$ is an arbitrary constant as small as needed. In what follows we write all of the terms $\I_1,\ldots,\I_9$ explicitly and estimate them. We start by considering the constant term
\begin{align*}
	\mathrm{I}_1 &= \frac{1}{L^6}\sum_{p,r,r^\prime} \left(\hat{g}(r+p) + \hat{g}(-r)\right)\hat{\varphi}(p)\hat{\eta}_{r,r^\prime}^\varepsilon(p)\widehat{\chi}_>(p)\widehat{\chi}_<(p) \hat{u}_\uparrow(r+p)\hat{u}_\downarrow(r^\prime - p)\hat{v}_\uparrow(r)\hat{v}_\downarrow(r^\prime) \;.
\end{align*}
We notice that since $p\in\mathrm{supp}(\widehat{\chi}_>)\cap \mathrm{supp}(\widehat{\chi}_<)$, we get that $4\rho^{2/9}\leq |p| \leq 5\rho^{2/9}$. For those values of $p$, we have 
\begin{equation}\label{eq: est phieta}
	|\hat{\varphi}(p)\hat{\eta}_{r,r^\prime}(p)| \leq C|p|^{-4} \;.
\end{equation}
Thus, similarly as in~\eqref{eq:I1_estimate_B1B1}, we can estimate 
\[
	|\mathrm{I}_1 | \leq \frac{C\|\hat{g}\|_\infty}{L^6}\sum_{4\rho^{\frac{2}{9}} \leq |p| \leq 5 \rho^{\frac{2}{9}}}\sum_{r,r^\prime}\frac{1}{|p|^4}\hat{v}_\uparrow(r)\hat{v}_\downarrow(r^\prime)\leq CL^3\|\hat{g}\|_\infty \rho^{\frac{5}{3} + \frac{1}{9}} \;.
\]
The next term is 
\begin{align*}
\mathrm{I}_2 = -\frac{1}{L^6}\sum_{\sigma\neq\sigma^\prime} \sum_{p,r,r^\prime} & \left(\delta_{\sigma,\uparrow} (\hat{g}(r+p) + \hat{g}(-r)) + \delta_{\sigma,\downarrow}(\hat{g}(r^\prime - p) + \hat{g}(-r^\prime)) \right)\times
\\ 
& \times \hat{\varphi}(p) \hat{\eta}_{r,r^\prime}(p)\widehat{\chi}_>(p)\widehat{\chi}_<(p)
\hat{u}_{\sigma^\prime}(r^\prime - p)\hat{v}_\sigma(r)\hat{v}_{\sigma^\prime}(r^\prime)
\hat{d}^\ast_{r+p, \sigma}\hat{d}_{r+p,\sigma} \;.
\end{align*}
In this case, we use again \eqref{eq: est phieta}, which implies that $|\hat{\varphi}(p)\hat\eta_{r,r'}(p)| \leq C|p|^{-4}\leq C\rho^{-\frac 89}$ for $p\in\mathrm{supp}(\widehat{\chi}_>)\cap \mathrm{supp}(\widehat{\chi}_<) $. As in~\eqref{eq:I2_estimate_B1B1}, we then write
\[
	|\langle \psi, \mathrm{I}_2\psi\rangle| 
    \leq \frac{C\|\hat{g}\|_\infty }{L^6}\sum_{\sigma\neq \sigma^\prime}\sum_{p,r,r^\prime} 
    \rho^{-\frac 89}\hat{v}_\sigma(r)\hat{v}_{\sigma^\prime}(r^\prime)\|\hat{d}_{r+p,\sigma}\psi\|^2 \leq C\|\hat{g}\|_\infty \rho^{\frac{10}{9}}\langle \psi, \mathcal{N}\psi\rangle \;.
\]
We now estimate the error term $\mathrm{I}_3$, given by
\begin{align*}
	\mathrm{I}_3 &= \frac{(8 \pi a)}{L^6}\int_0^{\infty}\di t \, e^{-2\varepsilon t} \sum_{p,k,r,r^\prime}\left(\hat{g}(r+p) + \hat{g}(-r)  \right) \hat{\varphi}^>(p) \widehat{\chi}_<(k)\times
	\\ 
	&\qquad \times \hat{v}_{t,\uparrow}(r)\hat{v}_{t,\downarrow}(r^\prime) 
        \hat{d}_{r+p,\uparrow}^\ast
        \hat{d}_{r^\prime - p, \downarrow}^\ast 
        \hat{d}_{t,r^\prime - k, \downarrow}
        \hat{d}_{t,r+k,\uparrow} \;.
\end{align*}
Note that in deriving this expression, we replaced $\hat{d}_{r+k,\uparrow}$ and $\hat{d}_{r'-k,\downarrow}$ by $\hat{d}_{r+k,\uparrow}^<$ and $\hat{d}_{r'-k,\downarrow}^<$, which is allowed due to the momentum constraints on $k,r,r^\prime$, and then applied the definition~\eqref{eq: def ut vt} of $ \hat{d}_{t,k,\sigma} $, as well as~\eqref{eq: formula eta t-integral}.
We then split {$\mathrm{I}_3$ as we did above in~\eqref{eq:I3_split}, i.\,e., } $ \mathrm{I}_3 = \mathrm{I}_{3;a} + \mathrm{I}_{3;b} $, and write $\mathrm{I}_{3;a}$ in position space
\begin{align*}
	\mathrm{I}_{3;a} 
    &= - (8\pi a)
        \int_0^\infty \di t \, e^{-2\varepsilon t}
        \int \di x \di y \di z \di z^\prime\, \varphi^>(x-y)\chi_<(z-z^\prime) v_{t,\uparrow}(x;z) v_{t,\downarrow}(y;z^\prime) \times
	\\ 
	&\qquad \times  
    d_\uparrow^\ast(g_x) d^\ast_{y,\downarrow} d_{t,z,\uparrow} d_{t,z',\downarrow} 
	\\ 
	&= - (8\pi a)
        \int_0^\infty \di t \, e^{-2\varepsilon t}
        \int \di x \di z^\prime d^\ast_{\uparrow}(g_x)
    \left(\int \di y\, \varphi^>(x-y) v_{t,\downarrow}(y;z^\prime) d_{y,\downarrow}^\ast \right)\times
	\\ 
	&\qquad \times \left(\int \di z\, \chi_<(z-z^\prime) v_{t,\uparrow}(x;z) d_{t,z,\uparrow}\right) d_{t,z',\downarrow} \;. 
\end{align*}
Since $0\leq \hat{u}_{\sigma}\leq 1$ and $0\leq \hat{u}_{t,\sigma}\leq e^{-t(k_{\F}^\sigma)^2}$, we can deduce as in~\eqref{eq:intoperator_estimate} that
\begin{align*}
	\left\|\int \di y \,\varphi^>(x-y) v_{t,\downarrow}(y;z^\prime) a_\downarrow(u_y)\right\| 
    & \leq \|\varphi^>_x v_{t,z^\prime,\downarrow}\|_2 \;,
\\
    \left\|\int \di z \,\chi_<(z-z^\prime) v_{t,\uparrow}(x;z) a_\uparrow(u_{t,z})\right\| 
    & \leq e^{-t(\kF^\uparrow)^2} \|\chi_{<,z^\prime} v_{t,x,\uparrow}\|_2 \;,
\end{align*}
where we used the same notations as above, in particular, $\chi_{<,z^\prime}(z)  = \chi_<(z^\prime - z)$. With the aid of Cauchy--Schwarz and of the two bounds above, we get 
\begin{align*}
	|\langle \psi, \mathrm{I}_{3;a}\psi\rangle| &\leq C\int_0^{\infty}\di t \, e^{-2\varepsilon t}\left(\int \di x \di z^\prime \, \|\varphi_x^>v_{t,z^\prime,\downarrow}\|^2_2 \|d_\uparrow(g_x)\psi\|^2\right)^{\frac{1}{2}}\times
	\\ 
	&\qquad \times \left(\int \di x \di z^\prime \, e^{-2t(\kF^\uparrow)^2}\|\chi_{<,z^\prime}v_{t,x,\uparrow}\|^2_2 \|d_{t,z',\downarrow} \psi\|^2\right)^{\frac{1}{2}}
	\\ 
	&\leq C\|\hat{g}\|_\infty \|\varphi^>\|_2\|\chi_<\|_2 \left(\int_0^{\infty}\di t \, e^{-2\varepsilon t}e^{-t(\kF^\uparrow)^2}e^{-t(\kF^\downarrow)^2}\|\hat{v}_{t,\downarrow}\|_2 \|\hat{v}_{t,\uparrow}\|_2 \right)\langle \psi, \mathcal{N}\psi\rangle
	\\ 
    & \leq C_\kappa \norm{\hat g}_\infty \rho^{\frac 59 - \kappa} \expect{\psi, \cN\psi} \;,
\end{align*}
where we used {\Cref{lem: bounds phi,lem:v bounds,lem:t}}. The term $\mathrm{I}_{3;b}$ is estimated similarly. {More precisely, we have 
\begin{align*}
	\mathrm{I}_{3;b} 
    &= - (8\pi a)
        \int_0^\infty \di t \, e^{-2\varepsilon t}
        \int \di x \di y \di z \di z^\prime\, \varphi^>(x-y)\chi_<(z-z^\prime) (v_{t,\uparrow}\ast g)(x;z) v_{t,\downarrow}(y;z^\prime)  \times
	\\ 
	&\qquad \times  
    d_{x,\uparrow}^\ast d^\ast_{y,\downarrow} d_{t,z,\uparrow} d_{t,z',\downarrow} \;.
\end{align*}
We can estimate $\mathrm{I}_{3;b}$ similarly as $\mathrm{I}_{3;a}$ getting $|\langle \psi, \mathrm{I}_{3;b}\psi\rangle| \leq C_\kappa \|\hat{g}\|_\infty \rho^{5/9 - \kappa} \langle \psi, \mathcal{N}\psi\rangle$ and conclude that
}
\[
	|\langle \psi, \mathrm{I}_{3}\psi\rangle|
    \leq C_\kappa \|\hat{g}\|_\infty \rho^{\frac{5}{9} -\kappa }\langle \psi, \mathcal{N}\psi\rangle \;.
\]
The next error term to estimate is
\begin{align*}
	\mathrm{I}_4 &= - \frac{1}{L^6}\sum_{\sigma\neq \sigma^\prime}\sum_{p,r, r^\prime} \left(\delta_{\sigma,\uparrow} (\hat{g}(r+p) + \hat{g}(-r)) + \delta_{\sigma,\downarrow}(\hat{g}(r^\prime - p)  + \hat{g}(-r^\prime)) \right)\times
	\\ 
	&\qquad \times \hat{\varphi}(p)\hat{\eta}_{r,r^\prime}(p)\widehat{\chi}_>(p)\widehat{\chi}_<(p)
    \hat{u}_\sigma(r+p)\hat{u}_{\sigma^\prime}(r^\prime - p)\hat{v}_{\sigma^\prime}(r^\prime)\hat{c}_{-r,\sigma}^\ast \hat{c}_{-r,\sigma} \;.
\end{align*}
 In this case, using \eqref{eq: est phieta} {for $p\in \supp\widehat{\chi}_< \cap\supp\widehat{\chi}_>$}, we find 
\begin{align*}
	|\langle \psi, \mathrm{I}_4 \psi\rangle| 
    &\leq \frac{C\|\hat{g}\|_\infty}{L^6}\sum_{\sigma\neq \sigma^\prime}\sum_{p,r,r^\prime}\frac{1}{|p|^4}\widehat{\chi}_>(p)\widehat{\chi}_<(p)\hat{v}_{\sigma^\prime}(r^\prime) \|\hat{c}_{-r,\sigma}\psi\|^2
    \leq C\|\hat{g}\|_\infty \rho^{\frac{7}{9}}\langle \psi,\mathcal{N}\psi\rangle \;.
\end{align*}
We now estimate $\mathrm{I}_5$, which reads as 
\begin{align*}
	\mathrm{I}_5 &= \frac{1}{L^6}\sum_{\sigma\neq \sigma^\prime}\sum_{p,k,r,r^\prime}\left(\delta_{\sigma,\uparrow} (\hat{g}(r+p) + \hat{g}(-r)) + \delta_{\sigma,\downarrow}(\hat{g}(r^\prime - p) + \hat{g}(-r^\prime)) \right)\hat{\varphi}^>(p)\times
	\\ 
	&\quad \times \hat{\eta}_{r,r^\prime - p +k}(k)\widehat{\chi}_<(k)
    \hat{u}_{\sigma^\prime}(r^\prime - p) \hat{v}_{\sigma}(r) 
        \hat{d}_{r+p,\sigma}^\ast 
        \hat{c}^\ast_{-r^\prime, \sigma^\prime}
        \hat{c}_{-r^\prime + p - k, \sigma^\prime}
        \hat{d}_{r+k,\sigma} \\
    &= \mathrm{I}_{5;a} + \mathrm{I}_{5;b} + \mathrm{I}_{5;c} + \mathrm{I}_{5;d} \;.
\end{align*}
We first consider $\mathrm{I}_{5;a}$, {which in position space reads as (see \eqref{eq: formula eta t-integral})}
\begin{align*}
	\mathrm{I}_{5;a} &= (8\pi a) \int_0^\infty \di t \, e^{-2\varepsilon t}\int \di x \di y \di z \di z^\prime\, \varphi^>(x-y)\chi_<(z-z^\prime)u_{t,\downarrow}(y;z^\prime)v_{t,\uparrow}(x;z)\times
	\\ 
	&\quad \times 
        d^\ast_{\uparrow}(g_x) 
        c_{y,\downarrow}^\ast 
        c_{t,z',\downarrow}
        d_{t,z,\uparrow} \;.
\end{align*}
Applying \Cref{lem:a*a*aa,lem: a*a with N,lem:v bounds,lem: bounds phi,lem:t} we get
\begin{align*}
	|\langle\psi,\mathrm{I}_{5;a}\psi\rangle| &\leq C\|\hat{g}\|_\infty \|\varphi^>\|_1 \|\chi_<\|_1  {\|\hat{v}_\downarrow\|_2}\times
	\\ 
	& \quad \times
    \left(\int_0^{\infty}\di t \,e^{-2t\varepsilon} e^{-t(k_{\F}^\uparrow)^2} \|\hat{u}_{t,\downarrow}\|_2 \|\hat{v}_{t,\uparrow}\|_2 \|\hat{v}_{t,\downarrow}\|_2 \right) \langle \psi, \mathcal{N}\psi\rangle
 	\\ 
 	&\leq C\|\hat{g}\|_\infty \rho\|\varphi^>\|_1 \|\chi_<\|_1 \left(\int_0^{\infty}\di t \,e^{-2t\varepsilon}  \|\hat{v}_{t,\downarrow}\|_2 \|\hat{u}_{t,\downarrow}\|_2\right)\langle \psi, \mathcal{N}\psi\rangle
 	\\ 
    &\leq C_\kappa \rho^{\frac{5}{6} - \kappa}\langle \psi, \mathcal{N}\psi\rangle \;.
\end{align*}

The other terms $\I_{5;b}, \I_{5;c}, \I_{5;d}$ can be estimated similarly, using the lemmas mentioned for $\mathrm{I}_{5;a}$. 
Thus
\[
	|\langle \psi, \mathrm{I}_5\psi\rangle| 
    \leq C_\kappa \|\hat{g}\|_\infty\rho^{\frac{5}{6} -\kappa}\langle \psi, \mathcal{N}\psi\rangle \;.
\]
Now, we estimate $\mathrm{I}_6$, which we split {in four terms similarly as for~\eqref{eq:I6_split}}
\begin{align*}
	\mathrm{I}_6 
    &= - \frac{1}{L^6}\sum_{\sigma\neq \sigma^\prime}\sum_{p,k,r,r^\prime, s}\left(\delta_{\sigma,\uparrow} (\hat{g}(r+p) + \hat{g}(-r)) + \delta_{\sigma,\downarrow}(\hat{g}(r^\prime -p) + \hat{g}(-r^\prime)) \right) \hat{\varphi}^>(p)\times
	\\ 
	&\qquad \times \hat{\eta}_{r,s}(k)\widehat{\chi}_<(k) 
    \hat{v}_\sigma(r) 
        \hat{d}_{r+p,\sigma}^\ast 
        \hat{c}_{-r^\prime, \sigma^\prime}^\ast
        \hat{d}^\ast_{r^\prime - p, \sigma^\prime}
        \hat{d}_{s-k,\sigma^\prime}
        \hat{c}_{-s, \sigma^\prime}
        \hat{d}_{r+k,\sigma}
	\\ 
	&= \mathrm{I}_{6;a} + \mathrm{I}_{6;b} + \mathrm{I}_{6;c} + \mathrm{I}_{6;d} \;.
\end{align*}
Note that using the constraints on $p,r^\prime$, we can replace 
$\hat d_{r'-p,\sigma'} = \hat{u}_{\sigma^\prime}(r^\prime -p) \hat a_{r'-p,\sigma'}$  with $\hat{u}^>_{\sigma^\prime}(r^\prime - p) \hat a_{r'-p,\sigma'}$, where $\hat{u}^>_\sigma$ is as defined in \eqref{eq: def u>}.
The first error term written in position space is
\begin{align*}
	\mathrm{I}_{6;a} 
    &= \frac{(8\pi a)}{L^3} \int_0^{\infty} \di t \, e^{-2\varepsilon t} 
    \sum_r \hat{v}_{t,\uparrow}(r)
    \left(\int \di x  \, e^{ir\cdot x} d^\ast_\uparrow(g_x) b^\ast_{\downarrow}(\varphi^>_x, v, u^>) \right)\times
	\\
	&\qquad \times \left(\int \di z \di z^\prime\, e^{-ir\cdot z}\chi_<(z-z^\prime) d_{t,z',\downarrow} c_{t,z',\downarrow} d_{t,z,\uparrow} \right) \;.
\end{align*}
Proceeding similarly as for the term $\mathrm{I}_{6;a}$ in Proposition \ref{pro:B2} and Proposition \ref{pro:B1}, and in particular applying  \Cref{lem:a*a*a*aaa,lem: a*a with N,lem: est b operator}, we get 
\[
	|\langle\psi,\mathrm{I}_{6;a}\psi\rangle | \leq C \|\hat{g}\|_\infty\|\chi_<\|_1\rho^{\frac{7}{18}}\left(\int_0^\infty \di t \, e^{-2\varepsilon t} \|\hat{v}_{t,\downarrow}\|_2 \|\hat{u}_{t,\downarrow}\|_2 \right)\langle \psi, \mathcal{N}\psi\rangle \;.
\]
Using then \Cref{lem: bounds phi}, and Lemma \ref{lem:t} to estimate the integral with respect to $t$ (see \eqref{eq: est int t uv}), we find 
\[
	|\langle \psi,\mathrm{I}_{6;a}\psi\rangle| 
    \leq C_\kappa \|\hat{g}\|_\infty \rho^{\frac{2}{3} - \kappa}\langle \psi,\mathcal{N}\psi\rangle \;.
\]
 {To estimate $\mathrm{I}_{6,b}$, $\mathrm{I}_{6;c}$ and $\mathrm{I}_{6;d}$, it suffices to argue as in the estimate of $\mathrm{I}_{6;a}$. More precisely, we apply the same lemmas as used for $\mathrm{I}_{6;a}$, together with Lemma \ref{lem:v bounds} for the terms $\mathrm{I}_{6;c}$ and $\mathrm{I}_{6;d}$ and the bound $|\hat{v}_{t,\uparrow}(r) \hat{g}(r)| \leq \|\hat{g}\|_\infty$ for the term $\mathrm{I}_{6;b}$. Hence, all the terms can be estimated in the same way, and we conclude that}
\[
	|\langle \psi, \mathrm{I}_6\psi\rangle| 
    \leq C_\kappa \|\hat{g}\|_\infty \rho^{\frac{2}{3} - \kappa}\langle \psi,\mathcal{N}\psi\rangle \;.
\]
We next consider
\begin{align*}
	\mathrm{I}_7 &= \frac{1}{L^6}\sum_{\sigma\neq \sigma^\prime}\sum_{p,r,r^\prime, s}\left(\delta_{\sigma,\uparrow} (\hat{g}(r+p) + \hat{g}(-r))  + \delta_{\sigma,\downarrow}(\hat{g}(r^\prime - p) + \hat{g}(-r^\prime)) \right) \hat{\varphi}^>(p)\times
	\\ 
	&\qquad \times \hat{\eta}_{s,r^\prime}^{\varepsilon}(p) \widehat{\chi}_<(p)
    \hat{u}_{\sigma^\prime}(r^\prime - p) \hat{v}_{\sigma^\prime}(r^\prime) 
    \hat{c}^{\ast}_{-r,\sigma}\hat{d}_{r+p,\sigma}^\ast \hat{d}_{s+p, \sigma}\hat{c}_{-s, \sigma} 
    = \mathrm{I}_{7;a} + \mathrm{I}_{7;b} + \mathrm{I}_{7;c} + \mathrm{I}_{7;d} \;.
\end{align*}
In order to estimate all of these terms, we write them in position space. We start with $\mathrm{I}_{7;a}$:
\begin{align*}
	\mathrm{I}_{7;a} 
    &= (8\pi a) \int_0^\infty \di t \, e^{-2\varepsilon t}\int \di x \di y \di z \di z'\, \varphi^>(x-y)\chi_<(z-z') v_{t,\downarrow}(y;z') u_{t,\downarrow}(y;z')\times
	\\ 
	&\quad \times c^\ast_{x,\uparrow} d_\uparrow^\ast(g_x) d_{t,z,\uparrow}c_{t,z,\uparrow} \;.
\end{align*}
We use Lemma \ref{lem:a*a*aa} together with \Cref{lem: a*a with N,lem:v bounds} and obtain:
\begin{align*}
	|\langle\psi, \mathrm{I}_{7;a}\psi\rangle| &\leq C\|\hat{g}\|_\infty \|\varphi^>\|_1 \| \chi_<\|_1 {\|\hat{v}_ \uparrow\|_2}\times
	\\ 
	&\quad \times  \left(\int_0^\infty \di t \, e^{-2t\varepsilon} e^{-t(k_{\F}^\uparrow)^2}\|\hat{v}_{t,\uparrow}\|_2 \|\hat{u}_{t,\downarrow}\|_2 \|\hat{v}_{t,\downarrow}\|_2 \right)\langle \psi,\mathcal{N}\psi\rangle 
	\\ 
	&\leq C\|\hat{g}\|_\infty \rho^{\frac{5}{9}}\left(\int_0^\infty \di t \, e^{-2\varepsilon t} \|\hat{v}_{t,\downarrow}\|_2 \|\hat{u}_{t,\downarrow}\|_2 \right)\langle \psi,\mathcal{N}\psi\rangle 
    \leq C_\kappa \|\hat{g}\|_\infty \rho^{\frac{5}{6} - \kappa}\langle \psi ,\mathcal{N}\psi\rangle \;,
\end{align*}
where in the last two estimates we used \Cref{lem: bounds phi}, and \Cref{lem:t} for the integral with respect to $t$. 
{To estimate $\mathrm{I}_{7,b}$, $\mathrm{I}_{7;c}$ and $\mathrm{I}_{7;d}$, it suffices to argue as in the estimate of $\mathrm{I}_{7;a}$. More precisely, we apply the same lemmas used for $\mathrm{I}_{7;a}$ and get the same bound for all the terms.} All together we find 
\[
	|\langle \psi, \mathrm{I}_{7}\psi\rangle| 
    \leq C_\kappa \|\hat{g}\|_\infty \rho^{\frac{5}{6} - \kappa}\langle \psi,\mathcal{N}\psi\rangle \;.
\]
Next, we estimate
\begin{align*}
	\mathrm{I}_8 
    &= \frac{1}{L^6}\sum_{p,k,r,r^\prime}\left(\hat{g}(r+p) + \hat{g}(-r)  \right) \hat{\varphi}^>(p) \hat{\eta}_{r+p-k,r^\prime-p+k}(k) \widehat{\chi}_<(k)\times
	\\ 
	&\quad \times \hat{u}_\uparrow(r+p)\hat{u}_\downarrow(r'-p)
        \hat{c}_{-r, \uparrow}^\ast 
        \hat{c}_{-r^\prime, \downarrow}^\ast 
        \hat{c}_{-r^\prime +p-k,\downarrow} 
        \hat{c}_{-r-p+k, \uparrow}
    = \mathrm{I}_{8;a} + \mathrm{I}_{8;b} \;.
\end{align*}
We write $\mathrm{I}_{8;a}$ as 
\begin{align*}
	\mathrm{I}_{8;a} 
    &= (8\pi a) \int_0^{\infty}\di t \, e^{-2\varepsilon t}\int \di x \di y \di z \di z^\prime\, \varphi^>(x-y)\chi_<(z-z^\prime)\times
	\\ 
	&\quad \times  (u_{t,\uparrow} * g)(x-z)u_{t,\downarrow}(y;z^\prime)
        c^\ast_{x,\uparrow}
        c_{y,\downarrow}^\ast 
        c_{t,z',\downarrow}
        c_{t,z,\uparrow} \;.
\end{align*}
Again applying Lemma \ref{lem:a*a*aa} and Lemma \ref{lem: a*a with N}, proceeding similarly as for $\mathrm{I}_{5;a}$ and $\mathrm{I}_{7;a}$ above, we obtain
\begin{align*}
	|\langle\psi,\mathrm{I}_{8;a}\psi\rangle| &\leq C\|\hat{g}\|_\infty \|\varphi^>\|_1 \|\chi_<\|_1 \|\hat{v}_\downarrow\|_2 \int_0^{\infty} \di t \,e^{-2t\varepsilon} e^{t(\kF^\uparrow)^2} \|\hat{u}_{t,\uparrow}\|_2 \|\hat{u}_{t,\downarrow}\|_2 \|\hat{v}_{t, \downarrow}\|_2 \langle \psi, \mathcal{N}\psi\rangle
 	\\ 
 	&\leq C\|\hat{g}\|_\infty \rho^{\frac{5}{9}}\int_0^{\infty} \di t \,e^{-2t\varepsilon} e^{t(\kF^\uparrow)^2}e^{t(\kF^\downarrow)^2} \|\hat{u}_{t,\uparrow}\|_2 \|\hat{u}_{t,\downarrow}\|_2 \langle \psi, \mathcal{N}\psi\rangle
	\\ 
	&\leq C \|\hat{g}\|_\infty \rho^{\frac 79}\langle \psi,\mathcal{N}\psi\rangle \;,
\end{align*}
where we used also Lemma~\ref{lem: bounds phi} and Lemma~\ref{lem:t}.
In the same way one obtains
\[
    |\langle \psi, \mathrm{I}_{8;b}\psi\rangle|\leq C \|\hat{g}\|_\infty \rho^{\frac 79}\langle \psi,\mathcal{N}\psi\rangle \;.
\]
Thus
\begin{align*}
    |\langle\psi,\mathrm{I}_{8}\psi\rangle|
    \leq C \|\hat{g}\|_\infty \rho^{\frac 79}\langle \psi,\mathcal{N}\psi\rangle \;.
\end{align*}
To conclude the proof, we estimate
\begin{align*}
	\mathrm{I}_9 &= - \frac{1}{L^6}\sum_{\sigma\neq \sigma^\prime}\sum_{p,k,r,r^\prime,s}\left( \delta_{\sigma,\uparrow} (\hat{g}(r+p) + \hat{g}(-r)) + \delta_{\sigma,\downarrow} (\hat{g}(r^\prime - p) + \hat{g}(-r^\prime)) \right) \hat{\varphi}^>(p)\times
	\\ 
	&\qquad \times \hat{\eta}_{r+p-k,s}(k)\widehat{\chi}_<(k)
    \hat{u}_\sigma(r+p)
        \hat{c}_{-r,\sigma}^\ast 
        \hat{c}_{-r^\prime, \sigma^\prime}^\ast
        \hat{d}_{r^\prime - p, \sigma^\prime}^\ast 
        \hat{d}_{s-k, \sigma^\prime}
        \hat{c}_{-s, \sigma^\prime}
        \hat{c}_{-r-p+k, \sigma}
	\\ &
    = \mathrm{I}_{9;a} + \mathrm{I}_{9;b} + \mathrm{I}_{9;c} + \mathrm{I}_{9;d} \;.
\end{align*}
Similar to $\mathrm{I}_6$ above, due to the constraints on $r^\prime$ and $p$, we can replace the factor $\hat{u}_{\sigma^\prime}(r^\prime -p)$ (implicit in $\hat d_{r'-p,\sigma'}^\ast$) with $\hat{u}_{\sigma^\prime}^>(r^\prime -p)$ using the definition of $\hat{u}^>_\cdot$ in \eqref{eq: def u>}.
We start with $\mathrm{I}_{9;a}$, which we write in position space as 
\begin{align*}
	\mathrm{I}_{9;a} 
    &= -\frac{8\pi a}{L^3} \int_0^{\infty}\di t \, e^{-2t\varepsilon}\sum_r \hat{u}_{t,\uparrow}(r) \hat{g}(r) \int \di x \di y \di z \di z^\prime\, e^{ir\cdot (x-z)}\varphi^>(x-y) \times
	\\ 
	&\quad \times \chi_<(z-z^\prime)
    c^\ast_{x,\uparrow} c^\ast_{y,\downarrow} d^\ast_{y,\downarrow} d_{t,z',\downarrow} c_{t,z',\downarrow}c_{t,z,\uparrow}
	\\ 
	&= \frac{8\pi a}{L^3} \int_0^\infty \di t \, e^{-2t\varepsilon} \sum_r \hat{u}_{t,\uparrow}(r) \hat{g}(r)
    \left( \int \di x \, e^{ir\cdot x} c^\ast_{x,\uparrow} b_{\downarrow}^\ast(\varphi^>_x, v, u^>)\right)\times
	\\ 
	&\quad \times \left(\int \di z \di z' \, e^{-ir\cdot z} \chi_<(z-z')
    d_{t,z',\downarrow} c_{t,z',\downarrow} c_{t,z,\uparrow}\right) \;.
\end{align*}
Proceeding similarly as for $\mathrm{I}_{6;a}$ and applying Lemma~\ref{lem:a*a*a*aaa}, we find
\begin{align*}
    |\langle \psi,\mathrm{I}_{9;a}\psi\rangle| 
    &\leq C\|\hat{g}\|_\infty \norm{\chi_<}_1 
    \sup_x \|b_{\downarrow}^\ast(\varphi^>_x, v, u^>)\| 
    \int_0^\infty \di t \, e^{-2t\varepsilon} 
    e^{-t(k_{\F}^\uparrow)^2}\times
    \\
	&\quad\times 
    \|\hat{u}_{t,\downarrow}\|_2 \|\hat{v}_{t,\downarrow}\|_2
    \left(\int \di x \, \|c_{x,\uparrow}\psi\|^2\right)^{\frac{1}{2}}
    \left(\int \di z\, \|c_{t,z,\uparrow} \psi\|^2\right)^{\frac{1}{2}}
	\\ 
	&\leq C\|\hat{g}\|_\infty \norm{\chi_<}_1 
    \rho^{\frac{7}{18}}
    \left(\int_0^\infty \di t \, e^{-2t\varepsilon} 
    \|\hat{u}_{t,\downarrow}\|_2 \|\hat{v}_{t,\downarrow}\|_2 \right)\langle \psi, \mathcal{N}\psi\rangle 
	\\ 
	& \leq C_\kappa \|\hat{g}\|_\infty \rho^{\frac{2}{3} - \kappa}\langle \psi,\mathcal{N}\psi\rangle \;,
\end{align*}
where we used \Cref{lem:t} to bound the $t$-integral. Similarly, we can estimate $\mathrm{I}_{9;b}$ through $\mathrm{I}_{9;d}$ getting $|\langle \psi, \mathrm{I}_{9;b} \psi\rangle| ,\, |\langle \psi, \mathrm{I}_{9;c} \psi\rangle|, \, |\langle \psi, \mathrm{I}_{9;d} \psi\rangle|\leq  C_\kappa \|\hat{g}\|_\infty \rho^{\frac{2}{3} - \kappa}\langle \psi,\mathcal{N}\psi\rangle  $. Therefore
\begin{equation*}
    \abs{\expect{\psi, \mathrm{I}_9 \psi}}
    \leq 
    C_\kappa \norm{\hat g}_\infty \rho^{\frac 23-\kappa} \expect{\psi, \cN\psi} \;.
\end{equation*}
This concludes the proof. 
\end{proof}

\subsection{Proof of the Main Theorem}\label{sec.proof.prop.main}

To finalize the proof of~\Cref{thm:main}, it remains to plug the estimates from~\Cref{pro:B2,pro:B1,pro:B1B2} into~\eqref{eq:Duhamelexpansion}, and to remove the regularizations by $\varepsilon$ and $\chi_<$, as well as to prove the energy statement \cref{eq:main_energy}.

\subsubsection{Removing the Regularizations}
We start by proving the following lemma which will allow us to remove the regularizations by $\varepsilon$ and $\chi_<$ in the excitation density.
\begin{lemma} \label{lem:remove_regularizations}
Let $\hat g$ be the restriction to $\Lambda^\ast$ of a function $\mathcal{F}(g) \in L^\infty(\RRR^3) \cap L^2(\RRR^3)$ and let $n^{(\Bel)}_{g} = \sum_{\sigma\in\{\uparrow, \downarrow\}} n^{(\Bel)}_{g,\sigma}$ as in \eqref{eq:ngBelyakov_2}. Let $\widehat{\chi}_<$ be as in \eqref{eq: def chi< chi>} and let $\varepsilon = \rho^{2/3 + \delta}$ for some $\delta > 2/9$. Define
\begin{equation}\label{eq: def constant thmd lim}
    \mathcal{C}_g
    := \frac{a^2}{\pi^4}\int_{\mathbb{R}^3}\di p  \int_{\scriptscriptstyle{\substack{|r| \leq \kF^\uparrow \\ \kF^\uparrow< |r + p|}}}\hspace{-0.2cm}\di r  \int_{\scriptscriptstyle{\substack{|r^\prime| \leq \kF^\downarrow \\ \kF^\downarrow< |r^\prime - p|}}}\hspace{-0.2cm}\di r ^\prime \frac{\left(\hat{g}(r+p) + \hat{g}(-r) + \hat{g}(r^\prime - p) + \hat{g}(-r^\prime)\right) \widehat{\chi}_<(p)^2}{(|r+p|^2 - |r|^2 + |r^\prime - p|^2 - |r^\prime|^2 + 2\varepsilon)^2} \;.
\end{equation}
Then
\begin{equation}\label{eq: nbelg - constant}
	\abs{n_g^{(\Bel)} - \mathcal{C}_g}
	\leq C \norm{\hat{g}}_\infty \big( \rho^{\frac 53 + \frac{1}{9}} \big) \;.
\end{equation}
\end{lemma}

\begin{proof}
Recall Belyakov's integral formula~\eqref{eq:ngBelyakov_2}:
\begin{equation*}
	n_{g, \sigma}^{(\Bel)}
	= \frac{a^2}{\pi^4} \int_{\RRR^3} \d p
		\int_{|r| < k_{\F}^\sigma < |r+p|} \d r
		\int_{|r'| < k_{\F}^{\sigma^\perp} < |r'-p|} \d r'
		\frac{\hat{g}(r+p) + \hat{g}(-r)}
		{(|r+p|^2 - |r|^2 + |r'-p|^2 - |r'|)^2} \;.
\end{equation*}
In the difference in \eqref{eq: nbelg - constant}, we have two errors to bound: one coming from the cutoff $ \widehat{\chi}_< $ and the one from the regularization $ \eps $. We first remove the cut-off. Writing $\widehat{\chi}_<(p)^2 = (1 - \widehat{\chi}_>(p))^2= 1 + (\widehat{\chi}_>(p))^2 - 2\widehat{\chi}_>(p)$, we have 
\begin{align*}
    \mathcal{C}_g &=  \frac{a^2}{\pi^4}\int_{\mathbb{R}^3}\di p  \int_{\scriptscriptstyle{\substack{|r| \leq \kF^\uparrow \\ \kF^\uparrow< |r + p|}}}\hspace{-0.2cm}\di r  \int_{\scriptscriptstyle{\substack{|r^\prime| \leq \kF^\downarrow \\ \kF^\downarrow< |r^\prime - p|}}}\hspace{-0.2cm}\di r ^\prime \frac{\hat{g}(r+p) + \hat{g}(-r) + \hat{g}(r^\prime - p) + \hat{g}(-r^\prime)}{(|r+p|^2 - |r|^2 + |r^\prime - p|^2 - |r^\prime|^2 + 2\varepsilon)^2} + \mathcal{E}_1 
\end{align*}
with
\begin{align}\label{eq: def cal e1}
	\cE_1
	&:= \frac{a^2}{\pi^4} \int_{\RRR^3} \d p
		\int_{|r| < k_{\F}^\uparrow < |r+p|} \d r
		\int_{|r'| < k_{\F}^\downarrow < |r'-p|} \d r' \times \nonumber
        \\ 
        &\qquad \times 	\frac{\left(\hat{g}(r+p) + \hat{g}(-r) +\hat{g}(r^\prime - p) + \hat{g}(-r^\prime)\right) \left(\widehat{\chi}_>(p)^2 -2\widehat{\chi}_>(p)\right)}
		{(|r+p|^2 - |r|^2 + |r'-p|^2 - |r'| + 2 \eps)^2} \;.
\end{align}
Recalling definition \eqref{eq:chi_q_alpha} of $ \widehat{\chi}_> $, the $p$-integration in $\mathcal{E}_1$ is always supported in the set $ |p| > 4 \rho^{2/9} $, for which $ |r+p|^2 - |r|^2 > C |p|^2 $ for $ \rho $ small enough and $|r+p| > \kF^\uparrow > |r|$. Thus, we can estimate
\begin{equation} \label{eq:cE1bound_final}
\begin{aligned}
	|\cE_1|
	& \leq C \norm{\hat{g}}_\infty \int_{|p| > 4 \rho^{2/9}} \d p
		\int_{|r| < k_{\F}^\uparrow < |r+p|} \d r
		\int_{|r'| < k_{\F}^\downarrow < |r'-p|} \d r'
		\frac{1}{|p|^4} \\
	& \leq C \norm{\hat{g}}_\infty \rho^2
		\int_{4 \rho^{2/9}}^\infty \d |p| \frac{1}{|p|^2}
	\leq C \norm{\hat{g}}_\infty \rho^{\frac 53 + \frac{1}{9}} \;.
\end{aligned}
\end{equation}
We now remove the regularization $ \eps $. This produces another error term. We indeed have 
\begin{align*}
    \mathcal{C}_g &= \frac{a^2}{\pi^4}\int_{\mathbb{R}^3}\di p  \int_{\scriptscriptstyle{\substack{|r| \leq \kF^\uparrow \\ \kF^\uparrow< |r + p|}}}\hspace{-0.2cm}\di r  \int_{\scriptscriptstyle{\substack{|r^\prime| \leq \kF^\downarrow \\ \kF^\downarrow< |r^\prime - p|}}}\hspace{-0.2cm}\di r ^\prime\, \frac{\hat{g}(r+p) + \hat{g}(-r) + \hat{g}(r^\prime - p) + \hat{g}(-r^\prime)}{(|r+p|^2 - |r|^2 + |r^\prime - p|^2 - |r^\prime|^2)^2} + \mathcal{E}_1 + \mathcal{E}_2 
      \\
    & = n^{(\mathrm{Bel})}_{g} + \mathcal{E}_1 + \mathcal{E}_2 \;,
\end{align*}
with $\mathcal{E}_1$ as in \eqref{eq: def cal e1} and
\begin{align}
	\cE_2
	& := \frac{a^2}{\pi^4} \int_{\RRR^3} \d p
		\int_{|r| \leq k_{\F}^\uparrow < |r+p|} \!\!\!\!\d r
		\int_{|r'| \leq k_{\F}^\downarrow < |r'-p|} \!\!\!\! \d r'\big( \hat{g}(r+p) + \hat{g}(-r) + \hat{g}(r^\prime - p) + \hat{g}(-r^\prime) \big) \times \nonumber\\
	& \quad \times \left(
		\frac{1}{(|r+p|^2 - |r|^2 + |r'-p|^2 - |r'| + 2\eps)^2}
		- \frac{1}{(|r+p|^2 - |r|^2 + |r'-p|^2 - |r'| )^2}
		\right) \;.
\end{align}
To estimate $\mathcal{E}_2$, we use that
\begin{equation}\label{eq: est const eps}
	\abs{\frac{1}{A^2} - \frac{1}{(A + 2\eps)^2}}
	= \frac{4 \eps A + 4 \eps^2}{A^2 (A + 2\eps)^2}
	\leq {C} \eps A^{-3} \;.
\end{equation}
Using \eqref{eq: est const eps} and rescaling the three integrals by $ k_{\F}^\uparrow \leq C \rho^{\frac 13} $, we thus get 
\begin{equation}
\begin{aligned}
	|\cE_2|
	& \leq C \norm{\hat{g}}_\infty \rho^{\frac 23 + \delta}
		\int_{\RRR^3} \d p
		\int_{|r| < k_{\F}^\uparrow < |r+p|} \hspace{-0.4cm}\d r
		\int_{|r'| < k_{\F}^\downarrow < |r'-p|}\hspace{-0.4cm} \d r'
		\frac{1}{(|r+p|^2 - |r|^2 + |r'-p|^2 - |r'|)^3} \\
	& \leq C \norm{\hat{g}}_\infty \rho^{\frac 53 + \delta}
		\int_{\RRR^3} \d p
		\int_{|r| < 1 < |r+p|} \hspace{-0.4cm} \d r 
		\int_{|r'| < x < |r'-p|}\hspace{-0.4cm} \d r'
		\frac{1}{(|r+p|^2 - |r|^2 + |r'-p|^2 - |r'|)^3} \;,
\end{aligned}
\end{equation}
with $ x = k_{\F}^\downarrow / k_{\F}^\uparrow $. Without loss of generality we assume $ x \in (0,1] $, {otherwise we could have rescaled with respect to $k_\F^\downarrow$}. The integral {above is} bounded by Lemma~\ref{lem:belyakov_integral_bound}, yielding
\begin{equation} \label{eq:cE2bound_final}
	|\cE_2|
	\leq C \norm{\hat{g}}_\infty \rho^{\frac 53 + \delta} \;,
\end{equation}
{where $C$ is a constant which depends on $k_\F^\downarrow/k_\F^\uparrow$.}
Recalling $\delta > \frac{2}{9}$ concludes the proof.
\end{proof}

\subsubsection{Isolating the Constant Term}
\label{subsec:proofmain}

\begin{proof}[Proof of~\Cref{thm:main}]
We choose the trial state $\Psi = R T_1 T_2 \Omega$ constructed in Section~\ref{subsec:trialstate}. It was proved in~\cite{GHNS24} that $L^{-3}\langle \Psi, H_N \Psi \rangle = e_{\textnormal{HY}}(\rho_\uparrow, \rho_\downarrow) + \mathcal{O}(\rho^{{7/3 + 1/9}}) + o_L(1)$, where $e_{\textnormal{HY}}(\rho_\uparrow, \rho_\downarrow)$ is the Huang--Yang energy {density}, defined in~\cite[(1.7), (1.8)]{GHNS24}.
On the other hand, combining the results from~\cite{GHNS24} and~\cite{GHNS25} yields $e(\rho_\uparrow, \rho_\downarrow) = e_{\textnormal{HY}}(\rho_\uparrow, \rho_\downarrow) + \mathcal{O}(\rho^{{7/3 + 1/120}})$, where $ e(\rho_\uparrow, \rho_\downarrow) = \lim_{L \to \infty} L^{-3}E_L(N_\uparrow, N_\downarrow) $. This readily proves the energy statement~\eqref{eq:main_energy}.

\medskip

To establish \Cref{eq:main}, we start from the Duhamel expansion in~\Cref{lem:Duhamelexpansion} and plug in the bounds from~\Cref{pro:B2,pro:B1,pro:B1B2}, extracting the main constant contribution. Recalling \Cref{eq:nexc_simplification} $\langle \Psi, n_g^{\exc} \Psi \rangle = \langle T_1 T_2 \Omega, n_g T_1T_2\Omega\rangle $, and with $\int_0^1 \di \lambda (1 - \lambda) = \frac 12$, we find
\begin{equation}\label{eqn.formula.ng.T1T2Omega}
\begin{aligned}
	&\langle T_1 T_2 \Omega, n_g T_1T_2\Omega\rangle
	\\ 
	&= \frac{(2 \pi)^3}{L^9}\sum_{\substack{ |r| \leq \kF^\uparrow < |r+p| \\ |r^\prime| \leq \kF^\downarrow < |r^\prime - p|}}\left(\hat{g}(r+p) + \hat{g}(-r) + \hat{g}(r^\prime - p) + \hat{g}(-r^\prime)\right)\left(\hat{\eta}^{\varepsilon}_{r,r^\prime}(p) \widehat{\chi}_<(p)\right)^ 2 + \mathcal{E}_{n_g} \;,
\end{aligned}
\end{equation}
where the error is bounded for all $\kappa >0$, as
\begin{align*}
|\mathcal{E}_{n_g}|&\leq C \|\hat{g}\|_\infty \rho^{\frac{5}{3} + \frac{1}{9}} + C_k L^{-3}\|\hat{g}\|_\infty \rho^{\frac{1}{2} - \kappa} \sup_\lambda\langle T_{2,\lambda}\Omega, \mathcal{N}T_{2,\lambda}\Omega\rangle 
    \\
    & \quad + CL^{-3}\rho^{\frac{7}{9}}\|\hat{g}\|_\infty \sup_\lambda\langle T_{1;\lambda}T_2\Omega, \mathcal{N}T_{1;\lambda}T_2\Omega\rangle \\
    & \quad + CL^{-3}\rho^{\frac{4}{9}}\|\hat{g}\|_2 \sup_\lambda\langle T_{1;\lambda}T_2\Omega, \mathcal{N}T_{1;\lambda}T_2\Omega\rangle 
    \\
    & \quad + C_kL^{-3}\|\hat{g}\|_\infty \rho^{\frac{5}{9} - \kappa}\sup_\lambda\langle T_{2;\lambda}\Omega, \mathcal{N}T_{2;\lambda}\Omega\rangle \;.
\end{align*}
Next, we apply~\Cref{lem: est number op} to bound the number operator and get:
\[
	|\mathcal{E}_{n_g}|\leq C \|\hat{g}\|_\infty \rho^{\frac 53 + \frac 19} + C_\kappa \rho^{2-\kappa}\|\hat{g}\|_2 \;,
\]
for any $\kappa>0$. Taking the thermodynamic limit and applying Lemma \ref{lem:remove_regularizations}, we obtain
\begin{align}\label{eq: thermod lim}
	&\langle T_1 T_2\Omega, n_g T_1T_2\Omega\rangle \nonumber \\
	& =  \frac{(2\pi)^3}{L^9}\sum_{\substack{ |r| \leq \kF^\uparrow < |r+p| \\ |r^\prime| \leq \kF^\downarrow < |r^\prime - p|}}\left(\hat{g}(r+p) + \hat{g}(-r) + \hat{g}(r^\prime - p) + \hat{g}(-r^\prime)\right)\left(\hat{\eta}^{\varepsilon}_{r,r^\prime}(p) \widehat{\chi}_<(p)\right)^ 2 
     + {\mathcal{E}}_{n_g}
    \nonumber
	\\ 
	&  = \frac{(8\pi a)^2}{(2\pi)^6} \int_{\mathbb{R}^3}\di p  \int_{\scriptscriptstyle{\substack{|r| \leq \kF^\uparrow \\ \kF^\uparrow< |r + p|}}}\hspace{-0.2cm}\di r \int_{\scriptscriptstyle{\substack{|r^\prime| \leq \kF^\downarrow \\ \kF^\downarrow< |r^\prime - p|}}}\hspace{-0.2cm}\di r ^\prime \frac{\left(\hat{g}(r+p) + \hat{g}(-r) + \hat{g}(r^\prime - p) + \hat{g}(-r^\prime)\right) \widehat{\chi}_<(p)^2}{(|r+p|^2 - |r|^2 + |r^\prime - p|^2 - |r^\prime|^2 + 2\varepsilon)^2} \nonumber 
	\\ 
	&\quad  + {\mathcal{E}}_{n_g} + o_{L\to\infty}(1)\nonumber
    \\
    &  = \frac{a^2}{\pi^4} \int_{\mathbb{R}^3}\di p  \int_{\scriptscriptstyle{\substack{|r| \leq \kF^\uparrow \\ \kF^\uparrow< |r + p|}}}\hspace{-0.2cm}\di r \int_{\scriptscriptstyle{\substack{|r^\prime| \leq \kF^\downarrow \\ \kF^\downarrow< |r^\prime - p|}}}\hspace{-0.2cm}\di r ^\prime \frac{\left(\hat{g}(r+p) + \hat{g}(-r) + \hat{g}(r^\prime - p) + \hat{g}(-r^\prime)\right) \widehat{\chi}_<(p)^2}{(|r+p|^2 - |r|^2 + |r^\prime - p|^2 - |r^\prime|^2 + 2\varepsilon)^2}\nonumber
	\\ 
	&\quad  + {\mathcal{E}}_{n_g} + o_{L\to\infty}(1) \;.
\end{align}
Removing the regularizations using \Cref{lem:remove_regularizations}, and using that according to \eqref{eq:nexc_simplification} we have $ \langle \Psi, n^{(\exc)}_g \Psi \rangle = \langle T_1 T_2 \Omega, n_g T_1 T_2 \Omega \rangle $, yields~\Cref{eq:main}.
%, i.\,e., 
%\[
%    	\abs{\langle \Psi, n_g^{(\exc)} \Psi\rangle - n^{(\Bel)}_{g}}
%	\leq C \|\hat{g}\|_\infty \rho^{\frac{5}{3} + \frac{1}{9}}
%        + C_\kappa\rho^{2-\kappa} \|\hat g\|_2 
%        + o_{L\to \infty}(1)\;.
%\]

It remains to bound $n^{(\Bel)}_{g,\sigma}$ as claimed in~\eqref{eq:ngBelyakovbound}. From its definition and by rescaling the integrals by $k_{\F}^{\sigma} \leq C \rho^{\frac 13}$, we get
\begin{equation}
	n_{g, \sigma}^{(\Bel)} = \frac{a^2}{\pi^4}(k_{\F}^\sigma)^{5}\int_{\mathbb{R}^3}\di q \, \hat{g}(\kF^\sigma q)\I^{(q)}_2 \;,
\end{equation}
with $\I^{(q)}_2$ as in \eqref{eq: def Is q}, where $x= k_{\F}^{\sigma^\perp} / k_{\F}^\sigma \in (0,\infty)$. Now, \eqref{eq:ngBelyakovbound} follows by estimating the triple integral $\I_2^{(q)}$ using Lemma~\ref{app:belyakov_integral_bound_q}.
\end{proof}

\appendix
\section{Estimates of Belyakov-Type Integrals}
\label{app:belyakov_integral_bound}
\begin{lemma} \label{lem:belyakov_integral_bound} For $s \in \{2,3\}$ and $x\in (0,1]$, let 
\[
	\mathrm{I}_s = \int_{\mathbb{R}^3}\d p \int_{|r| < 1 < |r+p|} \d r \int_{|r^\prime| < x < |r^\prime -p|}\d r^\prime \, \frac{1}{(|r+p|^2 - |r|^2 + |r^\prime - p|^2 - |r^\prime|^2)^s} \;.
\]
Then, there exists a $C>0$ (independent of $x$) such that
\begin{equation}\label{eq: integ power 2}
    \mathrm{I}_2
    \leq C \;, \qquad
    \mathrm{I}_3 
    \leq C x^{-3} \qquad \text{for all } x\in (0,1]\;.
\end{equation}
\end{lemma}

\begin{proof}
We introduce the excitation energies
\begin{equation}\label{eq: def er etilder}
	e_r := ||r|^2 - 1|\; , \qquad
    \widetilde{e}_{r} := ||r|^2 - x^2| \;.
\end{equation}
Then
\[
	\mathrm{I}_s = \int_{\mathbb{R}^3}\d p \int_{|r| < 1<|r+p|} \d r\int_{|r^\prime| < x < |r^\prime - p|} \d r^\prime \frac{1}{(e_{r+p} + e_r + \widetilde{e}_{r^\prime - p} + \widetilde{e}_{r^\prime})^s} \;.
\]
We distinguish
\[
    \mathrm{I}_{s}\big\vert_{\{p \in \mathbb{R}^3: |p|\leq 3\}} =: \mathrm{I}_s^< \;, \qquad
    \mathrm{I}_s\big\vert_{\{p \in \mathbb{R}^3: |p| > 3\}} =: \mathrm{I}_{s}^> \;.
\]
Here and in the following, given an integral $X = \int_B f$, we use the notation
\[
    X\big\vert_A := \int_{A\cap B} f \;.
\]

\noindent
\underline{Estimate for $\mathrm{I}_s^>$ for $s=2,3$.} Observe first that for $|p|> 3$, the integrand is bounded by $C|p|^{-2s}$, since $(e_{r+p} + e_r + \widetilde{e}_{r^\prime - p} + \widetilde{e}_{r^\prime})^s\geq C|p|^{2s}$. Hence, for $s=2,3$, we have
\begin{equation}\label{eq: est Is>}
	\mathrm{I}_{s}^> = \int_{|p|> 3}\d p \int_{|r| < 1<|r+p|} \d r\int_{|r^\prime| < x < |r^\prime - p|} \d r^\prime \frac{1}{(e_{r+p} + e_r + \widetilde{e}_{r^\prime - p} + \widetilde{e}_{r^\prime})^s} \leq C \;.
\end{equation}
\underline{Estimate for $\mathrm{I}_2^<$.} Recall that $\mathrm{I}_2^<$ is given by
\begin{equation}\label{eq: Is<}
	\mathrm{I}_2^< := \int_{|p|\leq 3}\d p 
        \int_{|r| < 1 < |r+p|} \d r
        \int_{|r^\prime| < x < |r^\prime - p|}\d r^\prime
        \frac{1}{(e_{r+p} + e_r + \widetilde{e}_{r^\prime - p} + \widetilde{e}_{r^\prime})^2} \;.
\end{equation}
This estimate was already proved in \cite[Appendix C]{GHNS24}; for completeness we recall the argument. Using that $|r^\prime -p| \geq x$ and $\widetilde{e}_{r^\prime}\geq x(x-|r^\prime|)$, we obtain
\[
	\mathrm{I}_2^< \leq \int_{|p|\leq 3} \d p\int_{|r| < 1 < |r+p|} \d r \int_{|r^\prime| < x < |r^\prime - p|} \d r^\prime \frac{1}{e_{r+p}^2 + e_{r}^2 + x^2(x-|r^\prime|)^2} \;.
\]
We distinguish the two cases $e_{r+p}\leq e_r$ and $e_r< e_{r+p}$ and we denote the corresponding integrals by $\mathrm{I}_{2;1}^<$ and $\mathrm{I}_{2;2}^<$. For the first integral, i.\,e., $\mathrm{I}_{2;1}^<$, the condition $e_{r+p}\leq e_r$ implies $1< |r+p| < 2-|r|^2$. Dropping the positive term $e_{r+p}^2$ in the denominator for an upper bound, we can integrate with respect to $p$ and then with respect to $r^\prime,r$, which yields
\[
	\mathrm{I}_{2;1}^< \leq C\int_{|r|< 1}\d r \int_{|r^\prime| < x}\d r^\prime \frac{e_r}{e_r^2 + x^2(x-|r^\prime|)^2} \leq Cx \leq C \;.
\]
The case $e_{r}< e_{r+p}$ is treated in the same way, by performing a suitable change of variable (see \cite[Appendix C]{GHNS24}). More precisely, by changing variables with $q= r+p$, and using $1 < |q| \leq |r|+|p|\leq 3$, we get
\begin{align*}
    \mathrm{I}_{2;2}^<
    &\leq C\int_{|r|<1} \d r\int_{|r^\prime|<1}\d r^\prime \int_{1 < |q| < 3}\d q \; \chi_{\{\sqrt{(2-|q|^2)_+} \leq |r| \leq 1\}}\frac{C}{e_q^2 + e_{r^\prime}^2}
\\ 
    &\leq C\int_{|r^\prime| < 1}\d r^\prime\int_{|q| <3}\d r\frac{e_q}{e_q^2 + \widetilde{e}_{r^\prime}^2}
    \leq C \;,
\end{align*}
where in the next to the last inequality we used that $1- (2-|q|^2)^{3/2}_+ \leq C(|q|^2 - 1) = Ce_q$ and then concluded similarly as for $\mathrm{I}_{2;1}^<$.

\medskip

\noindent \underline{Estimate for $\mathrm{I}_3^<$.}
%To estimate $\mathrm{I}_3^<$, a sharper analysis is required.
We contrast the argument to the case $s=2$.
The key point is that the bound $\widetilde{e}_{r^\prime - p} + \widetilde{e}_{r^\prime}\geq x (x-|r^\prime|)$ is no longer sufficient. In particular, one can no longer ignore $\widetilde{e}_{r^\prime - p}$ by positivity. It is therefore convenient to split the integrand into four disjoint sets depending on which of $e_{r}$, $e_{r+p}$, $\widetilde{e}_{r^\prime}$ and $\widetilde{e}_{r^\prime -p}$ is the largest. Correspondingly, we write 
\begin{equation}\label{eq: int I3<}
	\mathrm{I}_3^< = \mathrm{I}_{3;1}^< + \mathrm{I}_{3;2}^< + \mathrm{I}_{3;3}^< + \mathrm{I}_{3;4}^< \;,
\end{equation}
where $e_{r}$ dominates in $\mathrm{I}_{3;1}^<$, $e_{r+p}$ in $\mathrm{I}_{3;2}^<$, $\widetilde{e}_{r^\prime}$ in $\mathrm{I}_{3;3}^<$ and $\widetilde{e}_{r^\prime -p}$ in $\mathrm{I}_{3;4}^<$. We estimate each integral separately. We discuss the estimate for $\mathrm{I}_{3;1}^<$ in detail, and then indicate how the remaining integrals are handled in the same spirit.

We introduce a cut-off parameter $\epsilon > 0$ and distinguish the cases $e_{r} < \epsilon$ and $e_{r}>\epsilon$. It is going to be convenient to fix $\epsilon = c x^2$ for some $0< c < 3/4$ small enough. In the latter case, {i.\,e., $e_r >\epsilon$,} we get a bound on the integral $\mathrm{I}_3^<$ restricted to the region where $e_r > \epsilon$, which we denote $\I_{3;1}^{> \epsilon}$. In fact,
\begin{align}\label{eq: I31 big eps}
	\I_{3;1}^{> \epsilon}
    &\leq \int_{|p|\leq 3} \d p \int_{|r|< 1 < |r+p|} \d r\int_{|r^\prime| < x <|r^\prime -p|} \d r^\prime \frac{\chi_{\{r: e_r > \epsilon\}}}{e_r^3 + e_{r+p}^3 + \widetilde{e}_{r^\prime -p}^3 + \widetilde{e}_{r^\prime}^3} \nonumber \\ 
	&\leq \epsilon^{-3} \int_{|p|\leq 3} \d p \int_{|r|< 1 } \d r\int_{|r^\prime| < x } \d r^\prime\, 1\leq C\epsilon^{-3} x^3 \leq \frac{C}{x^3} \;.
\end{align}
Here we used that $\epsilon = c x^2$. We next consider the case $e_r < \epsilon$. Recalling the constraints of the integration domain, we have
\begin{equation}\label{eq: constraints I3<1}
\begin{split}
&	|p| \leq 3 \;, \quad 
    |r| < 1 < |r+p| \;, \quad
    |r^\prime| < x <|r^\prime -p | \;, \quad 
    e_r < \epsilon \;, \\
&    e_r \geq e_{r+p} \;,    \quad
    e_r \geq  \widetilde{e}_{r^\prime} \;, \quad \text{and }
    e_r \geq  \widetilde{e}_{r^\prime -p} \;.
\end{split}
\end{equation}
We denote by $\mathrm{I}_{3;1}^{<\epsilon}$ the corresponding integral under the constraints \eqref{eq: constraints I3<1}.
Note that the conditions $e_r \geq e_{r+p}$, $e_r \geq \widetilde{e}_{r^\prime -p}$, and $e_r \geq \widetilde{e}_{r^\prime}$ are equivalent to
\begin{equation}\label{eq: regions A1 A2}
	1 < |r+p|^2 < 1+ e_r \;, \qquad
    x^2 < |r^\prime -p |^2 < x^2 + e_r \;, \qquad 
    x^2 - e_r < |r^\prime|^2 < x^2 \;.
\end{equation}
In particular, the last condition together with $e_r < \epsilon$ implies that $|r^\prime|$ must be close to $x$. Unlike in the estimate of $\mathrm{I}_2^<$, here we need to use the first two conditions in \eqref{eq: regions A1 A2} more carefully. We define the following annuli, see Figure~\ref{fig:Shell_intersection}
\begin{equation} \label{eq:Ar_Arrx_Bx}
\begin{aligned}
    A_{r} &:= \left\{p\in\mathbb{R}^3\, : \, 1 < |r+p|^2 < 1+ e_r \right\} \;, \\
    {A}_{r,r^\prime,x} &:= \left\{p\in\mathbb{R}^3\, : \, x^2 < |r^\prime -p |^2 < x^2 + e_r \right\} \;, \\
    {B}_{r,x} &:= \left\{r^\prime\in\mathbb{R}^3\, : \, x^2 - e_r < |r^\prime|^2 < x^2\right\} \;.
\end{aligned}
\end{equation}
\begin{figure}
 	\centering
	\scalebox{0.9}{\input{Shell_intersection.tex}}
\caption{In $\I_{3;1}^{<\epsilon}$, fixing $(r,r')$, the integral runs over such $p \in \mathbb{R}^3$ that $p+r$ and $r'-p$ are outside the respective Fermi balls, but the energy $e_r$ still dominates. The set of such $p$ is an intersection of two annuli $A_r$ and $A_{r,r',x}$, defined in~\eqref{eq:Ar_Arrx_Bx}.}
\label{fig:Shell_intersection}
\end{figure}
A natural first attempt in estimating $\mathrm{I}_{3;1}^{<\epsilon}$ is therefore to bound $|A_{r}\cap {A}_{r,r^\prime,x}|$. The volume of the intersection $A_{r}\cap {A}_{r,r^\prime,x}$ depends on the relative position of $r$ to $r^\prime$. However, when $r,r^\prime$ are such that $|r+r^\prime|$ is close to $1-x$, the intersection angle between the two sets becomes shallow, and consequently $|A_{r}\cap {A}_{r,r^\prime,x}|$ becomes large. To separate such ``dangerous" configurations from the typical ones, we distinguish the following two cases in which a large intersection volume can occur: The case $||r+r^\prime| - (1-x)| \leq (1/2) e^{1/2}_r $ where $r$ and $r'$ are almost pointing in the opposite direction, and the case $||r+r^\prime| - (1+x)| \leq (1/2) e^{1/2}_r $, where $r$ and $r'$ are almost aligned.

We then introduce the region (see Figure~\ref{fig:dangerous})
\begin{equation}\label{eq: def Crx}
\begin{aligned}
	C_{r,x}
    &:= C_{r,x}^{(\mathrm{in})} \cup C_{r,x}^{(\mathrm{out})} \;, \\
    C_{r,x}^{(\mathrm{in})}
    &:= \left\{ r^\prime\in\mathbb{R}^3\, : \, |r^\prime| < x,\, ||r+r^\prime| - (1-x)| \leq 2 e^{1/2}_r \right\} \;, \\
    C_{r,x}^{(\mathrm{out})}
    &:= \left\{ r^\prime\in\mathbb{R}^3\, : \, |r^\prime| < x,\, ||r+r^\prime| - (1+x)| \leq 2 e^{1/2}_r \right\} \;. \\
\end{aligned}
\end{equation}
Since $e_r< \epsilon$, the region $C_{r,x}$ corresponds to the situation in which $|A_{r}\cap {A}_{r,r^\prime,x}|$ is large. In what follows, we estimate $\mathrm{I}_{3;1}^{<\epsilon}$ separately on the ``non-dangerous" region $\{ r^\prime\in\mathbb{R}^3\, : \, |r^\prime| < x\}\setminus C_{r,x}$ and on $C_{r,x}$.
\begin{figure}
 	\centering
	\scalebox{1.0}{\input{Shell_intersection_dangerous.tex}}
    \hspace{4em}
	\scalebox{1.0}{\input{Dangerous_belt.tex}}
\caption{\textbf{Left:} Example of a possible ``dangerous area'', in which the intersection volume of the two annuli $ A_r $ and $ A_{r,r',x} $ becomes large. 
\textbf{Right:} We represent the region $C_{r,x}^{(\mathrm{in})}$. If $r'$ is in this region, then $|A_r \cap A_{r,r',x}|$ becomes large.}
\label{fig:dangerous}
\end{figure}
We write
\begin{equation}\label{eq: splitting I31}
	\mathrm{I}_{3;1}^{<\epsilon} = \mathrm{I}_{3;1}^{<\epsilon} \big\vert_{C_{r,x}} + \mathrm{I}_{3;1}^{<\epsilon} \big\vert_{\left\{ r^\prime\in\mathbb{R}^3\, : \, |r^\prime| < x\right\}\setminus C_{r,x}} =: \mathrm{I}_{3;1;1}^{<\epsilon} + \mathrm{I}_{3;1;2}^{<\epsilon} \;.
\end{equation}

\noindent
\underline{Estimate of $\mathrm{I}_{3;1;2}^{<\epsilon}$.} 
Inside the integration domain, $|r| > 1-e_r$ and $|r'|\leq x$. Thus, $|r+r'| \geq |r| - |r'| \geq 1 - x - e_r$. By definition of $C_{r,x}^{\mathrm{(in)}}$ we must thus have $|r+r'| \geq 1 - x + 2e_r^{1/2}$. Similarly, one can conclude that $|r+r'| \leq 1+x-2e_r^{1/2}$. We next use these bounds, to control the overlap $A_r \cap A_{r,r',x}$. We claim 
\begin{equation}\label{eq: estimate area ar arprime nondangerous}
	|A_r \cap {A}_{r,r^\prime,x}| \leq C x^{-1/2}e_r^{3/2} \;.
\end{equation}
To prove this, let $\beta$ be the intersection angle of the tangents at the intersection of $\partial B_1(0)$ and $\partial B_x(r+r^\prime)$ (see Figure~\ref{fig:Angles}). Then, $|A_r \cap {A}_{r,r^\prime,x}|$ is bounded by $(\sin\beta)^{-1}$ times the product of the thicknesses of the two annuli, where $A_r$ has thickness $\leq C e_r$, while $A_{r,r',x}$ has thickness $\leq C x^{-1} e_r$, so
\begin{equation}\label{eq: area ar arprime}
	|A_r \cap {A}_{r,r^\prime,x}| \leq C x^{-1} e_r^2 (\sin\beta)^{-1} \;.
\end{equation}
In the region considered, we obtain from $|r+r'| \geq 1 - x + 2e_r^{1/2}$
\begin{equation} \label{eq:beta_cosine_relation}
	((1-x) + 2 e_r^{1/2})^2 < |r+r^\prime|^2 = 1+x^2 - 2x\cos\beta \;.
\end{equation}
Using that $\cos\beta > 1-\beta^2/2$, we get \begin{equation}\label{eq: lower bound beta}
	\beta^2 > \frac{2 e_r^{1/2}(2(1-x) + 2 e_r^{1/2})}{x} {\geq \frac{c e_r}{x}}\;.
\end{equation}
Similarly, from the constraint $|r+r'| \leq 1 + x - 2 e_r^{1/2}$, {with $ e_r < (3/4) x^2 $} we find
\begin{equation}\label{eq: lower bound beta_bis}
	(\pi - \beta)^2
    > \frac{2 e_r^{1/2}(2(1+x) - 2 e_r^{1/2})}{x}
    \geq \frac{c e_r}{x} \;.
\end{equation}
Combining these and using $x\leq 1$ {together with $\sin \beta \geq (2/\pi) \min(\beta, \pi - \beta)$,} we conclude $\sin\beta > C e_r^{1/2} x^{-1/2}$, which inserted into~\eqref{eq: area ar arprime} yields \eqref{eq: estimate area ar arprime nondangerous}.

Using now the definition of $\mathrm{I}_{3;1}^{<\epsilon}$ (in particular the constraints \eqref{eq: constraints I3<1}) and the definition of $\mathrm{I}_{3;1;2}^<$ in \eqref{eq: splitting I31}, we obtain
\begin{align*}
	\mathrm{I}_{3;1;2}^{<\epsilon}
    &\leq \int_{|p|\leq 3}\d p\int_{|r| < 1} \d r\int_{|r^\prime| < x} \d r^\prime \frac{\chi_{A_r\cap {A}_{r,r^\prime,x}}(p) \chi_{B_{r,x}}(r^\prime) (1-\chi_{C_{r,x}})(r^\prime)\chi_{\{e_r < \epsilon\}}(r)}{e_{r+p}^3 + e_r^3 + \widetilde{e}_{r^\prime -p}^3 + \widetilde{e}_{r^\prime}^3}
	\\
	&\leq C x^{-1/2} \int_{|r|< 1} \d r \int_{|r^\prime| < x}\d r^\prime e_r^{-3/2}\,\chi_{B_{r,x}}(r^\prime)(1-\chi_{C_{r,x}}(r^\prime))
	\\
	&\leq C x^{-1/2} \int_{|r|< 1}\d r\int_{|r^\prime| < x}\d r^\prime e_r^{-3/2}\chi_{\{\sqrt{x^2 - e_r}\leq |\cdot| \leq x\}}(r^\prime)
	\\
	&\leq C x^{1/2} \int_{|r|<1} \d r\, e^{-1/2}_r \;.
\end{align*}
In the second inequality, we used that the integral over $p$ is bounded, together with the estimate \eqref{eq: estimate area ar arprime nondangerous} and the positivity of $e_{r+p}^3 + \widetilde{e}_{r^\prime-p}^3 + \widetilde{e}_{r^\prime}^3\geq 0$. In the last estimate, we used the constraints on $r^\prime$. The remaining integral has a singularity at $ r = 1 $, which is integrable. Therefore
\[
	\mathrm{I}_{3;1;2}^{<\epsilon} \leq {C x^{1/2}}\;.
\]
\begin{figure}
 	\centering
	\scalebox{0.95}{\input{Angle_alpha.tex}}
	\hspace{4em}
	\scalebox{0.95}{\input{Angle_beta.tex}}
\caption{\textbf{Left:} The intersection angle between the two tangents at the intersection point of the annuli $A_r$ and $A_{r,r^\prime,x}$ is called $\beta$. \textbf{Right:} The angle $ \tilde\beta $ quantifies the area of the dangerous region $C_{r,x}^{\mathrm{(in)}}$ at fixed $|r^\prime| $. This area is a spherical cap. Similarly, $C_{r,x}^{\mathrm{(out)}} \cap \partial B_s(0)$ is a spherical cap. The set $C^{(\mathrm{out})}_{r,x}$ is drawn on top for comparison.}
\label{fig:Angles}
\end{figure}

\noindent
\underline{Estimate of $\mathrm{I}_{3;1;1}^{<\epsilon}$.} In this region, we perform a different analysis. {We first discuss the proof for the region $C_{r,x}^{(\mathrm{in})}$.} We estimate $|A_r \cap A_{r,r^\prime, x}|\leq |A_r| \leq Ce_r$ and use the positivity $e_{r+p}^3 + \widetilde{e}_{r^\prime - p}^3 + \widetilde{e}_{r^\prime}^3 \geq 0$ to write 
\begin{align*}
    &\int_{|p|\leq 3} \d p\int_{|r|<1 <|r+p|}\d r\int_{|r^\prime| < x } \d r^\prime \frac{\chi_{A_r\cap A_{r,r^\prime,x}}(p)\chi_{B_{r,x}}(r^\prime)\chi_{C_{r,x}^{\mathrm{(in)}}}(r^\prime)\chi_{\{e_r \leq \epsilon\}}(r)}{e_{r+p}^3 + e_{r}^3 + \widetilde{e}_{r^\prime - p}^3 + \widetilde{e}_{r^\prime}^3}
	\\
	&\quad \leq C\int_{|r|<1}\d r
        \int_{|r^\prime| < x} \d r'\, 
        \frac{\chi_{B_{r,x}}(r^\prime)\chi_{C_{r,x}^{\mathrm{(in)}}}(r^\prime)\chi_{\{e_r \leq \epsilon\}}(r)}{e^{2}_r} \;
    \\
    &\quad = C\int_{|r|< 1}\d r\,\chi_{\{e_r \leq \epsilon\}}(r) e_r^{-2}  \int_{\sqrt{x^2 - e_r}}^x \d s\,s^2 \int_{\mathbb{S}^2} \d \theta\, \chi_{C_{r,x}^{\mathrm{(in)}}}(r^\prime(s, \theta))\\
    & \quad 
    \leq C\int_{|r|< 1}\d r \, \chi_{\{e_r \leq \epsilon\}}(r) e_r^{-2}  \int_{\sqrt{x^2 - e_r}}^x \d s \, |C_{r,x}^{\mathrm{(in)}} \cap \partial B_s(0)|\;.
\end{align*}
(For the equality we used spherical coordinates for the integration with respect to $r^\prime$.)   
For fixed $ s\in (\sqrt{x^2 - e_r}, x) $, the surface portion inside $ C_{r,x}^{(\mathrm{in})} $ is a spherical cap of radius less than $x$ and opening angle $ \tilde\beta $ (see Figure~\ref{fig:Angles}). Hence,
\begin{equation*}
    |C_{r,x}^{(\mathrm{in})} \cap \partial B_s(0)|
    = 2 \pi s^2 (1 - \cos\tilde\beta)  \qquad \text{for all } s\in  (\sqrt{x^2 - e_r}, x)\;.
\end{equation*}
The spherical cap is limited by a circle at which $|r+r'| = (1-x)+ 2 e_r^{1/2}$. Then, by the cosine relations in the triangle with angle $\tilde\beta$ shown in Figure~\ref{fig:Angles},
\begin{equation}
\begin{aligned}\label{eq: 1-cosbetatilde}
    ((1-x)+  2 e_r^{1/2})^2
    &= (|r|-s)^2 + 2 |r| s (1 - \cos\tilde\beta) \;.
\end{aligned}
\end{equation}
Using that $e_r < \epsilon < \frac 34 x^2 < \frac 34$, we deduce $ |r| > \frac 12 $ and since $\frac 34 x^2 > e_r \ge \widetilde{e}_{r'}$, we further deduce $ s = |r'| > \frac x2 $. Also, we have $e_r = (1+|r|)(1-|r|) \ge (1-|r|)$ and $|r'| < x$, so $|r|-|r'| > 1-x-e_r$. Using these bounds in \eqref{eq: 1-cosbetatilde} together with the constraint $e_r < \epsilon$ with $\epsilon$ small enough, we get 
\begin{equation}\label{eq:Ad_bound}
    (1 - \cos\tilde\beta)\leq \frac{C}{x} e_r^{1/2} 
    \quad \Rightarrow \quad
    |C_{r,x}^{(\mathrm{in})} \cap \partial B_s(0)| \leq C x e_r^{1/2} \;.
\end{equation}
Thus, we conclude that the contribution from $C_{r,x}^{\mathrm{(in)}}$ to $\I_{3;1;1}^{<\epsilon}$ is bounded by 
\begin{equation*}
    C x \int_{|r|<1} \ud \r \, \chi_{\{e_r\leq\epsilon\}}(r) e_r^{-3/2} \abs{x - \sqrt{x^2-e_r}}
    \leq C \int_{|r|<1} \ud r \, e_r^{-1/2} \leq C \;.  
\end{equation*}
The estimate for $C_{r,x}^{(\mathrm{out})}$ is analogous: Also here, $C_{r,x}^{(\mathrm{out})} \cap \partial B_{s}(0)$ is a spherical cap.
In conclusion, we find
\[
	\mathrm{I}_{3;1}^{<\epsilon}\leq \mathrm{I}_{3;1,1}^{<\epsilon} + \mathrm{I}_{3;1;2}^{<\epsilon} 
    \leq C \;,
\]
which, combined with {\eqref{eq: I31 big eps}}, gives 
\begin{equation}
	\mathrm{I}_{3;1}^< \leq C {x^{-3}} \;.
\end{equation}

The estimate of $\mathrm{I}^<_{3;2}$ is obtained in the same way as for $\mathrm{I}_{3;1}^<$, after exchanging the roles of $e_{r+p}$ and $e_r$ and performing a suitable change of variables. Similarly, $\mathrm{I}_{3;3}^<$ can be handled by the same argument with the roles of $r$ and $r^\prime$, and of $e_{r}$ and $\widetilde{e}_{r^\prime}$ interchanged. Finally, the bound for $\mathrm{I}_{3;4}^<$ follows by adapting the analysis of $\mathrm{I}_{3;3}^<$ in the same way as the estimate of $\mathrm{I}_{3;1}^<$ is obtained from that of $\mathrm{I}_{3;2}^<$. This concludes the analysis of $\I_3$.
\end{proof}
\begin{lemma} \label{app:belyakov_integral_bound_q}
For $x \in (0,\infty)$, let
\begin{equation}\label{eq: def Is q}
	\I_2^{(q)}  
	:= \int_{\RRR^3} \d p
	\int_{|r| < 1 < |r+p|} \d r
	\int_{|r'| < x < |r'-p|} \d r'
	\frac{\delta(r-q) + \delta(r+p-q)}{\big( |r+p|^2 - |r|^2 + |r'-p|^2 - |r'|^2 \big)^2} \;.
\end{equation}
Then there exists $c,C > 0$ (independent of $x$) such that
\begin{equation}\label{eq: bound I2q}
    \frac{c x^3}{(1+x^4+|q|^4)(1+x)}
    \leq \mathrm{I}_2^{(q)}
    \leq \frac{{C (x^{-1} + x^3)}}{1+|q|^4} \qquad \text{for all } x\in(0,\infty)\;.
\end{equation}

\end{lemma}

\begin{proof} We distinguish the two cases $|q| \geq 3$ and $|q|< 3$.

\medskip

\noindent\underline{Case $|q| \geq 3$.} Here $|r|\leq 1$, so only the term $\delta(r+p-q)$ contributes, and therefore 
\[
    \mathrm{I}_{2}^{(q)} = \int_{\mathbb{R}^3}\d p  \int_{|r|< 1 <|r+p|} \d r\int_{|r^\prime| < x < |r^\prime - p|}\d r^\prime \frac{\delta(r+p-q) }{(e_r + e_{r+p} + \widetilde{e}_{r^\prime} + \widetilde{e}_{r^\prime -p})^2} \;.
\]
Using that $e_{r} + \widetilde{e}_{r^\prime -p} + \widetilde{e}_{r^\prime} \geq 0$, we obtain
\begin{align*}
    \mathrm{I}_{2}^{(q)}
     \leq C\int_{|r| < 1 < |q|}\d r\int_{|r^\prime| < x < |r^\prime - q+r|} \d r^\prime\frac{1}{e_q^2}
    \leq \frac{C}{|q|^4}\int_{|r| < 1} \mathrm{d}\r\int_{|r^\prime| < x}\d r^\prime \; 1
    \leq \frac{Cx^3}{|q|^4 } \;,
\end{align*}
where in the next to the last estimate we used that $e_q^2 \geq C|q|^4$ for $|q| \geq 3$. For a lower bound, we use that under the integration constraints on $r$, $r^\prime$, and $p$, we have $(e_{r+p} + e_r + \widetilde{e}_{r^\prime} + \widetilde{e}_{r^\prime-p})^2 \leq C(|r+p|^2 +|r^\prime - p|^2)^2 \leq C(x^2 + |r+p|^2)^2 \leq  C(x^4 +|r+p|^4)$, and therefore
\begin{align*}
    \mathrm{I}_2^{(q)}
    &\geq C \int_{\mathbb{R}^3}\d p\int_{|r| < 1 < |r+p|}\d r\int_{|r^\prime| < x < |r^\prime - p|}\d r^\prime \,\frac{\delta(r+p-q)}{{{x^4 +} |r+p|^4}}
\\
    &\geq  C\int_{|r| < 1}\d r\int_{|r^\prime| < x }\d r^\prime \frac{1}{{x^4 +} |q|^4}
    \geq \frac{C x^3}{{{x^4 +} |q|^4}} \;,
\end{align*}
where we used that for $|q|\geq 3$, the constraints we omitted in the $r$ and $r^\prime$ integration are always satisfied. This proves \eqref{eq: bound I2q} for $|q| \geq 3$.

\medskip

\noindent \underline{Case $|q| < 3$.} 
As a lower bound, we restrict the integration domain in $p$ to $1 < |p| < 4$. The constraint $|p| < 4$, together with $|r| < 1$ and $|r'| < x$ implies $|r+p|^2 - |r|^2 + |r'-p|^2 - |r'|^2 \leq C{(1 + x)^2}$. 
Thus, 
\begin{align*}
    \I_2^{(q)}
    & \geq c \int_{1 < |p| < 4} \ud p \int_{|r'| < x < |r-'p|} \ud r' \int_{|r| < 1 < |r+p|} \ud r \frac{\delta(r-q) + \delta(r+p-q)}{{(1 + x)^4}}
    \\
    & = c \int_{1 < |p| < 4} \d p
	\int_{|r'| < x < |r'-p|} \d r'
    \frac{\chi(|q| < 1 < |q+p|)
    + \chi(|q-p| < 1 < |q|) }{{(1 + x)^4}} \;.
\end{align*}
The integration domain for $ r' $ is the set $\{ r' : |r'| < x < |r'-p|\} = B_x(0) \setminus B_x(p)$, whose volume is lower-bounded by $|\{ r' : |r'| < x < |r'-p|\}| \geq c \min\{|p|,x\} x^2 \geq c \frac{x^3}{1+x}$, since $|p| > 1$. Therefore
\begin{equation}
\begin{aligned}
    \I_2^{(q)}  
    &\geq C x^3 \int_{1 < |p| < 4} \d p
    \frac{ \chi(|q| < 1 < |q+p|)
    + \chi(|q-p| < 1 < |q|) }{{(1 + x)^5}}
    \geq \frac{c x^3}{{(1 + x)^5}} \;.
\end{aligned}
\end{equation}
%where we used that for any $|q| < 3$, the integration domain in $p$ is of volume $\ge c$.
We conclude that, without restriction on $\lvert q\rvert$, we have 
\[
    \mathrm{I}^{(q)}_2 \geq \frac{c x^3}{(1+x^4+|q|^4)(1+x)} \;.
\]

It remains to show the upper bound on $\mathrm{I}_2^{(q)}$ in the case $|q| < 3$. As for $\I_2$, we distinguish the case $e_{r+p} \leq e_r$ and the case $e_r < e_{r+p}$ and denote the corresponding integrals by $\I^{(q)}_{2;1}$ and $\I^{(q)}_{2;2}$, i.\,e., 
\begin{equation*}
    \I^{(q)}_{2}
    = \I^{(q)}_{2;1} + \I^{(q)}_{2;2} \;.
\end{equation*}

\noindent\underline{Estimate for $\I^{(q)}_{2;1}$.} We split $\I^{(q)}_{2;1} = \I^{(q=r)}_{2;1} + \I^{(q=r+p)}_{2;1}$ where
\begin{align*}
	\I^{(q=r)}_{2;1} 
    & := 
    \int_{\mathbb{R}^3}\d p\int_{|r|< 1 < |r+p|} \d r\int_{|r^\prime| < x < |r^\prime -p|}\d r^\prime \frac{\delta(r-q)\chi(e_{r+p} \leq e_r)}{(|r+p|^2 - |r|^2 + |r^\prime- p|^2 - |r^\prime|^2)^2} \;,
    \\
    \I^{(q=r+p)}_{2;1} 
    & := 
    \int_{\mathbb{R}^3}\d p\int_{|r|< 1 < |r+p|} \d r\int_{|r^\prime| < x < |r^\prime -p|}\d r^\prime \frac{\delta(r+p-q)\chi(e_{r+p} \leq e_r)}{(|r+p|^2 - |r|^2 + |r^\prime- p|^2 - |r^\prime|^2)^2} \;.
\end{align*}
\noindent\underline{Estimate for $\I^{(q=r)}_{2;1}$.} $\mathrm{I}_{2;1}^{(q=r)}$ is estimated similarly as~\eqref{eq: Is<}:
\begin{equation*}
\begin{aligned}
    \I^{(q=r)}_{2;1}
    &= \int_{\RRR^3} \d p
	\int_{|r'| < x < |r'-p|} \d r'
	\frac{\chi(|q| < 1 < |q+p|) \chi(e_{q+p} \leq e_q)}{\big( e_q + e_{q+p} + \widetilde{e}_{r'} + \widetilde{e}_{r'-p} \big)^2} 
   \\
    &\leq C \int_{|r'| < x} \d r'
	\frac{e_q}{e_q^2 + x^2 (x-|r'|)^2 }
    \leq Cx \;.
\end{aligned}
\end{equation*}
\noindent\underline{Estimate for $\I^{(q= r+p)}_{2;1}$.} For $\I^{(q=r+p)}_{2;1}$, the above argument does not work since the integration volume $|\{p : e_q \leq e_{q-p}\}|$ is no longer of order $e_q$
%, but rather $1$. So
and the above bound blows up as $e_q \to 0$. We need a more refined analysis.

If $|r|< 1/2$, then the integral is finite since 
\begin{align*}
    \I^{(q=r+p)}_{2;1} \Big\vert_{|r| < \frac 12}
    &\leq \int_{|r|< \frac 12}\d r
        \int_{|r^\prime| < x < |r^\prime -(q-r)|}\d r^\prime \frac{1}{{(e_r + e_q + \widetilde{e}_{r'} + \widetilde{e}_{r'-(q-r)} )^2}} \\
    &\leq C\int_{|r| < \frac 12}\d r
        \int_{|r^\prime| < x}\d r^\prime
        \frac{1}{e_r^2}
    \leq Cx^3 \int_{|r| < \frac 12} \d r \frac{1}{(1-|r|^2)^2}
    \leq Cx^3.
\end{align*}

So in the following we may restrict to the case $1/2 < |r| <1$.
We further divide the problem into three cases according to which of the quantities $e_{r}$, $\widetilde{e}_{r^\prime}$, and $\widetilde{e}_{r^\prime - p}$ is the largest, defining
\begin{equation}
\begin{aligned}\label{eq: regions S123qrx}
    S_{x}^{(1)}
    &:= \{ r',r,p \in \RRR^3 \; : \; \max(\widetilde{e}_{r'}, \widetilde{e}_{r'-p}) \leq e_r \} \;, \\
    S_{x}^{(2)} 
    &:= \{  r',r,p \in \RRR^3 \; : \; \max (e_r, \widetilde{e}_{r'-p}) \leq \widetilde{e}_{r'} \} \;, \\
    S_{x}^{(3)} 
    &:= \{  r',r,p \in \RRR^3 \; : \; \max(e_r, \widetilde{e}_{r'}) \leq \widetilde{e}_{r'-p} \} \;,
\end{aligned}
\end{equation}
and correspondingly writing
\begin{equation}
\label{eq: def int sjqrx}
    \I_{2;1}^{(q=r+p)} \Big\vert_{|r| > \frac 12}
    = \I_{2;1;1}^{(q=r+p)} 
        + \I_{2;1;2}^{(q=r+p)} 
        + \I_{2;1;3}^{(q=r+p)} \;,
\end{equation}
with
\begin{align}
    \I_{2;1;j}^{(q=r+p)}
    & := \int_{\mathbb{R}^3}\d p \int_{ 1/2 < |r|< 1 < |r+p| } \d r \int_{ |r^\prime| < x < |r^\prime -p| } \d r^\prime \, \chi_{S^{(j)}_x}(p,r,r^\prime) \times \nonumber \\
    & \qquad \times \frac{\delta(r+p-q)\chi(e_{r+p} \leq e_r)}{(|r+p|^2 - |r|^2 + |r^\prime- p|^2 - |r^\prime|^2)^2} \;.
\end{align}
Note that since $r+p =q$, the constraint $|r+p| = |q| \ge 1$ is automatically satisfied.

\medskip
\noindent \underline{Bound for $\I_{2;1;1}^{(q=r+p)}$.} In this case, $e_r \geq \widetilde{e}_{r^\prime}$ and $e_r \geq \widetilde{e}_{r^\prime - p}$. Due to the $\delta (r+p-q)$ in the integrand we have $p = r-q$, so that the constraints read as
\[
    x^2 - e_r \leq |r^\prime|^ 2 \leq x^2\;, \qquad x^2 \leq | r^\prime - (q-r)|^2 \leq x^2 + e_r \;.
\]
Thus,
%we can rewrite the region $S_{x}^{(1)}$ with respect to $p=r-q$ and obtain that
the integration domain for the variable $r^\prime$ is
\[ 
    S_{q,r,x}^{(1)}= \{ r' \in \mathbb{R}^3\; : x^2 - e_r \leq |r^\prime|^ 2 \leq x^2 \text{ and }  x^2 \leq | r^\prime - (q-r)|^2 \leq x^2 + e_r \} \;.
\]
Similarly as above, we distinguish the case\footnote{The condition $|r|> 1/2$ already implies $e_r \leq 3/4$. However, we want to find a bound on $\mathrm{I}_2^{(q)}$ valid for any $x\in (0,\infty)$ and thus it is convenient to distinguish also between $e_r\leq (3/4) x^2$ and $e_r > (3/4) x^2$.} $e_r \leq (3/4)x^2$ and the case $e_r > (3/4) x^2$. In the latter case, the integral can be estimated by 
\begin{align*}
    \mathrm{I}_{2;1;1}^{(q= r+p)} \Big\vert_{e_r > \frac 34 x^2}
    &\leq \int_{\frac 12 < |r| < 1}\d r \int_{S^{(1)}_{q,r,x}}\d r^\prime \, \frac{\chi ( e_r \geq (3/4)x^2)}{(e_r + e_q + \widetilde{e}_{r^\prime} + \widetilde{e}_{r^\prime - (q-r)})^2} \\
    &\leq \int_{|r| < 1 }\d r \int_{|r^\prime| < x}\d r^\prime \, \frac{C}{x^4} 
    \leq \frac{C}{x} \;.
\end{align*}
Instead, in the former case $e_r \leq (3/4) x^2$, we estimate $\mathrm{I}_{2;1;1}^{(q= r+p)}$ as follows: we claim that
\begin{equation}\label{eq: volume S1}
    |S^{(1)}_{q,r,x}| 
    \leq C \left(1+\frac{1}{x}\right) \frac{e_r^2}{|q-r|} \;.
\end{equation}
In fact, note that $S^{(1)}_{q,r,x}$ is an intersection of two annuli, one shifted with respect to the other by $ p = q-r$, and both of thickness proportional to $e_r x^{-1}$. We estimate $|S^{(1)}_{q,r,x}|$ differently according to the three cases $|q-r|\leq 2e_r x^{-1}$, $2 e_r x^{-1} \leq |q-r | \leq 2x - 2e_r x^{-1}$, and $|q-r| \geq 2x - 2e_r x^{-1}$, which correspond to different types of intersection of the annuli, see Figure~\ref{fig:both}.

\begin{figure}[h]
\begin{subfigure}[t]{0.48\textwidth}
 	\centering
	\scalebox{0.9}{\input{S1_bound_1.tex}}
\caption{In the first case we have $|q-r| \leq 2e_r x^{-1}$. We bound the volume of $S^{(1)}_{q,r,x}$ by the volume of one annulus.} 
\label{fig:S1_boundL}
\end{subfigure}
\hfill
\begin{subfigure}[t]{0.48\textwidth}
 	\centering
	\scalebox{0.9}{\input{S1_bound_2.tex}}
\caption{In the second case $2 e_r x^{-1} \leq |q-r| \leq 2x - 2e_r x^{-1}$. We use the angle $\beta$ between the tangents to bound the volume of $S^{(1)}_{q,r,x}$.} 
\label{fig:S1_boundR}
\end{subfigure}

\bigskip

\begin{subfigure}{\textwidth}
 	\centering
	\scalebox{0.9}{\input{S1_bound_3.tex}}
\caption{In the third case we have $|q-r| \geq 2x - 2 e_r x^{-1}$.
% The set $S^{(1)}_{q,r,x}$ has thicknesses $\delta$ and $\gamma$, which we express using the angle $\alpha$.
The volume of $S^{(1)}_{q,r,x}$ is constrained by the two lengths $\delta$ and $\gamma$, which we express using the angle $\alpha$.
}
\label{fig:S1_bound_3}
\end{subfigure}
\caption{Illustration of the three cases encountered in estimating the volume of the set $S_{q,r,x}^{(1)}$, which is an intersection of two annuli, one shifted with respect to the other by the vector $p = q-r$.}
\label{fig:both}
\end{figure}

In the first case $|q-r| \leq 2e_r x^{-1}$ (shown in \Cref{fig:S1_boundL}), we estimate the volume of $S^{(1)}_{q,r,x}$ with the volume of the annulus determined by the condition $\widetilde{e}_{r^\prime}>e_r$, which gives 
\begin{equation}\label{eq: volume S1 case 1}
    |S^{(1)}_{q,r,x}| \leq C\frac{e_r}{x} 4\pi x^2 \leq C \frac{e_r^2}{|q-r|} \;, \qquad
   {\mbox{if}\,\,\, |q-r| \leq 2 \frac{e_r}{x}} \;.
\end{equation}

In the second case $2 e_r x^{-1} \leq |q-r| \leq 2x - 2e_r x^{-1}$, the set $S_{q,r,x}^{(1)}$ is the rotated solid of the \emph{two-dimensional} intersection of two \emph{two-dimensional annuli} as in \Cref{fig:S1_boundR}. 
Hence, its volume is bounded by the radius of the rotation times the area of the two-dimensional intersection. Let $\alpha \in (0,\frac{\pi}{2})$ and $\beta \in (0,\pi)$ as in \Cref{fig:S1_boundR}. The rotation is with a circle of radius $\leq C x \sin \alpha$ with $\alpha=\frac{\pi - \beta}{2}$ and the two-dimensional area is $\leq C e_r^2 x^{-2} (\sin\beta)^{-1}$, since the width of the two annuli are $e_r/x$. 
Thus, $|S_{q,r,x}^{(1)}| \leq C x^{-1} e_r^2 \sin \frac{\pi - \beta}{2} (\sin \beta)^{-1}$.
From the cosine relation for the triangle in \Cref{fig:S1_boundR} we have $2x^2 \cos {\beta} = 2x^2 - |q-r|^2$, which with $\cos \beta = 1 - 2 \sin^2 \frac{\beta}{2}$ implies $\sin \frac{\beta}{2} = \frac{|q-r|}{2x}$. 
From this and elementary trigonometric identities we have 
\begin{equation} \label{eq: volume S1 case 2}
    |S_{q,r,x}^{(1)}| \leq C e_r^2 \frac{\sin \frac{\pi - \beta}{2}}{x \sin \beta} 
    = C e_r^2 \frac{\cos \frac{\beta}{2}}{2 x \sin \frac{\beta}{2} \cos \frac{\beta}{2}} 
    = C \frac{e_r^2}{|q-r|} \;. 
\end{equation}

In the third case $|q-r| \geq 2x - 2e_r x^{-1}$, let $\delta$ be the radial displacement (or shell thickness) along the line the line connecting the centers of the two annuli and $\gamma$ the thickness of the intersection measured perpendicular to the line connecting the centers, see \Cref{fig:S1_bound_3}. Then 
\[
    x - x \cos \alpha \leq \delta \leq Ce_r x^{-1} \;,
\]
which implies $\alpha\leq C e_r^{1/2} x^{-1}$, {using $1-\cos \alpha = 2 \sin^2(\alpha/2) \geq (2/\pi^2) \alpha^2$ with $\alpha\in (0,\pi)$}. With this bound on $\alpha$, we deduce that $\gamma \leq C\alpha x \leq C e_{r}^{1/2}$. All together this implies that we can estimate the volume of the intersection of the two annuli by 
\begin{equation}\label{eq: volume S1 case 3}
    |S^{(1)}_{q,r,x}| \leq C \delta \gamma^2 \leq C\frac{e_r^{2}}{x}\leq C\frac{e_r^2}{|q-r|} \;,\qquad
    \mbox{if}\,\,\, |q-r| \geq 2x - {2}e_r x^{-1}  \;,
\end{equation}
where the last bound is a consequence of $e_r \leq (3/4) x^2$ and of the fact that $S_{q,r,x}^{(1)}$ is nonempty only if  
\[
    |q-r| \leq 2x + 2e_r x^{-1} \;.
\]
Indeed, if  $|q-r| \geq 2x + 2e_rx^{-1}$ the two annuli in $S^{(1)}_{q,r,x}$ are disjoint. Hence, 
\[
    |q-r| \leq 2x + 2 e_r x^{-1} \leq C x \;.
\]

Combining \eqref{eq: volume S1 case 1}, \eqref{eq: volume S1 case 2}, and~\eqref{eq: volume S1 case 3}, we proved \eqref{eq: volume S1}. We now use this bound to estimate $\mathrm{I}_{2;1;1}^{(q=r+p)}$. Recalling that we are considering only the values of $r$ such that $1/2 < |r| <1$, before applying \eqref{eq: volume S1}, we estimate further 
\[
    |q-r| \geq c(e_r + \theta) \;, 
\]
for some constant $c>0$ and with $\theta \in (0,\pi)$ being the angle enclosed by $q$ and $r$. The bound above can be proved using that 
\[
    |q-r|^2 \geq 2|q||r| (1-\cos\theta) \geq 1-\cos\theta = 2\sin^2\left(\frac{\theta}{2}\right) \geq C \theta^2  
\]
and 
\[
   |q-r|\geq   ||q|- |r||\geq 1- |r| = \frac{e_r}{1+|r|}\geq \frac{e_r}{2} \;.
\]
Using \eqref{eq: volume S1}, spherical integration and that $ \sin \theta \leq \theta$ {together with the positivity of $e_r$ which implies $\sin \theta (e_r + \theta)^{-1} \leq 1$}, we conclude
\[
    \mathrm{I}_{2;1;1}^{(q= r+p)} \Big\vert_{e_r \leq \frac 34 x^2}
	\leq C\left(1+\frac{1}{x}\right) \int_{\frac{1}{2}}^1 \d |r|
	\int_0^\pi \d \theta \; \sin \theta
	\frac{e_r^2 (e_r + \theta)^{-1}}{e_r^2}
	\leq C\left( 1 + x^{-1} \right)\;.
\]

\noindent \underline{Bound for $\I_{2;1;2}^{(q=r+p)}$}. Here, {we estimate the integral in \eqref{eq: def int sjqrx} with $j=2$. In this case it holds that $\widetilde{e}_{r^\prime} > \widetilde{e}_{r^\prime - p}$ and $\widetilde{e}_{r^\prime} \geq e_r$. The analysis is similar to the one of the term $\mathrm{I}_{3;1}^<$ introduced in \eqref{eq: int I3<}.} {Similarly as above}, we treat the case $\widetilde{e}_{r^\prime}\geq (3/4) x^2$ and the case $\widetilde{e}_{r^\prime} \leq (3/4) x^2$ separately.

In the first case $\widetilde{e}_{r^\prime}> (3/4) x^2$, using $\delta(r+p-q)$, we fix $p= q-r$. The condition $\widetilde{e}_{r^\prime}> (3/4) x^2$ implies that $|r'| \leq \frac x2$, so
\begin{equation} \label{eq:bound_I212_smallrprime}
    \I_{2;1;2}^{(q=r+p)} \Big\vert_{\widetilde{e}_{r'} > \frac 34 x^2}
    \leq \int_{\frac 12 < |r| < 1} \d r \int_{|r'| \leq \frac x2} 
        \frac{\d r'}{\widetilde{e}_{r'}^2}
    \leq {\int_{\frac 12 < |r| < 1} \d r \int_{|r'| \leq \frac x2} \d r'\, \frac{C}{x^4}}
    \leq C x^{-1} \;.
\end{equation}

We now discuss the second case $\widetilde{e}_{r^\prime} \leq (3/4) x^2$. The analysis in this case is analogous to that of $\I_3$ in the case where $e_{r+p}$ dominates, except that the dominating energy is now $\widetilde{e}_{r'}$ and not $e_{r+p}$. It is convenient to compute the integral with respect to $r$ first. The Dirac delta then fixes $r=q-p$. In this case the condition $e_r \geq e_{r+p}$ becomes $e_{q-p} \geq e_q$; the constraints in $S^{(2)}_x$ ($\widetilde{e}_{r^\prime}\geq e_{r}$ and $\widetilde{e}_{r^\prime}\geq \widetilde{e}_{r^\prime -p}$), become $\widetilde{e}_{r^\prime}\geq e_{q-p}$ and $\widetilde{e}_{r^\prime}\geq \widetilde{e}_{r^\prime - p}$, respectively. Consequently we obtain the constraints
\[
    x^2 < |r'-p|^2 < x^2 + \widetilde{e}_{r'} \;, \qquad  1-\widetilde{e}_{r'} < |q-p|^2 < 1 - e_q \;.
\]
Accordingly we introduce the annuli
\begin{align*}
 \widetilde{A}_{r',x}
    & := \left\{ p \in \RRR^3 \, : \, x^2 < |r'-p|^2 < x^2 + \widetilde{e}_{r'} \right\} \;, \\
    \widetilde{A}_{q,r'} 
    & := \left\{ p \in \RRR^3 \, : \, 1-\widetilde{e}_{r'} < |q-p|^2 < 1 - e_q \right\} \;.
\end{align*}
Thus,
%in the case $\widetilde{e}_{r^\prime} \leq \frac 34 x^2$, 
\[
 \mathrm{I}_{2;1;2}^{(q= r+p)} \Big\vert_{\widetilde{e}_{r^\prime} \leq \frac 34 x^2} = \int_{\mathbb{R}^3} \d p\int_{|r^\prime| < x < |r^\prime - p|}\d r^\prime\, \frac{\chi_{\widetilde{A}_{r^\prime, x}\cap \widetilde{A}_{q,r^\prime}}(p) \chi_{\{|p-q| < 1 < |q|\}}(p)}{(|q|^2 - |q-p|^2 + |r^\prime- p|^2 - |r^\prime|^2)^2} \;.
\]
In this case, $\widetilde e_{r'} \ge e_q$ implies $|r'| \le \sqrt{x^2 - e_q}$. Hence, the only potentially problematic regime corresponds to $e_q \to 0$, in which case $|r'| \to x^2$ and the integral develops a singularity. For fixed $e_q$, one can derive a bound which however diverges as $e_q \to 0$. So a more careful analysis is required, based on decomposing $A_{r',x} \cap A_{q,r'}$ into a “dangerous” and a “non-dangerous” region similarly as we already did above.
The volume of $\widetilde{A}_{r',x}\cap \widetilde{A}_{q,r'}$ depends on the position of $r'$ relative to $q$. Similarly to~\eqref{eq: def Crx}, we introduce the following set where $|\widetilde{A}_{r',x}\cap \widetilde{A}_{q,r'}|$ gets large:
\begin{equation}\label{eq:Cqx}
\begin{aligned}
    \widetilde{C}_{q,x} 
    &:= \widetilde{C}_{q,x}^{(\mathrm{in})} \cup \widetilde{C}_{q,x}^{(\mathrm{out})} \;, \\
    \widetilde{C}_{q,x}^{(\mathrm{in})}
    &:= \left\{ r' \in \RRR^3\, : \, |r'| < x,\, ||q-r^\prime| - {\abs{1-x}}| \leq 2 \widetilde{e}_{r'}^{1/2} \right\} \;, \\
    \widetilde{C}_{q,x}^{(\mathrm{out})}
    &:= \left\{ r' \in \RRR^3\, : \, |r'| < x,\, ||q-r^\prime| - (1+x)| \leq 2 \widetilde{e}_{r'}^{1/2} \right\} \;.
\end{aligned}
\end{equation}
Since $\widetilde{e}_{r'}$ depends on $r'$, these are not exactly annuli as was the case for the similarly constructed $C_{r,x}$ in \eqref{eq: def Crx} above.
%As we shall see, however, we can still use these to bound the desired integrals.
Note also that possibly $x>1$.

If $ r' \notin \widetilde{C}_{q,x} $, we use the same arguments as for~\eqref{eq: estimate area ar arprime nondangerous} while replacing $r+r'$ by $q-r'$: Similarly as in Figure~\ref{fig:Angles}, we define $\beta$ as the angle between the tangents at the intersection of $\partial B_1(0)$ and $\partial B_x(q-r')$. The thicknesses of the annuli $\widetilde{A}_{r',x}$ and $\widetilde{A}_{q,r'}$ are now $\widetilde{e}_{r'}/x$ and $\widetilde{e}_{r'}$ respectively (recall $e_q \leq \widetilde{e}_{r'}$), so we have
\begin{equation}
    |\widetilde{A}_{r^\prime,x}\cap \widetilde{A}_{q,r^\prime}| \leq C x^{-1} \widetilde{e}_{r'}^2 (\sin\beta)^{-1} \;.
\end{equation}
The relation~\eqref{eq:beta_cosine_relation} is still true with $e_r$ replaced by $\widetilde{e}_{r'}$, with $r+r'$ replaced by $q-r'$ {and $(1-x)$ replaced by $|1-x|$}. So as in~\eqref{eq: lower bound beta}, but with $\widetilde{e}_{r'} \leq (3/4) x^2$, we conclude
\begin{align*}
    \beta^2
    &> \frac{2 \widetilde{e}_{r'}^{1/2} (2|1-x| + 2 \widetilde{e}_{r'}^{1/2})}{x}
    >  c\frac{\widetilde{e}_{r'}}{x}\;, 
    \\
    (\pi - \beta)^2
    & > \frac{2 \widetilde{e}_{r'}^{1/2} (2(1+x) - 2 \widetilde{e}_{r'}^{1/2})}{x} 
    > c \frac{\widetilde{e}_{r'}}{x}\;,
\end{align*}
so $ \sin\beta > c  x^{-1/2} \widetilde{e}_{r'}^{1/2} $ and therefore
\begin{equation}\label{eq: est not danger i212q}
    |\widetilde{A}_{r^\prime,x}\cap \widetilde{A}_{q,r^\prime}| \leq C 
    {x^{- 1/2}} \widetilde{e}^{3/2}_{r^\prime} \qquad \text{if } r^\prime \in \{r^\prime\in\mathbb{R}^3 \,:\, |r^\prime| < x\}\setminus \widetilde{C}_{q,x} \;.
\end{equation}
Thus,
%in case $\widetilde{e}_{r'} \leq (3/4) x^2$,
for $r^\prime \notin \widetilde{C}_{q,x}$, using \eqref{eq: est not danger i212q} we get the bound 
\begin{equation*}
\begin{aligned}
    \I_{2;1;2}^{(q=r+p)} \Big\vert_{r' \notin \widetilde{C}_{q,x}}
    & \leq \int_{|r'| \leq \sqrt{x^2 - e_q}} \d r' \; \chi_{\widetilde{C}_{q,x}^c}\!(r')
        \int_{\widetilde{A}_{r',x} \cap \widetilde{A}_{q,r'}} 
	\frac{\d p}{\big( e_{q-p} + e_q + \widetilde{e}_{r'} + \widetilde{e}_{r'-p} \big)^2} \\
    &\leq C x^{- 1/2} \int_{|r'| \leq \sqrt{x^2 - e_q}} \d r' \; \chi_{\widetilde{C}_{q,x}^c}\!(r')
        \frac{\widetilde{e}_{r'}^{3/2}}{\widetilde{e}_{r'}^2} 
    \leq C x^{- 1/2} \int_{|r'| < x} \!\!\d r' \widetilde{e}_{r'}^{-1/2} 
    \leq C x \,.
\end{aligned}
\end{equation*}

For $r' \in \widetilde{C}_{q,x}$, {we consider first the case $x\leq1$ and} proceed as around~\eqref{eq:Ad_bound}: Recall that $|r'| \geq \frac x2$, since we are in the case $\widetilde{e}_{r'} \leq (3/4) x^2$. We use $|\widetilde{A}_{r^\prime,x}\cap \widetilde{A}_{q,r^\prime}|\leq |\widetilde{A}_{r^\prime,x}| \leq C \widetilde{e}_{r'}$, and consider the sets $\widetilde{C}^{(\mathrm{in})}_{q,x} \cap \partial B_s(0)$ and $\widetilde{C}^{(\mathrm{out})}_{q,x} \cap \partial B_s(0)$ for fixed $s \in [\frac x2,\sqrt{x^2-e_q})$. Both are spherical caps of some opening angle, called $\tilde\beta$, and of area $2 \pi s^2 (1 - \cos\tilde\beta)$. The argument around \eqref{eq:Ad_bound} still applies, even though $\widetilde{C}_{q,x}^{\mathrm{(in)}}$ is not an annulus. 
This is because, for fixed $|r'|=s$, the number $\widetilde{e}_{r'}$ is in fact constant $\widetilde{e}_{r'} = x^2-s^2$, and thus the analysis proceeds exactly as for $C_{r,x}$ defined in \eqref{eq: def Crx} above.
(Note that $\tilde\beta$ depends on $\widetilde{e}_{r'}$ and therefore on $s$.)
For $\widetilde{C}^{(\mathrm{in})}_{q,x}$, the relation~\eqref{eq: 1-cosbetatilde} is still true with $|r|$ replaced by $|q|$ and $e_r$ replaced by $\widetilde{e}_{r'}$. Using $|q| > 1$, $s \geq \frac x2$ and $\lvert q\rvert -s \geq 1-x$ and $\widetilde{e}_{r^\prime} \leq (3/4)x^2 \leq 1$, we arrive as in~\eqref{eq:Ad_bound} at
\begin{equation}
    1 - \cos\tilde\beta
    \leq C \frac{\widetilde{e}_{r'}^{\frac 12}}{x}
    \quad \Rightarrow \quad
    |\widetilde{C}^{(\mathrm{in})}_{q,x} \cap \partial B_s(0)|
    \leq C x\widetilde{e}_{r'}^{\frac 12} \;.
\end{equation}
Analogously, for $\widetilde{C}^{(\mathrm{out})}_{q,x}$, we get
\begin{equation}
    1 - \cos\tilde\beta
    \leq C \frac{\widetilde{e}_{r'}^{\frac 12}}{x}
    \quad \Rightarrow \quad
    |\widetilde{C}^{(\mathrm{out})}_{q,x} \cap \partial B_s(0)|
    \leq C x\widetilde{e}_{r'}^{\frac 12} \;.
\end{equation}

Consider next the case $x > 1$. We proceed similarly as above, the only difference is the estimate for the set $\widetilde C_{q,x} \cap \partial B_s(0)$. The set $\widetilde C_{q,x}^{\mathrm{(in)}} \cap \partial B_s(0)$ is an annulus on the sphere $\partial B_s(0)$ delimited by the opening angles $\widetilde{\beta}_>$ and $\widetilde{\beta}_<$, see \Cref{fig:spherical-cap-x-geq-1}.
\begin{figure}[h]
 	\centering
	\scalebox{0.9}{\input{angle-beta-x-ge-1.tex}}
\caption{The figure shows in light green the set $\widetilde C_{q,x}^{\mathrm{(in)}}$ as defined in \eqref{eq:Cqx}, and in dark green its intersection with the sphere $\partial B_s(0)$ for the case $x>1$. In particular, it shows how to define the angles $\widetilde{\beta}_<$ and $\widetilde{\beta}_>$ that we employ to estimate $\abs{\widetilde C_{q,x}^{\mathrm{(in)}} \cap \partial B_s(0)}$.}
\label{fig:spherical-cap-x-geq-1}
\end{figure}
Thus, for $x>1$,
\begin{equation*}
    \abs{\widetilde C_{q,x}^{\mathrm{(in)}} \cap \partial B_s(0)} = 2 \pi s^2 (\cos \widetilde{\beta}_< - \cos \widetilde{ \beta}_>) \;.
\end{equation*}
From the cosine relations we have $\cos \widetilde{\beta}_\# = \frac{s^2 + |q|^2 - \ell_\#^2}{2 s |q|}$ for $\#\in\{<,>\}$ and $\ell_\#$ as in \Cref{fig:spherical-cap-x-geq-1}. Using that $s\geq x/2$ and that $|q|>1$ we get $\cos \widetilde{\beta}_< - \cos \widetilde{ \beta}_> \leq 8 \widetilde{e}_{r'}^{1/2}$ and thus
$\abs{\widetilde C_{q,x}^{\mathrm{(in)}} \cap \partial B_s(0)} \leq 8 x^2 \widetilde{e}_{r'}^{1/2}$.
The bound for $\widetilde{C}_{q,r}^{\mathrm{(out)}}\cap \partial B_s(0)$ is exactly as for the case $x \leq 1$. 
We conclude the final bound, valid for any $x > 0$, of 
\begin{equation}\label{eq: est danger i212q}
    |\widetilde{C}_{q,x} \cap \partial B_s(0)| 
     \leq C (x + x^2) \widetilde{e}_{r'}^{1/2} 
     \qquad \text{for all } s\in \left( {\frac x2}, \sqrt{x^2 - e_q} \right)
     \;.
\end{equation}
Thus,
\begin{align*}
\I_{2;1;2}^{(q=r+p)}\bigg\vert_{\widetilde C_{q,x}}
    & = \int_{x/2 \leq |r'| \leq \sqrt{x^2-e_q}} \ud r' \, \chi_{\widetilde{C}_{q,x}}(r')
	\int_{\widetilde{A}_{r'} \cap \widetilde{A}_{q,r'}} 
	\frac{\d p}{\big( e_{q-p} + e_q + \widetilde{e}_{r'} + \widetilde{e}_{r'-p} \big)^2} 
    \\ & 
    \leq \int_{x/2}^{\sqrt{x^2-e_q}} \ud |r'| \, \abs{\widetilde{C}_{q,x}\cap \partial B_{|r'|}(0)} \abs{\widetilde{A}_{r',x} \cap \widetilde{A}_{q,r'}} \widetilde{e}_{r'}^{-2} 
    \\ & 
    \leq C (x+ x^2) \int_{x/2}^{\sqrt{x^2-e_q}} \widetilde{e}_{r'}^{-1/2} \ud |r'|
    \leq C (x+x^2) \;.
\end{align*}

Summarizing, we have shown that
\begin{align*}
    \I_{2;1;2}^{(q=r+p)}  & = \I_{2;1;2}^{(q=r+p)} \Big\vert_{\widetilde{e}_{r'} > \frac 34 x^2} + \I_{2;1;2}^{(q=r+p)} \Big\vert_{\widetilde{e}_{r'} \leq \frac 34 x^2} \\
    & = \I_{2;1;2}^{(q=r+p)} \Big\vert_{\widetilde{e}_{r'} > \frac 34 x^2} + \I_{2;1;2}^{(q=r+p)}\bigg\vert_{\widetilde C_{q,x}} + \I_{2;1;2}^{(q=r+p)} \Big\vert_{r' \notin \widetilde{C}_{q,x}} \leq C(x+x^2 + x^{-1}) \;.
    %\leq C (x^{-2} + x^3) \;.
\end{align*}

\noindent \underline{Bound for $\I_{2;1;3}^{(q=r+p)} $.} Proceeding analogously to the estimates for $\I_{2;1;2}^{(q=r+p)}$, we obtain $\I_{2;1;3}^{(q=r+p)} \leq C(x+x^2 + x^{-1})$. We omit the details.

\medskip

We conclude that there exists $C > 0 $ such that for all $x \in (0,\infty)$ we have
\[
    \I_{2;1}^{(q=r+p)} \leq C (x^{-1} + x^3) \;.
\]
The contribution $\I^{(q=r)}_{2;2}$ can be estimated like $\I^{(q=r+p)}_{2;1}$, and $\I^{(q=r+p)}_{2;2}$ like $\I^{(q=r)}_{2;1}$.
\end{proof}

\section*{Acknowledgments}
Niels Benedikter was supported by the European Union through the ERC StG \textsc{FermiMath}, grant agreement nr.~101040991. Sascha Lill was supported by the European Union, partially by the ERC StG \textsc{FermiMath}, grant agreement nr.~101040991, and partially by the ERC AdvG MathBEC nr.~101095820. Views and opinions expressed are however those of the authors only and do not necessarily reflect those of the European Union or the European Research Council Executive Agency. Neither the European Union nor the granting authority can be held responsible for them. Emanuela L. Giacomelli was partially funded by  Deutsche Forschungsgemeinschaft (DFG, German Research Foundation) via TRR 352 -- Project-ID 470903074. Moreover Niels Benedikter, Sascha Lill, and Emanuela L. Giacomelli were partially supported by Gruppo Nazionale per la Fisica Matematica in Italy. Asbjørn Bækgaard Lauritsen was supported by the French State support managed by ANR under the France 2030 program through the MaQuI CNRS Risky and High-Impact Research programme (RI)$^2$ (grant agreement ANR-24-RRII-0001).

\section*{Statements and Declarations}
The authors have no competing interests to declare.

\section*{Data Availability}
As purely mathematical research, there are no datasets related to this article.

\renewcommand{\bibfont}{\footnotesize}
\printbibliography

\end{document}

%% file: Shell_intersection.tex
\begin{tikzpicture}

% big Fermi ball
\draw[thick] (0,0) circle (2.5) ;
\draw[dashed, red!50!blue] (0,0) circle (2.6) ;
\fill[red!50!blue, opacity = .2] (0,2.6) arc[start angle=90, end angle=450, radius=2.6] -- (0,2.5) arc[start angle=90, end angle=-270, radius=2.5];
\fill (0,2.4) circle (0.05) node[anchor = north west] {\footnotesize $ r $};
\fill (0,0) circle (0.05) node[anchor = west] {\footnotesize $ 0 $};
\draw[thick] (0,2.4) -- (0,0);
\draw[red!50!blue] (1.1,2.3) -- ++(0.3,0.2) node[anchor = west] {\footnotesize $ A_r + r $};

%radius
\draw (0,0) -- (-2,-1.5);
\draw (0,0) -- ++(0.12,-0.16);
\draw (-2,-1.5) -- ++(0.12,-0.16);
\draw[<->] (0.06,-0.08) --node[anchor = south east] {\footnotesize $ 1 $} ++(-2,-1.5);

% allowed region for r'
%\draw[dotted] (0,2.4) circle (1) ;
%\draw[dotted] (0,2.4) circle (0.75) ;
%\fill[green!50!black, opacity = .1] (0,3.4) arc[start angle=90, end angle=450, radius=1] -- (0,3.15) arc[start angle=90, end angle=-270, radius=0.75];
%\draw[green!50!black, <-] (0.2,3.2) -- ++(0.6,0.2) node[anchor = west] {\footnotesize allowed region};
%\node[green!50!black] at (2,3) {\footnotesize for $ r' $};

% small Fermi ball
\fill[thick] (-0.7,1.9) circle (0.05) node[anchor = east] {\footnotesize $ r+r' $};
\draw[thick] (0,2.4) -- (-0.7,1.9);
\draw[thick] (-0.7,1.9) circle (1);
\draw[dashed, red!50!blue] (-0.7,1.9) circle (1.15);
\fill[red!50!blue, opacity = .2] (-0.7,3.05) arc[start angle=90, end angle=450, radius=1.15] -- (-0.7,2.9) arc[start angle=90, end angle=-270, radius=1];
\draw[red!50!blue] (-1.6,2.6) -- ++(-0.3,0.2) node[anchor = east] {\footnotesize $ A_{r,r',x} + r $};
%\draw[red!50!blue, -] (-1.8,1.8) -- ++(-0.4,0.2) node[anchor = east] {\footnotesize allowed region};
%\node[red!50!blue] at (-3.5,1.6) {\footnotesize for $ p $};
%\node[red!50!blue] at (-3.5,1.2) {\footnotesize $ A^{(1)}_{r,r'} \cap A^{(2)}_{r,r'} + r' $};

%radius
\draw (-0.7,1.9) -- ++(-0.8,-0.6);
\draw (-0.7,1.9) -- ++(0.12,-0.16);
\draw (-1.5,1.3) -- ++(0.12,-0.16);
\draw[<->] (-0.64,1.82) -- ++(-0.8,-0.6);
\node at (-0.9,1.4) {\footnotesize $ x $};

%thicknesses
\draw[thick, <-] (2.6,0) -- ++(0.2,0);
\draw[thick, <-] (2.5,0) -- ++(-0.2,0) node[anchor = east] {\footnotesize $ \sim e_r $};
\draw[thick, <-] (-0.7,0.9) -- ++(0,0.2);
\draw[thick, <-] (-0.7,0.75) -- ++(0,-0.2) node[anchor = north] {\footnotesize $ \sim e_r x^{-1} $};

\end{tikzpicture}

%% file: Shell_intersection_dangerous.tex
\begin{tikzpicture}

% big Fermi ball
\draw[thick] (0,0) circle (2.5) ;
\draw[dashed, red!50!blue] (0,0) circle (2.6) ;
\fill[red!50!blue, opacity = .2] (0,2.6) arc[start angle=90, end angle=450, radius=2.6] -- (0,2.5) arc[start angle=90, end angle=-270, radius=2.5];
\fill (0,2.4) circle (0.05) node[anchor = north west] {\footnotesize $ r $};
\fill (0,0) circle (0.05) node[anchor = west] {\footnotesize $ 0 $};
\draw[thick] (0,2.4) -- (0,0);
\draw[red!50!blue] (1.1,2.3) -- ++(0.3,0.2) node[anchor = west] {\footnotesize $ A_r + r $};

% small Fermi ball
\fill[thick] (-0.3,1.5) circle (0.05) node[anchor = east] {\footnotesize $ r+r' $};
\draw[thick] (0,2.4) -- (-0.3,1.5);
\draw[thick] (-0.3,1.5) circle (1);
\draw[dashed, red!50!blue] (-0.3,1.5) circle (1.15);
\fill[red!50!blue, opacity = .2] (-0.3,2.65) arc[start angle=90, end angle=450, radius=1.15] -- (-0.3,2.5) arc[start angle=90, end angle=-270, radius=1];
\draw[red!50!blue] (-0.6,0.4) -- ++(-0.2,-0.3) node[anchor = east] {\footnotesize $ A_{r,r',x} + r $};

\end{tikzpicture}

%% file: Dangerous_belt.tex
\begin{tikzpicture}

% big Fermi ball
\draw[thick] (0,0) circle (2.5) ;
\draw[dashed, red!50!blue] (0,0) circle (2.6) ;
\fill[red!50!blue, opacity = .2] (0,2.6) arc[start angle=90, end angle=450, radius=2.6] -- (0,2.5) arc[start angle=90, end angle=-270, radius=2.5];
\fill (0,2.4) circle (0.05) node[anchor = north west] {\footnotesize $ r $};
\fill (0,0) circle (0.05) node[anchor = north west] {\footnotesize $ 0 $};
\draw[thick] (0,2.4) -- (0,0);

%dangerous belt
\draw[dashed] (0,0) circle (1.5);
\draw[->] (0,0) -- (1.2,0.9) node[anchor =west] {\footnotesize $1-x$};
\draw[dotted] (0,0) circle (1.2);
\draw[dotted] (0,0) circle (1.8);
\fill[blue, opacity = .2] (0,1.8) arc[start angle=90, end angle=450, radius=1.8] -- (0,1.2) arc[start angle=90, end angle=-270, radius=1.2];
\draw[blue, <-] (-0.8,-1) -- ++(0.1,0.4) node[anchor = south] {\footnotesize $C_{r,x}^{(\mathrm{in})} + r$};
%\node[blue] at (1,-1) {\footnotesize $A^{(\d)}_r$};

%allowed region for r+r'
\draw (0,2.4) circle (1);
\draw[green!50!black, dotted] (0,2.4) circle (0.75);
\fill[green!50!black, opacity = .2] (0,3.4) arc[start angle=90, end angle=450, radius=1] -- (0,3.15) arc[start angle=90, end angle=-270, radius=0.75];
\fill (-0.5,1.7) circle (0.05) node[anchor = east] {\footnotesize $ r+r' $};
\draw[thick] (-0.5,1.7) -- (0,2.4);
\draw[green!50!black, <-] (-0.8,2.7) -- ++(-0.4,0.1) node[anchor = east] {\footnotesize $B_{r,x} + r$};
%\node[green!50!black] at (-2.4,2.4) {\footnotesize for $ r+r' $};

%thicknesses
\draw[thin] (0,-1.5) -- ++(-1.2,0);
\draw[thin] (0,-1.8) -- ++(-1.2,0);
\draw[thick, <-] (-1,-1.5) -- ++(0,0.2);
\draw[thick, <-] (-1,-1.8) -- ++(0,-0.2) node[anchor = west] {\footnotesize $ 2 e_r^{1/2} $};

\end{tikzpicture}

%% file: Angle_alpha.tex
\def\delta{19.381} %intersection angle w. r. to large Fermi ball
\def\gamma{56.059} %intersection angle w. r. to small Fermi ball

\begin{tikzpicture}

% big Fermi ball
\fill (0,0) circle (0.05) node[anchor = west] {\footnotesize $ 0 $};
\draw (0,0) -- (0,6.2);
\draw[thick] (0,5) arc[start angle=90, end angle=60, radius=5];
\draw[thick] (0,5) arc[start angle=90, end angle=120, radius=5];
\draw[dashed, red!50!blue] (0,5.2) arc[start angle=90, end angle=60, radius=5.2];
\draw[dashed, red!50!blue] (0,5.2) arc[start angle=90, end angle=120, radius=5.2];
\fill[red!50!blue, opacity = .2] (0,5.2) arc[start angle=90, end angle=60, radius=5.2] -- ++({-sin(30)*0.2},{-cos(30)*0.2}) arc[start angle=60, end angle=120, radius=5] -- ++({-sin(30)*0.2},{cos(30)*0.2}) arc[start angle=120, end angle=90, radius=5.2] ;

%dangerous belt
\draw[dashed] (0,3) arc[start angle=90, end angle=45, radius=3];
\draw[dashed] (0,3) arc[start angle=90, end angle=135, radius=3];
\draw[dotted] (0,3.4) arc[start angle=90, end angle=45, radius=3.4];
\draw[dotted] (0,3.4) arc[start angle=90, end angle=135, radius=3.4];
\draw[dotted] (0,2.6) arc[start angle=90, end angle=45, radius=2.6];
\draw[dotted] (0,2.6) arc[start angle=90, end angle=135, radius=2.6];
\fill[blue, opacity = .2] (0,3.4) arc[start angle=90, end angle=45, radius=3.4] -- ++({-sin(45)*0.8},{-cos(45)*0.8}) arc[start angle=45, end angle=135, radius=2.6] -- ++({-sin(45)*0.8},{cos(45)*0.8}) arc[start angle=135, end angle=90, radius=3.4] ;
\draw[->] (0,0) -- ({-3*sin(45)}, {3*cos(45)}) node[anchor = north east] {\footnotesize $1-x$};

% small Fermi ball
\fill[thick] (0,3.6) circle (0.05) node[anchor = west] {\footnotesize $ r+r' $};
\draw[thick] (0,3.6) circle (2);
\draw[dashed, red!50!blue] (0,3.6) circle (2.3);
\fill[red!50!blue, opacity = .2] (0,5.9) arc[start angle=90, end angle=450, radius=2.3] -- (0,5.6) arc[start angle=90, end angle=-270, radius=2];

%geometric constructions
\coordinate (X) at ({-5*sin(\delta)}, {5*cos(\delta)});
\fill[thick] (X) circle (0.05);
\draw[thick] (0,0) -- (0,3.6);
\draw[thick] (0,3.6) -- (X);
\node at (-1.4,3.4) {\footnotesize $ 1 $};
\draw[thick] (0,0) -- (X);
\node at (-0.6,4.3) {\footnotesize $ x $};

\draw (X) -- ++({1.5*cos(\delta)}, {1.5*sin(\delta)});
\draw (X) -- ++({-1*cos(\delta)}, {-1*sin(\delta)});
\draw (X) -- ++({2*cos(\gamma)}, {2*sin(\gamma)});
\draw (X) -- ++({-1.5*cos(\gamma)}, {-1.5*sin(\gamma)});

\draw[thick] (X) -- ++({0.9*sin(\delta)}, {-0.9*cos(\delta)}) arc[start angle= {-90+\delta}, end angle= {-90+\gamma}, radius=0.9];
\node at (-1.25,4.2) {\footnotesize $ \beta $};
\draw[thick] (X) -- ++({0.9*cos(\delta)}, {0.9*sin(\delta)}) arc[start angle= {\delta}, end angle= {\gamma}, radius=0.9];
\node at (-1.1,5.1) {\footnotesize $ \beta $};

%thicknesses
\draw (0,5.2) -- ++(0.4,0);
\draw (0,5) -- ++(0.4,0);
\draw[thick, <-] (0.3,5.2) -- ++(0,0.2);
\draw[thick, <-] (0.3,5) -- ++(0,-0.2) node[anchor = north] {\footnotesize $ ~~\sim e_r $};
\draw (0,5.9) -- ++(0.5,0);
\draw (0,5.6) -- ++(0.5,0);
\draw[thick, <-] (0.4,5.9) -- ++(0,0.2) node[anchor = south] {\footnotesize $ ~~\sim e_r $};
\draw[thick, <-] (0.4,5.6) -- ++(0,-0.2);

\end{tikzpicture}

%% file: Angle_beta.tex
\def\alpha{18.296} %intersection angle w. r. to small Fermi ball
\def\Beta{34.177} %intersection angle w. r. to dangerous belt
\def\betaout{22.0154} %intersection angle w. r. to outside dangerous belt

\begin{tikzpicture}

% big Fermi ball
\fill (0,0) circle (0.05) node[anchor = west] {\footnotesize $ 0 $};
\draw (0,0) -- (0,7);
\draw[thick] (0,5) arc[start angle=90, end angle=60, radius=5];
\draw[thick] (0,5) arc[start angle=90, end angle=120, radius=5];
\draw[dashed, red!50!blue] (0,5.2) arc[start angle=90, end angle=60, radius=5.2];
\draw[dashed, red!50!blue] (0,5.2) arc[start angle=90, end angle=120, radius=5.2];
\fill[red!50!blue, opacity = .2] (0,5.2) arc[start angle=90, end angle=60, radius=5.2] -- ++({-sin(30)*0.2},{-cos(30)*0.2}) arc[start angle=60, end angle=120, radius=5] -- ++({-sin(30)*0.2},{cos(30)*0.2}) arc[start angle=120, end angle=90, radius=5.2] ;

%dangerous belt
\draw[dashed] (0,3) arc[start angle=90, end angle=45, radius=3];
\draw[dashed] (0,3) arc[start angle=90, end angle=135, radius=3];
\draw[dotted] (0,3.4) arc[start angle=90, end angle=45, radius=3.4];
\draw[dotted] (0,3.4) arc[start angle=90, end angle=135, radius=3.4];
\draw[dotted] (0,2.6) arc[start angle=90, end angle=45, radius=2.6];
\draw[dotted] (0,2.6) arc[start angle=90, end angle=135, radius=2.6];
\fill[blue, opacity = .2] (0,3.4) arc[start angle=90, end angle=45, radius=3.4] -- ++({-sin(45)*0.8},{-cos(45)*0.8}) arc[start angle=45, end angle=135, radius=2.6] -- ++({-cos(45)*0.8},{sin(45)*0.8}) arc[start angle=135, end angle=90, radius=3.4] ;
\draw[->] (0,0) -- ({-3*sin(45)}, {3*cos(45)}) node[anchor = north east] {\footnotesize $1-x$};

%outer dangerous belt
\draw[dashed] (0,7) arc[start angle=90, end angle=68, radius=7];
\draw[dashed] (0,7) arc[start angle=90, end angle=112, radius=7];
\draw[dotted] (0,7.4) arc[start angle=90, end angle=68, radius=7.4];
\draw[dotted] (0,7.4) arc[start angle=90, end angle=112, radius=7.4];
\draw[dotted] (0,6.6) arc[start angle=90, end angle=68, radius=6.6];
\draw[dotted] (0,6.6) arc[start angle=90, end angle=112, radius=6.6];
\fill[blue, opacity = .2] (0,7.4) arc[start angle=90, end angle=68, radius=7.4] -- ++({-cos(68)*0.8},{-sin(68)*0.8}) arc[start angle=68, end angle=112, radius=6.6] -- ++({-cos(68)*0.8},{sin(68)*0.8}) arc[start angle=112, end angle=90, radius=7.4] ;
%\node[blue] at (1,7.4) {\footnotesize $C^{(\mathrm{out})}_{r,x} + r$};

% ball around r
\fill[thick] (0,4.8) circle (0.05) node[anchor = east] {\footnotesize $ r $};
\draw[thick] (0,4.8) circle (2);
\draw[dotted, green!50!black] (0,4.8) circle (1.5);
\fill[green!50!black, opacity = .2] (0,6.8) arc[start angle=90, end angle=450, radius=2] -- (0,6.3) arc[start angle=90, end angle=-270, radius=1.5];

%geometric constructions
\coordinate (X) at ({-3.4*sin(\alpha)}, {3.4*cos(\alpha)});
\draw[thick] (0,4.8) -- (X);
\draw[dashed] (0,4.8) circle (1.9);
\draw[line width = 1.5, green!50!black] (0,2.9) arc[start angle=-90, end angle={-90 + \Beta}, radius=1.9];
\draw[line width = 1.5, green!50!black] (0,2.9) arc[start angle=-90, end angle={-90 - \Beta}, radius=1.9];
\draw[thick] (0,4) arc[start angle=-90, end angle={-90 - \Beta}, radius=0.8];
\node at (-0.15,4.3) {\footnotesize $ \tilde\beta $};
\node at (-1,4) {\footnotesize $ s $};

%geometric constructions for outer dangerous belt
\draw[line width = 1.5, green!50!black] (0,6.7) arc[start angle=90, end angle={90 + \betaout}, radius=1.9];
\draw[line width = 1.5, green!50!black] (0,6.7) arc[start angle=90, end angle={90 - \betaout}, radius=1.9];
\draw (0,4.8) -- ({1.9*sin(\betaout)}, {4.8+1.9*cos(\betaout)}) node[near end, below] {\footnotesize $\quad s$};
%\draw (0,5.6) arc[start angle=90, end angle={90-\betaout}, radius=0.8];
%\node at (0.2,5.9) {\footnotesize $ \tilde\beta' $};
\draw[->] (0,0) -- ({-7*sin(22)}, {7*cos(22)});
\node[anchor=east] at ({-7*sin(22)}, {7*cos(22)}) {\footnotesize $1+x$};
\node[inner sep=0pt, green!50!black] (OutLabel) at (1,7.7) {\footnotesize $(C^{(\mathrm{out})}_{r,x}\cap \partial B_s(0)) + r$};
\draw[->, green!50!black] (OutLabel) -- ({1.9*sin(10)}, {4.8 + 1.9*cos(10)});

\draw (X) -- (0,0);

\fill ({1.9*cos(68)}, {4.8-1.9*sin(68)}) circle (0.05) node[anchor = north west] {\footnotesize $ r+r' $};
%\node[inner sep=0pt, green!50!black, anchor=west] (InLabel) at ({0.2 + 1.9*cos(85)}, {3.3-1.9*sin(85)}) {\footnotesize $ (C^{\mathrm{(in)}}_{r,x} \cap \partial B_{s}(0)) + r $};
%\draw[->, green!50!black] (InLabel) ({1.9*cos(85)}, {4.8-1.9*sin(85)});
\draw[green!50!black, <-] ({1.9*cos(85)}, {4.8-1.9*sin(85)}) -- ++(0.2,-1.5) node[inner sep=0, anchor = north west] {\footnotesize $ (C^{\mathrm{(in)}}_{r,x} \cap \partial B_{s}(0)) + r $};

\end{tikzpicture}

%% file: S1_bound_1.tex
\def\kF{1.8}
\def\Delta{0.2}
\def\shift{0.15}
\begin{tikzpicture}

\fill[red!50!blue, opacity = 0.2] ({\kF},0) 
     arc(360:0:{\kF})
     -- ++({-\Delta},0)
     arc(0:360:{\kF-\Delta});
\draw[thick] (0,0) circle ({\kF});
\draw[red!50!blue, dashed] (0,0) circle ({\kF - \Delta});

\begin{scope}[shift={(\shift,0)}]
\fill[red!50!blue, opacity = 0.2] ({\kF},0) 
     arc(360:0:{\kF})
     -- ++({\Delta},0)
     arc(0:360:{\kF+\Delta});
\draw[dashed] (0,0) circle ({\kF});
\draw[red!50!blue, dashed] (0,0) circle ({\kF + \Delta});
\end{scope}

\draw[red!50!blue] ({(\kF-\Delta)*cos(35)},{(\kF+\Delta)*sin(35)}) -- ++(0.6,0.3) node[anchor = west] {\footnotesize $\{r' : x^2 - e_r \leq |r'|^2 \leq x^2\}$};
\draw[red!50!blue] ({(\kF+\Delta)*cos(10) + \shift},{(\kF+\Delta)*sin(10)}) -- ++(0.4,0.2) node[anchor = west] {\footnotesize $\{r' : x^2 \leq |r'-p| \leq x^2 + e_r\}$};

\fill[red] (0,0) circle (0.05) node[anchor = north east] {$0$};
\fill[red] ({\shift},0) circle (0.05) node[anchor = north west] {$p$};
\draw[red,->] (0,0) -- node[anchor = west] {$x$} ({(\kF)*cos(70)},{(\kF)*sin(70)});
%\node[red] at (5.8,-1) {$p=q-r$};

\end{tikzpicture}

%% file: S1_bound_2.tex
\def\kF{1.8}
\def\Delta{0.2}
\def\shift{0.6}
\begin{tikzpicture}

\fill[red!50!blue, opacity = 0.2] ({\kF},0) 
     arc(360:0:{\kF})
     -- ++({-\Delta},0)
     arc(0:360:{\kF-\Delta});
\draw[thick] (0,0) circle ({\kF});
\draw[red!50!blue, dashed] (0,0) circle ({\kF - \Delta});

\begin{scope}[shift={(\shift,0)}]
\fill[red!50!blue, opacity = 0.2] ({\kF},0) 
     arc(360:0:{\kF})
     -- ++({\Delta},0)
     arc(0:360:{\kF+\Delta});
\draw[thick] (0,0) circle ({\kF});
\draw[red!50!blue, dashed] (0,0) circle ({\kF + \Delta});
\end{scope}

\fill[red] (0,0) circle (0.05) node[anchor = north] {$0$};
\fill[red] ({\shift},0) circle (0.05) node[anchor = north] {$p$};

\pgfmathsetmacro{\dx}{\shift/2}
\pgfmathsetmacro{\dy}{sqrt(\kF*\kF - (\shift/2)*(\shift/2))}
\pgfmathsetmacro{\betaangle}{atan(\dx/\dy)}

\draw[red] (0,0) -- ({\dx},{\dy});
\draw[red] ({\shift},0) -- ({\dx},{\dy});
\draw[red] ({\dx},{\dy}) -- ++({-\dy*0.75},{\dx*0.75});
\draw[red] ({\dx},{\dy}) -- ++({-\dy*0.75},{-\dx*0.75});
\draw[red] (0,0) -- ({\shift},0);

\draw[red] ({\dx*0.4},{\dy*0.4}) arc ({270-\betaangle} : {270+\betaangle} : {\kF*0.6});

\draw[red] ({\shift/2},0) arc[start angle = 0, end angle=80.40593, radius={\shift/2}];

\node[red] at ({\shift/2},1) {$\beta$};
\node[red] at ({\shift/2+0.03},{\shift/2+0.03}) {$\alpha$}; 

\end{tikzpicture}

%% file: S1_bound_3.tex
\def\kF{1.8}
\def\Delta{0.2}
\def\shift{3.55}
\begin{tikzpicture}

\fill[red!50!blue, opacity = 0.2] ({\kF},0) 
     arc(360:0:{\kF})
     -- ++({-\Delta},0)
     arc(0:360:{\kF-\Delta});
\draw[thick] (0,0) circle ({\kF});
\draw[red!50!blue, dashed] (0,0) circle ({\kF - \Delta});

\begin{scope}[shift={(\shift,0)}]
\fill[red!50!blue, opacity = 0.2] ({\kF},0) 
     arc(360:0:{\kF})
     -- ++({\Delta},0)
     arc(0:360:{\kF+\Delta});
\draw[dashed] (0,0) circle ({\kF});
\draw[red!50!blue, dashed] (0,0) circle ({\kF + \Delta});
\end{scope}

\pgfmathsetmacro{\dx}{-(2*\kF*\Delta + \Delta*\Delta - \shift*\shift)/(2*\shift)}
\pgfmathsetmacro{\dy}{sqrt(\kF*\kF - \dx*\dx)}
\pgfmathsetmacro{\alphaangle}{atan(\dy/\dx)}

\fill[red] (0,0) circle (0.05) node[anchor = north] {$0$};
\fill[red] ({\shift},0) circle (0.05) node[anchor = north] {$p$};
\draw[red] (0,0) -- ({\dx},{\dy});
\draw[red] (0,0) -- ({\kF},0);
\draw[red] ({0.6*\kF},0) arc (0 : {\alphaangle} : {0.6*\kF});
\node[red] at (0.8,0.15) {$\alpha$};

\draw[red,<->] (2.5,{-\dy}) -- node[anchor = west] {$\gamma$} (2.5,{\dy});
\draw[red,<->] ({\shift-\kF-\Delta},1.2) -- node[anchor = south] {$\delta$} ++({\Delta},0);

\end{tikzpicture}

%% file: angle-beta-x-ge-1.tex
\def\alpha{18.296} %intersection angle w. r. to small Fermi ball
\def\Beta{34.177} %intersection angle w. r. to dangerous belt

\begin{tikzpicture}

\begin{scope}
\draw[dotted] ({5.3*cos(140)},{5.3*sin(140)}) arc[start angle = 140, end angle = 40, radius = 5.3] node[at start,anchor=north east,inner sep=0pt] {$\partial B_s(0)$};
\fill[white] (0,3.2) circle (3.4);
\clip (0,3.2) circle (3.4);
\draw[line width=1.5pt,green!50!black] (-5.3,0) arc[start angle = 180, end angle = 0, radius = 5.3];
\fill[white] (0,3.2) circle (2.6);
\clip (0,3.2) circle (2.6);
\draw[dotted] (-5.3,0) arc[start angle = 180, end angle = 0, radius = 5.3];
\end{scope}

% big Fermi ball
\fill (0,0) circle (0.08) node[anchor = west] {$ 0 $};
\draw (0,0) -- (0,6.5);
\draw[thick] (0,3) arc[start angle=90, end angle=40, radius=3];
\draw[thick] (0,3) arc[start angle=90, end angle=160, radius=3];
%\draw[dashed, red!50!blue] (0,3.4) arc[start angle=90, end angle=30, radius=3.4];
%\draw[dashed, red!50!blue] (0,3.4) arc[start angle=90, end angle=150, radius=3.4];
%\fill[red!50!blue, opacity = .2] (0,3.4) arc[start angle=90, end angle=30, radius=3.4] -- ++({-sin(60)*0.4},{-cos(60)*0.4}) arc[start angle=30, end angle=150, radius=3] -- ++({-sin(60)*0.4},{cos(60)*0.4}) arc[start angle=150, end angle=90, radius=3.4] ;
\draw (0,0) -- ({-3*sin(60)}, {3*cos(60)}) node[near end, below] {$1$};

%dangerous belt
\draw (0,6) arc[start angle=90, end angle=40, radius=6];
\draw (0,6) arc[start angle=90, end angle=140, radius=6];
%\draw[dotted] (0,6.4) arc[start angle=90, end angle=45, radius=6.4];
%\draw[dotted] (0,6.4) arc[start angle=90, end angle=135, radius=6.4];
%\draw[dotted] (0,5.6) arc[start angle=90, end angle=45, radius=5.6];
%\draw[dotted] (0,5.6) arc[start angle=90, end angle=135, radius=5.6];
%\fill[blue, opacity = .2] (0,6.4) arc[start angle=90, end angle=45, radius=6.4] -- ++({-sin(45)*0.8},{-cos(45)*0.8}) arc[start angle=45, end angle=135, radius=5.6] -- ++({-sin(45)*0.8},{cos(45)*0.8}) arc[start angle=135, end angle=90, radius=6.4] ;
\draw (0,0) -- ({-6*sin(45)}, {6*cos(45)}) node[near end,below] {$x$};

% ball around q
\fill[thick] (0,3.2) circle (0.08) node[anchor = south east] {$ q $};
\draw[dashed] (-3,3.2) arc[start angle = 180, end angle = 0, radius = 3];
%\draw[thick] (0,3.2) circle (3);
%\draw[dotted] (-2.6,3.2) arc[start angle = 180, end angle = 0, radius = 2.6];
%\draw[dotted] (-3.4,3.2) arc[start angle = 180, end angle = 0, radius = 3.4];
%\fill[green!50!black, opacity = .2] (0,6.6) arc[start angle=90, end angle=0, radius=3.4] -- (2.6,3.2) arc[start angle=0, end angle = 180, radius = 2.6] -- (-3.4,3.2) arc[start angle = 180, end angle = 90, radius=3.4]; 
\draw (0,3.2) -- ({-sin(80)*3},{3.2+cos(80)*3}) node[midway,below] {$x-1$};
%\node[anchor = south east] at (0,3.2) {$q$}

% Coordinates: (1.46361, 5.8187496), (1.9383902, 4.9671282), (??, 3.2)
% (1.46361, 5.8187496), (3.2548572, 4.182813), (??, 3.2)
%\draw[dotted] plot[smooth] coordinates {(1.46361, 5.8187496) (1.9383902, 4.9671282) (2.2, 3.2)};
%\draw[dotted] plot[smooth] coordinates {(1.46361, 5.8187496) (2.31116, 5.30696) (3.2548572, 4.182813) (3.6, 3.2)};

\fill[green!50!black, opacity=0.2] (1.46361, 5.8187496) to[curve through={(1.9383902, 4.9671282)}] (2.2, 3.2) -- (3.5, 3.2) to[curve through={(3.2548572, 4.182813) .. (2.31116, 5.30696)}] (1.46361, 5.8187496);
\draw[dotted] (1.46361, 5.8187496) to[curve through={(1.9383902, 4.9671282)}] (2.2, 3.2);
\draw[dotted] (3.5, 3.2) to[curve through={(3.2548572, 4.182813) .. (2.31116, 5.30696)}] (1.46361, 5.8187496);

\fill[green!50!black, opacity=0.2] (-1.46361, 5.8187496) to[curve through={(-1.9383902, 4.9671282)}] (-2.2, 3.2) -- (-3.5, 3.2) to[curve through={(-3.2548572, 4.182813) .. (-2.31116, 5.30696)}] (-1.46361, 5.8187496);
\draw[dotted] (-1.46361, 5.8187496) to[curve through={(-1.9383902, 4.9671282)}] (-2.2, 3.2);
\draw[dotted] (-3.5, 3.2) to[curve through={(-3.2548572, 4.182813) .. (-2.31116, 5.30696)}] (-1.46361, 5.8187496);

\node[fill=white,inner sep = 2pt] (Name) at (0,4.8) {\color{green!50!black}$\widetilde C_{q,x}^{\mathrm{(in)}} \cap \partial B_s(0) $};
\node[fill=white,inner sep = 2pt] (Cq) at (0,6.6) {\color{green!50!black}$\widetilde C_{q,x}^{\mathrm{(in)}}$};

\draw[->,green!50!black] (Name) edge ({0-sin(28)*5.3},{cos(28)*5.3});
\draw[->,green!50!black] (Name) edge ({0+sin(28)*5.3},{cos(28)*5.3});

\draw[->,green!50!black] (Cq) edge [bend right=10] (2,5.2);
\draw[->,green!50!black] (Cq) edge [bend left=10] (-2,5.2);

\draw (0,3.2) -- ({0.365734 * 5.3}, {0.937194 * 5.3}) -- (0,0);
\draw (0,3.2) -- ({0.614124 * 5.3}, {0.789210 * 5.3}) -- (0,0);

\draw (0,2.1) arc[start angle = 90, end angle = 69.5855, radius=2.1];
\draw (0,2.3) arc[start angle = 90, end angle = 52.1117, radius=2.3];

\node at (0.3,1.7) {$\widetilde\beta_<$};
\node at (0.6,2.6) {$\widetilde\beta_>$};

\node[inner sep = 2pt] (s) at (1.8,3.1) {$s$};
\draw[->] (s) -- ({0.365734 * 5.3 *0.65}, {0.937194 * 5.3 *0.65});
\draw[->] (s) -- ({0.614124 * 5.3 *0.65}, {0.789210 * 5.3 *0.65});

\node (l-small) at (0.55,4.1) {$\ell_<$};
\node (l-large) at (2,3.5) {$\ell_>$};

%\node[anchor=south west] at (3,2) {$\ell_< = x - 1 - 2\widetilde{e}_{r'}^{1/2}$};
%\node[anchor=north west] at (3,2) {$\ell_> = x - 1 + 2\widetilde{e}_{r'}^{1/2}$};
\node[anchor=west] at (3,2) {$\begin{aligned}
\ell_> & = x - 1 + 2\widetilde{e}_{r'}^{1/2}
\\
\ell_< & = x - 1 - 2\widetilde{e}_{r'}^{1/2}
\\
\widetilde{e}_{r'} & = {x^2 - s^2}
\end{aligned}$};

\end{tikzpicture}